\documentclass[10pt,sort&compress]{elsarticle}

\usepackage[utf8]{inputenc}
\usepackage{lineno}

\journal{}

\usepackage{amssymb,amsmath,amsthm,stmaryrd,mathtools}
\usepackage[setbb,longarrows,setcolon]{kmath}
\usepackage{cases}
\usepackage{xfrac}
\usepackage{booktabs}
\usepackage{multirow}

\newtheorem{definition}{Definition}
\newtheorem{theorem}{Theorem}

\newtheorem{lemma}{Lemma}

\newtheorem{corollary}{Corollary}
\newtheorem{proposition}{Proposition}

\newtheorem{example}{Example}

\renewcommand{\theconditionBis}{\the\numexpr\value{conditionBis}bis}
\renewcommand{\thetheoremBis}{\the\numexpr\value{theoremBis}bis}
\renewcommand{\thelemmaBis}{\the\numexpr\value{lemmaBis}bis}

\newcommand{\RNum}[1]{\expandafter{\romannumeral #1\relax}}

\newcommand{\HideComment}[1]{}

\usepackage{algorithmic}
\usepackage[algoruled, linesnumbered, onelanguage]{algorithm2e}
\usepackage{soul}		         

\SetCommentSty{mycommfont}

\usepackage{stmaryrd}
\SetSymbolFont{stmry}{bold}{U}{stmry}{m}{n}

\usepackage{xcolor}
\usepackage{comment}

\definecolor{myblue}{RGB}{18,75,126}
\definecolor{myorange}{RGB}{217,72,1}
\definecolor{cb6B6}{RGB}{8,81,156}
\definecolor{green_commentary}{RGB}{35,139,69}

\definecolor{emph1}{RGB}{18,75,126} 
\definecolor{emph2}{RGB}{244,109,67} 
\definecolor{emph3}{RGB}{51,160,44} 

\usepackage{pgffor}
\foreach \x in {a,...,z}{
  \expandafter\xdef\csname bf\x \endcsname{\noexpand\ensuremath{\noexpand\mathbf{\x}}}
  \expandafter\xdef\csname bs\x \endcsname{\noexpand\ensuremath{\noexpand\boldsymbol{\x}}}
}

\foreach \x in {A,...,Z}{
  \expandafter\xdef\csname bs\x \endcsname{\noexpand\ensuremath{\noexpand\boldsymbol{\x}}}
  \expandafter\xdef\csname bf\x \endcsname{\noexpand\ensuremath{\noexpand\mathbf{\x}}}
  \expandafter\xdef\csname bb\x \endcsname{\noexpand\ensuremath{\noexpand\mathbb{\x}}}
  \expandafter\xdef\csname ds\x \endcsname{\noexpand\ensuremath{\noexpand\mathds{\x}}}
  \expandafter\xdef\csname cal\x \endcsname{\noexpand\ensuremath{\noexpand\mathcal{\x}}}
}

\DeclareMathOperator{\cardOp}{card} 
\newkfunc{\card}{\cardOp}

\DeclareMathOperator{\spanOp}{span} 
\newkfunc{\spann}{\spanOp}

\newcommand{\intervint}[2]{\llbracket{#1,#2\rrbracket}}

\newkfunc{\cmf}{\varphi}
\newcommand{\varIntCmf}{u}
\newcommand{\diff}{\,\mathrm{d}}
\newcommand{\element}[1]{\kbracket{#1}}
\newcommand{\inlinevec}[1]{\ktranspose{\kparen{#1}}}

\newcommand{\Vobs}{\bsy}
\newcommand{\nbAtome}{k}
\newcommand{\coeff}{c}

\newcommand{\coeffe}{\widehat{\coeff}}
\newcommand{\coeffv}{\mathbf{\coeff}}

\newcommand{\spike}[1]{\kdelta_{#1}}
\newcommand{\spaceObs}{\calH}

\newcommand{\paramSet}{\Theta}

\newcommand{\True}{\star}

\newcommand{\minSepNonZero}{\Delta_0}

\DeclareMathOperator{\suppOp}{supp} 
\newkfunc{\supp}{\suppOp}
\newcommand{\support}{\calS}
\newcommand{\supportt}{\calS^\star}
\newcommand{\supporte}{\widehat{\calS}}

\newkfunc{\atome}{\bsa}

\newkfunc{\kernelPaper}{\kappa}

\newcommand{\scalprod}[2]{\kangle{#1, #2}}
\newcommand{\nor}[1]{\left\| #1 \right\|}

\newcommand{\dico}{\calA}
\newcommand{\dimParam}{D}

\newcommand{\param}{\theta}
\newcommand{\paramt}{\theta^\True}
\newcommand{\parame}{\widehat{\theta}}

\newkfunc{\SetAug}{\mathrm{Grid}}

\newcommand{\idxit}{t}

\newcommand{\Gtrue}{\bfG}
\newcommand{\gtrue}{\bfg}
\newcommand{\gthetaSymb}{\bfg}
\newcommand{\gtheta}[1][\param]{\gthetaSymb_{#1}}

\newcommand{\vone}{\boldsymbol{1}}

\newcommand{\cmfkernelclassP}[1][1]{\calK_{\mathtt{CMF}}(#1)}
\newcommand{\laplacekernelclassP}[1][1]{\calK_{\mathtt{Lap}}(#1)}

\newcommand{\condSumPositivKernelSymb}{\psi}
\newcommand{\condOneNegKernelSymb}{\phi}

\def\comp{OMP}
\newkfunc{\lasso}{\text{Lasso}}
\newkfunc{\blasso}{\text{BLasso}}

\newcommand{\eg}{\textit{e.g.}, }
\newcommand{\ie}{\textit{i.e.}, }

\newcommand{\BPCD}{BP for continuous dictionaries}

\usepackage{enumitem}
\usepackage{cancel}

\usepackage[hyperindex=true,
  			pdftex,
            pdfauthor={},
            pdftitle={},
            pdfsubject={},
            pdfkeywords={},
            pdfproducer={pdflatex},
            pdfcreator={}]{hyperref}

\hypersetup{
  linkcolor  = orange!65!black,
  citecolor  = myblue!100!black,
  urlcolor   = green!85!black,
  colorlinks = true
}

\numberwithin{equation}{section}
\usepackage{cleveref}
\crefname{algorithm}{algorithm}{algorithms}
\Crefname{algorithm}{Algorithm}{Algorithms}
\crefname{condition}{condition}{conditions}
\Crefname{condition}{Condition}{Conditions}
\crefname{conditionBis}{condition}{conditions}
\Crefname{conditionBis}{Condition}{Conditions}
\crefname{theoremBis}{theorem}{theorems}
\Crefname{theoremBis}{Theorem}{Theorems}
\crefname{lemmaBis}{Lemma}{lemmas}
\Crefname{lemmaBis}{Lemma}{Lemmas}
\crefname{fact}{fact}{facts}
\Crefname{appendix}{Appendix}{Appendices}

\def\boolpagebreaksection{0}
\def\booltoc{0}
\newcommand{\conditionalPagebreak}{
	\ifx\boolpagebreaksection\undefined
	\else
		\if\boolpagebreaksection1
			\pagebreak
		\fi
	\fi
}

\definecolor{CEcolor}{RGB}{0,109,44} 
\definecolor{RGcolor}{RGB}{0,0,255} 
\definecolor{CScolor}{RGB}{255, 115, 0} 

\newcommand{\say}[1]{``#1''}

\begin{document}

\hypersetup{
  linkcolor  = myorange!100!black,
  citecolor  = myblue!100!black,
  urlcolor   = myblue
}

\begin{frontmatter}



\title{When does OMP achieve exact recovery with continuous dictionaries?}


\author{
	Clément Elvira$^{1}$, 
	Rémi Gribonval$^{1}$,
	Charles Soussen$^{2}$ and
	Cédric Herzet$^{1}$
}

\address{
	$^{1}$ Univ Rennes, Inria, CNRS, IRISA\\ F-35000 Rennes, France \\
	$^{2}$ L2S, CentraleSupélec-CNRS-Université Paris-Saclay \\ 91192 Gif-sur-Yvette, France
}

\begin{abstract}
	This paper presents new theoretical results on sparse recovery guarantees for a greedy algorithm, Orthogonal Matching Pursuit (\comp{}), in the context of continuous parametric dictionaries.
	Here, the continuous setting means that the dictionary is made up of an infinite uncountable number of atoms.
	In this work, we rely on the Hilbert structure of the observation space to express our recovery results as a property of the kernel defined by the inner product between two atoms.
	Using a continuous extension of Tropp's Exact Recovery Condition, we identify key assumptions allowing to analyze OMP in the continuous setting. 
	Under these assumptions, OMP unambiguously identifies in exactly $k$ steps the atom parameters from any observed linear combination of $k$ atoms.
	These parameters play the role of the  so-called support of a sparse representation in traditional sparse recovery.
	In our paper, any kernel and set of parameters that satisfy these conditions are said to be \emph{admissible}. 

	In the one-dimensional setting, we exhibit a family of kernels relying on \emph{completely monotone functions} for which admissibility  holds for any set of atom parameters. 
	For higher dimensional parameter spaces, the analysis turns out to be more subtle.
	An additional assumption, so-called \emph{axis admissibility}, is imposed to ensure a form of delayed recovery (in at most $\nbAtome^\dimParam$ steps, where $\dimParam$ is the dimension of the parameter space). 
	Furthermore, guarantees for recovery in exactly $\nbAtome$ steps are derived under an additional algebraic condition involving a finite subset of atoms (built as an extension of the set of atoms to be recovered).
	We show that the latter technical conditions simplify in the case of Laplacian kernels, allowing us to derive simple conditions for $\nbAtome$-step exact recovery, and to carry out a coherence-based analysis in terms of a minimum separation assumption between the atoms to be recovered.
\end{abstract}


\end{frontmatter}



\conditionalPagebreak

\renewcommand\appendixname{} 
\ifx\booltoc\undefined
\else
	\if\booltoc1
		\tableofcontents
	\fi
\fi


\pagebreak

\section{Introduction}
	\label{sec:introduction}

Finding a sparse signal representation is a fundamental problem in signal processing.~It consists in decomposing a signal $\Vobs{}$ belonging to some vector space $\spaceObs{}$ as the linear combination of a few elements of some set $\dico\subseteq\spaceObs$, that is
\begin{equation}
	\label{eq:intro:observed_signals}
	\Vobs{} = \sum_{\ell=1}^{\nbAtome{}} \coeff_\ell \, \atome[\ell]\quad \mbox{where $\coeff_\ell\in\kR*$, $\atome[\ell]\in \dico$. }
\end{equation}
Sparsity refers to the fact that the number of elements involved in the decomposition \eqref{eq:intro:observed_signals} should be much smaller than the ambient dimension, \textit{i.e.}, the dimension of $\spaceObs{}$. 
The set $\dico$ is commonly referred to as a \textit{dictionary} and its elements as \textit{atoms}. 
In the sequel, we will assume that $\dico$ is a parametric dictionary defined as:
\begin{equation}
	\label{eq:intro:dictionary}
	\dico = \kset{ \atome(\param) }{\param\in\paramSet}
\end{equation}
where $\paramSet= \kR^\dimParam$ and $\atome : \paramSet \rightarrow \spaceObs$ is some continuous and injective function.
In this setup, \eqref{eq:intro:observed_signals} implies that there exist $\nbAtome$ parameters $\{\paramt_\ell\}_{\ell=1}^\nbAtome$  such that $\Vobs$ can be expressed as a linear combination of the atoms $\{\atome(\paramt_\ell)\}_{\ell=1}^\nbAtome$.

Over the past decade, sparse representations have proven to be of great interest in many applicative domains.
As a consequence, numerous practical procedures, along with their theoretical analyses, have been proposed in the literature.
Most contributions addressed the sparse-representation problem in the \say{\textit{discrete}}  setting, where the dictionary contains a \textit{finite} number of elements, see~\cite{Foucart2013}.
Recently, several works tackled the problem of sparse representations in \say{\textit{continuous}} dictionaries, where $\dico$ is made up of an infinite \textit{uncountable} number of atoms but $\atome : \paramSet \rightarrow \spaceObs$ enjoys some continuity property, see \eg \cite{Ekanadham2011,candes2014,Duval2014}.
We review the contributions most related to the present work in \Cref{sec:state_of_the_art}.

Before dwelling over the state of the art, we briefly describe the scope of our paper. In this work, we focus on the continuous setting and assume that $\spaceObs{}$  is a Hilbert space with inner product $\scalprod{\cdot}{\cdot}$ and induced norm $\nor{\cdot}$. 
We derive exact recovery conditions for \say{Orthogonal Matching Pursuit} (\comp{})~\cite{pati1993}, a natural  adaptation to the continuous setting of a popular greedy procedure of the literature (see \Cref{alg:continuousOMP}).
The main question addressed in this paper is as follows.  
Let $\{\paramt_\ell\}_{\ell=1}^{\nbAtome}$ be $\nbAtome$ pairwise distinct elements of $\paramSet$ and assume that $\Vobs$ obeys~\eqref{eq:intro:observed_signals} with $\atome[\ell] = \atome(\paramt_\ell)$ for some $\{\coeff_\ell\}_{\ell=1}^{\nbAtome}\subset\kR*$.
Under which conditions does OMP achieve exact recovery (that is, correct unambiguous identification) of the parameters $\{\paramt_\ell\}_{\ell=1}^{\nbAtome}$ and the coefficients $\{\coeff_\ell\}_{\ell=1}^{\nbAtome}$? In particular, is exact recovery possible in $\nbAtome$ steps? 
This is of course only possible if the preimage $\paramt_\ell$ of an atom $\atome[\ell]= \atome(\paramt_\ell)$ is unique, hence the assumption that $\atome(\cdot)$ is injective.

\begin{algorithm}[t]
	\caption[A1]{
	\label{alg:continuousOMP}
	Orthogonal Matching Pursuit (\comp)
	}
	\KwIn{observation $\Vobs\in\spaceObs$, normalized dictionary $\dico = \kset{ \atome(\param) }{\param\in\paramSet}$.}
	
	$\bsr \leftarrow \Vobs$      \tcp{residual vector}
	$\supporte \leftarrow \emptyset$ \tcp{estimated support}
	$t\leftarrow 0$ \;
	\While{$\bsr\neq{\bf0}_{\spaceObs{}}$} {
		$t \leftarrow t+1$ \;
		$\parame_t \in \kargmax_{\param\in\paramSet}\, \kvbar{\kangle{\atome(\param),\bsr} } $ \label{line:algo:continuousOMP:findtheta} 
		\tcp{atom selection}
		$\supporte \leftarrow \supporte \cup \{ \parame_t \}$ 
		\tcp{support update}
		$\displaystyle \kparen{\coeffe_1,\dotsc,\coeffe_t} \leftarrow \kargmin_{\kparen{\coeff_1,\dotsc,\coeff_t}\in\kR^t}\, 
		\kvvbar{ \Vobs - \sum_{\ell=1}^{t} \coeff_\ell \, \atome(\parame_\ell) }$
		\label{line:algo:continuousOMP:LS} 
		\tcp{least-squares update} 
		$\displaystyle \bsr \leftarrow \Vobs - \sum_{\ell=1}^{t} \coeffe_\ell\,\atome(\parame_\ell)$ 
		\label{line:algo:continuousOMP:update_residual}
		\tcp{residual vector}	
	}
	$\widehat{\nbAtome}=t$ \label{line:algo:continuousOMP:def_khat} \;
	\KwOut{estimated support $\supporte=\{\parame_\ell\}_{\ell=1}^{\widehat{\nbAtome}}$ and coefficients $\{\coeffe_\ell\}_{\ell=1}^{\widehat{\nbAtome}}$.
	}
\end{algorithm}

We note that, in the context of continuous dictionaries, the fact that \comp{} could correctly identify a set of $\nbAtome$ atoms in exactly $\nbAtome$ iterations may seem surprising in itself. 
Indeed, inspecting \Cref{alg:continuousOMP}, we see that this implies that \comp{} must identify one correct atom at \textit{each} iteration $t$ of the algorithm, that is $\parame_\idxit\in \{\paramt_\ell\}_{\ell=1}^{\nbAtome}$ $\forall \idxit\in\intervint{1}{\nbAtome}$.
 The following simple example suggests that such a requirement may never be met for continuous dictionaries:

\begin{example}[The Gaussian deconvolution problem]
	\label{example:intro:deconvolution}
	Consider $\paramSet=\kR$ and let $\spaceObs = L^2(\kR)$ be the space of square integrable functions on $\kR$. Assume $\atome(\cdot)$ is defined as
	\begin{equation}
		\label{eq:intro:gaussian_atom}
		\begin{split}
		\atome : \kR \;\longrightarrow\;& L^2(\kR) \\
			\param \;\longmapsto\;& \cstpi^{-\frac{1}{4}} \cste^{-\frac{1}{2}(\cdot - \param)^2} 
			.
		\end{split}
	\end{equation}

Suppose $\Vobs{}$ results from the positive linear combination of $\nbAtome = 2$ distinct atoms, that is $\Vobs = \coeff_1 \atome(\paramt_1)+\coeff_2 \atome(\paramt_2)$, $\paramt_1\neq\paramt_2$, $\coeff_1>0$, $\coeff_2>0$. 
Then, even in this very simple case, \comp{} never selects an atom in $\kbrace{\paramt_1,\paramt_2}$ at the first iteration.
Indeed, particularizing step \ref{line:algo:continuousOMP:findtheta} of \Cref{alg:continuousOMP} to the present setup, we have that, at the first iteration, \comp{} will select the parameter $\param$ maximizing
\begin{equation}
	\label{eq:intro:gaussian_sp}
	\kvbar{\scalprod{\atome(\param)}{\Vobs}}
	= 
	\coeff_1 \cste^{-\frac{1}{4}(\param - \paramt_1)^2}+ \coeff_2 \cste^{-\frac{1}{4}(\param - \paramt_2)^2}. 
\end{equation}
Now, since the right-hand side of \eqref{eq:intro:gaussian_sp} is continuously differentiable, first-order optimality conditions tell us that any maximizer of $\param\mapsto\kvbar{\scalprod{\atome(\param)}{\Vobs}}$ must satisfy
\begin{align}
	(\param - \paramt_1) \coeff_1 \cste^{-\frac{1}{4}(\param - \paramt_1)^2} + (\param - \paramt_2) \coeff_2 \cste^{-\frac{1}{4}(\param - \paramt_2)^2} = 0. 
\end{align}
Since $\paramt_1\neq \paramt_2$, $\coeff_1\neq0$, $\coeff_2\neq0$, this equality cannot be verified by either $\param=\paramt_1$ or $\param=\paramt_2$. As a consequence, OMP necessarily selects some  $\param\notin\{\paramt_1,\paramt_2\}$.  \qed
\end{example}

Nevertheless, we show in this paper that exact recovery in $\nbAtome$ steps is possible with OMP for some \emph{particular} families of dictionaries $\dico$.
Our recovery conditions are expressed in terms of the \emph{kernel function} $\kernelPaper\kparen{\param, \param'}$ associated to the inner product between two atoms, \ie 
\begin{equation}
	\label{eq:intro:hyp_scalar_product}
	\kernelPaper\kparen{\param, \param'} \triangleq \scalprod{\atome( \param )}{\atome( \param' )}.
\end{equation}
We show that if the kernel $\kernelPaper\big(\param, \param'\big)$ and the atom parameters $\{\paramt_{\ell}\}_{\ell=1}^{\nbAtome}$ verify some particular conditions (see \Cref{subsec:contrib:abstract_conditions}), then exact recovery in $\nbAtome$ steps is possible with OMP. 
We emphasize moreover that these conditions are satisfied for a family of kernels of the form:
\begin{align}
\label{eq:intro:cmf_kernel}
\kernelPaper(\param, \param') & 
		= \cmf{}\kparen{ \kvvbar{\param -  \param'}_p^p } \quad  \mbox{ $0 < p \leq 1$},
\end{align}
where $\kvvbar{\cdot}_p$ is the $\ell_p$ quasi-norm (norm for $p=1$) and $\cmf{}$ is a \textit{completely monotone} function (see \Cref{def:intro:def_cmf}). 
This family encompasses the well-known Laplace kernel \citep{miller2001}. 
Hereafter, we will refer to kernels taking  the form \eqref{eq:intro:cmf_kernel} as \say{CMF kernels}.

A first (perhaps surprising) outcome of our analysis is as follows.
If $\paramSet=\kR$ and the dictionary is defined by a CMF kernel \eqref{eq:intro:cmf_kernel}, 
OMP correctly identifies \emph{any} pairwise distinct atom parameters $\{\paramt_\ell\}_{\ell=1}^{\nbAtome}\subset\paramSet$  and coefficients $\{\coeff_\ell\}_{\ell=1}^{\nbAtome}\subset\kR*$ in exactly $\nbAtome$ iterations for \emph{any} $\nbAtome\in\kN$ (see \Cref{th:contrib:cmf_uniformRecov_1D}). 
We emphasize that no separation (\emph{i.e.}, minimal distance between parameters $\{\paramt_\ell\}_{\ell=1}^{\nbAtome}$) is needed. 
To our knowledge, this is the first recovery of this kind in continuous dictionaries when no sign constraint is imposed on the coefficients.
It turns out that this ``universal'' exact recovery result is valid for very particular families of dictionaries: 
CMF kernels
 exhibit a discontinuity in their derivatives (\textit{e.g.}, the partial derivative of $\kappa$ with respect to $\theta$ when $\theta=\theta'$) and the space $\spaceObs$ in which the corresponding dictionary lives is necessarily infinite-dimensional (see \Cref{sec:CMF_dico}).

When $\paramSet= \kR^\dimParam$ with $\dimParam>1$ and the dictionary is defined by a CMF kernel \eqref{eq:intro:cmf_kernel}, we show that such an exact recovery result no longer holds (see \Cref{ex:intro:erc_required_dim2}).
Nevertheless, for dictionaries based on CMF kernels, under an additional hypothesis (referred to as ``\emph{axis admissibility}'', see~\Cref{def:contrib:grille_admissibility}), we demonstrate that 
 a form of {\em delayed} exact recovery (that is, in more than $\nbAtome$ iterations) holds.
The number of iterations sufficient to identify a set of $\nbAtome$ parameters is then upper-bounded by $\nbAtome^\dimParam$ (see \Cref{th:contrib:dimD:grid_recovery}).
Moreover, under the above-mentioned hypothesis of axis admissibility, sufficient and necessary conditions for exact recovery of a given subset\footnote{Here and in the sequel, when referring to a subset $\{\paramt_\ell\}_{\ell=1}^{\nbAtome}$, we implicitly assume that the elements $\paramt_{\ell}$ are pairwise distinct.} $\{\paramt_\ell\}_{\ell=1}^{\nbAtome}$ in $\nbAtome$ steps (irrespective of the choice of the coefficients $\{\coeff_\ell\}_{\ell=1}^{\nbAtome}$) can be written in terms of a \textit{finite} number of atoms of the dictionary (smaller than $\nbAtome^\dimParam$) including $\{\paramt_\ell\}_{\ell=1}^\nbAtome$ (see \Cref{th:contrib:dimD:grid_recovery}). 
We leverage this result to prove that exact recovery in $\nbAtome$ steps is possible as soon as the elements of the subset $\{\paramt_\ell\}_{\ell=1}^{\nbAtome}$ obey some \say{minimum separation} condition (see \Cref{prop:intro:erc_ell1:separation}).

The rest of this paper is organized as follows. 
\Cref{sec:state_of_the_art} draws connections with the sparse recovery literature.
\Cref{subsec:recovery} elaborates on the main ingredients of the \say{continuous} setup and defines the notions of recovery that are used in the statements of our results.
In \Cref{subsec:contrib:abstract_conditions}, we exhibit a sufficient condition on atom parameters and kernel such that exact recovery of a \emph{given} set of atom parameters holds. 
We then present the family of CMF dictionaries in \Cref{sec:CMF_dico} and show in \Cref{subsec:contrib:ksparserecov} that different forms of recovery can be achieved in these dictionaries. 
Concluding remarks are given in \Cref{sec:discussion}. 
The technical details of our results are contained in the appendices of the paper.
The proofs of our main recovery results are exposed in Appendices~\ref{sec:technical_details} and \ref{sec:app:proofs-CMF}.
Appendix~\ref{sec:app:misc} contains some auxiliary technical details.
Finally, Appendices~\ref{sec:other:calcul_separation} and \ref{sec:app:dimD:uniformrecovery_k_2} are dedicated to some mathematical developments related to two examples discussed in the paper.

\conditionalPagebreak


\section*{Notations}
The following notations will be used in this paper.
The symbols $\kR, \kR*, \kR+, \kR*+$ refer to the set of real, non-zero, non-negative and positive numbers, respectively. 
Boldface lower and upper cases  (\eg $\bfg$, $\bfG$) are used to denote (finite-dimensional) vectors and matrices, respectively.
The notation $\element{i}$ refers to the $i$th element of a vector, and $\element{i,j}$ for the element at the $i$-th row and $j$-th column of a matrix.
Italic boldface letters (\eg $\Vobs$ or $\atome$) denote elements of a Hilbert space $\spaceObs$.
All-one and all-zero column vectors in $\kR^{\nbAtome}$ are denoted ${\bf1}_{\nbAtome}$ and ${\bf0}_{\nbAtome}$, respectively. The $\ell$-th vector of the canonical basis in $\kR^\dimParam$ will be denoted $\bfe_{\ell}$.  
The notations $\scalprod{\cdot}{\cdot}$ and $\nor{\cdot}$ refer to the inner product and its induced norm on $\spaceObs$, while $\kvvbar{\cdot}_p$ with $p>0$ refers to the classical $\ell_p$ (pseudo-) norm on $\kR^\dimParam$.
Finally, calligraphic letters (\eg $\support, \calG$) are used to describe finite subsets of the parameter space $\paramSet{}$, while $\intervint{m}{n}$ denotes the set of integers $i$ such that $m\leq i \leq n$. Given $\support\subseteq\paramSet$, we let $\support^c \triangleq \paramSet\backslash\support$ be the complementary set of $\support$ in $\paramSet$. 
The cardinality of a set is denoted $\card(\cdot)$. 
Finally, if $\varphi: \kR \mapsto \kR$ is a function, the notation $\varphi^{(n)}$ refers to its $n$-th derivative.
The main notations used in this paper are summarized in Appendix~\ref{sec:table_of_notations}.

\section{Related works and state of the art}
	\label{sec:state_of_the_art}

Over the last decade, sparse representations have sparked a surge of interest in the signal processing, statistics, and machine learning communities.
A question of broad interest which has been addressed by many scientists is the identification of the \say{sparsest} representation of an input signal $\Vobs$ (that is, the representation involving the smallest number of elements of $\dico$).
Since this problem has been shown to be NP-hard~\cite{Natarajan1995}, many sub-optimal procedures have been proposed to approximate its solution.\footnote{The term \say{sub-optimal} has to be understood in the following sense:  these procedures are heuristics that only find the sparsest solution of the input vector $\Vobs$ under some restricted conditions. They can fail when these conditions do not hold.} 
Among the most popular, one can mention methodologies based on \emph{convex relaxation} and \emph{greedy algorithms}. 
  
Greedy procedures have a long history in the signal processing and statistical literature, which can be traced back to (at least) the 60's~\cite{Miller2002}.~In the signal processing community, the most popular instances of greedy algorithms  are known under the names of \emph{Matching Pursuit} (MP)~\cite{Mallat1993}, \emph{Orthogonal Matching Pursuit} (OMP)~\cite{pati1993} (also known as \emph{Orthogonal Greedy Algorithm (OGA)}~\cite{DeVore1996,Liu2012}) and \emph{Orthogonal Least Squares} (OLS)~\cite{soton1989}.
Although these algorithms were already known under different names in other communities~\cite{Friedman1981}, they have been \say{rediscovered} many times, see \eg \cite{Huber1985,Neira2002,Temlyakov2008}.
Extensions to more general cost functions and kernel dictionaries are discussed in~\cite{Vincent2002}.

Sparse representations based on the resolution of convex optimization problems were initially proposed in geophysics~\cite{Claerbout1973} for seismic exploration.~These methods have been popularized in the signal processing community by the seminal work by Chen~\emph{et al.}~\cite{Chen1998} and by Tibshirani in Statistics~\cite{tibshirani1996}.
Well-known instances of convex-relaxation approaches for sparse representations are \emph{Basis Pursuit} (BP)~\cite{Chen1998} and \emph{\lasso}~\cite{tibshirani1996}, also known as \emph{Basis Pursuit Denoising}, which correspond to different convex optimization formulations.
Many algorithmic solutions to efficiently address these problems have been proposed, see \eg \cite{Tibshirani2004,Beck2009,Boyd2010}. 

All the early contributions mentioned above have been made in the \emph{discrete} setting, where the dictionary contains a \emph{finite} number of atoms.
Although Mallat and Zhang~\cite{Mallat1993} already defined MP for continuous dictionaries, the wide practice of MP is in the discrete setting.
Greedy sparse approximation in the context of dictionaries made up of an infinite (possibly uncountable) number of atoms has only been studied more recently~\cite{Temlyakov2008,Gribonval2008,Borup2008}.
Practical procedures to implement greedy procedures in continuous dictionaries can be found in~\cite{Knudson2014,eftekhari2015greed,dorffer2018}.

On the side of convex relaxation approaches, it was shown that a \emph{continuous} version of \lasso{} can be expressed as a convex optimization problem over the space of Radon measures~\cite{Bredies2012} and later referred to as the \emph{Beurling Lasso} (\blasso)~\cite{castro2012}.~A continuous version of BP was also proposed~\cite{candes2014} for specific continuous dictionaries by exploiting similar ingredients. 
Motivated by an increasing demand in efficient solvers, different strategies to find the solution of this problem (to some accuracy) were proposed over the past few years.
When dealing with dictionaries made up of complex exponentials that depend on a one-dimensional parameter (that is $\dimParam=1$), the Blasso problem can be reformulated as a semidefinite program (SDP)~\cite{candes2014,Tang2013}.  
These methods have been further extended to the multidimensional case by considering SDP approximations of the problem~\cite{DeCastro2017}.
The conditional gradient method (CGM) has also proven to be applicable to address the BLasso problem~\cite{Bredies2012} and further enhanced with non-convex local optimization extra steps~\cite{Boyd2015,2019-Catala-LowRank,Denoyelle2019}.
Interestingly, the CGM has been shown to be equivalent to the so-called exchange method in~\cite{Eftekhari2019,Flinth2019arXiv}.
More recently, gradient-flow methods on spaces of measures have also been investigated to address the BLasso problem~\cite{Chizat2018aa,chizat:hal-02190822}. 	
Finally, we also mention the existence of a vast literature on non-convex and non-variational procedures leveraging the celebrated Prony's method~\cite{de1795essai}.
Among others, one may cite its extension to the multivariate case~\cite{Kunis2016}, the MUSIC~\cite{Liao2016} and ESPRIT~\cite{Roy1989} frameworks, as well as finite rate of innovation methods~\cite{Wei2016}.

Because (most of) the approaches mentioned above (both in the discrete and continuous settings) are heuristics looking for the sparsest representation of some $\Vobs$, many theoretical works have been carried out to analyze their performance.  
Hereafter, we review the contributions of the literature most related to the present work. 
In particular, we focus on the contributions dealing with exact recovery of some subset $\{\paramt_\ell\}_{\ell=1}^{\nbAtome}$ for any choice of the coefficients $\{\coeff_\ell\}_{\ell=1}^{\nbAtome}$ (sometimes assuming some specific sign patterns). 
In our discussion, we will use the short-hand notation $\supportt = \{\paramt_\ell\}_{\ell=1}^{\nbAtome}$ and refer to the latter as \textit{``support''}. Since we always implicitly assume that the parameters $\{\paramt_\ell\}_{\ell=1}^{\nbAtome}$ are pairwise distinct, we have $\card(\supportt) = \nbAtome$. 
The presentation is organized in two parts, dealing respectively with the discrete and the continuous cases. 
In the discrete setting, we restrict our attention to contributions addressing the performance of MP, OMP and OLS, \ie the greedy procedures the  most connected to the framework of this paper. 
In the continuous setting, recovery analysis, including stability and robustness to noise, have only been addressed for convex-relaxation approaches.
We review these conditions below and draw some similarities and differences with the guarantees derived for \comp{}.

\subsection{Discrete setting}
	\label{subsec:stat_of_the_art:discreteSetting}

The discrete setting refers to the case where the dictionary contains a finite number of elements, that is $\card(\dico)<\infty$.
Hereafter, we will restrict our discussion to parametric dictionaries of the form \eqref{eq:intro:dictionary} since they are the main focus of this paper. In this context, the discrete setting refers to $\card(\paramSet) < +\infty$.

\paragraph{Exact Recovery Condition} 
The first thorough analysis of OMP exact ``$\nbAtome$-step'' recovery of some $\supportt\triangleq\{\paramt_\ell\}_{\ell=1}^{\nbAtome}$ is due to Tropp in \cite{tropp2004}. 
Introducing the notations
	\begin{align}
	\begin{array}{rl}
		\bfG\element{\ell,\ell'} &\triangleq\; \kernelPaper\kparen{\paramt_\ell, \paramt_{\ell'}},  \\
		\bfg_{\param}\element{\ell} &\triangleq\; \kernelPaper\kparen{\param, \paramt_\ell}
		,
	\end{array}
\end{align}
Tropp's result can be rephrased as follows:
\begin{theorem}[Tropp's ERC] 
	\label{th:Tropp's-ERC}
	Consider $\supportt=\{\paramt_\ell\}_{\ell=1}^{\nbAtome}$ and assume that the atoms $\{\atome(\paramt_\ell)\}_{\ell=1}^{\nbAtome}$ are linearly independent. If 
	\begin{equation}
		\label{eq:ERCTropp}
		\stepcounter{equation}
		\tag{$\theequation-\mathrm{ERC}$}
		\forall {\param\in \paramSet\setminus\support^\True}, \quad \kvvbar{ \kinv{\bfG} \bfg_{\param}}_1<1
		,
	\end{equation}
	then OMP with $\Vobs=\sum_{\ell=1}^{\nbAtome} \coeff_\ell\, \atome(\paramt_\ell)$ as input unambiguously identifies $\supportt$ and $\{\coeff_\ell\}_{\ell=1}^{\nbAtome}$ in $\nbAtome$ iterations for any choice of the coefficients $\{\coeff_\ell\}_{\ell=1}^{\nbAtome}\subset\kR*$.
	Conversely, if~\eqref{eq:ERCTropp} is not satisfied, there exist not all-zero coefficients $\{\coeff_\ell\}_{\ell=1}^{\nbAtome}$ such that OMP with $\Vobs=\sum_{\ell=1}^{\nbAtome} \coeff_\ell\, \atome(\paramt_\ell)$ as input selects some $\param\notin\supportt$ at the first iteration. 
\end{theorem}
\noindent
A proof of the direct part of this result can be found in \cite[Th.~3.1]{tropp2004}.
The converse part is a slight variation of Tropp's original statement \cite[Th.~3.10]{tropp2004} and a proof can be found in \eg~\cite[Prop.~3.15]{Foucart2013}.

Condition \eqref{eq:ERCTropp} is usually referred to as the ``Exact Recovery Condition'' in the literature, and simply denoted ERC. Assuming linear independence of the atoms $\{\atome(\paramt_\ell)\}_{\ell=1}^{\nbAtome}$, it can be reformulated in the following (and perhaps more interpretable) way: 
\begin{equation}
	\label{eq:ERC-no-bad-selection}
	\forall\bsr\in \calR_{\supportt}\setminus\{{\bf0_{\spaceObs}}\},\; \forall {\param\in \paramSet\setminus\support^\True},
	\quad
	\kvbar{\kangle{\atome(\param),\bsr} } < \max_{\param'\in\supportt}\kvbar{\kangle{\atome(\param'),\bsr} }
\end{equation}
where $\calR_{\supportt}\triangleq\mathrm{span}(\{\atome(\paramt_\ell)\}_{\ell=1}^{\nbAtome})$. 
In other words, it implies that OMP \emph{always} selects a parameter in $\supportt$ during the first $\nbAtome$ iterations for any input vector $\Vobs$ resulting from the linear combination of the atoms $\{\atome(\paramt_\ell)\}_{\ell=1}^{\nbAtome}$. 
The converse part shows that \eqref{eq:ERCTropp} is worst-case necessary in the following sense: if \eqref{eq:ERCTropp} is not satisfied, there exists (at least) one non-trivial linear combination of the atoms $\{\atome(\paramt_\ell)\}_{\ell=1}^{\nbAtome}$ such that OMP selects an element $\param\notin\supportt$ at the first iteration; in this case the correct identification of $\supportt$ in $\nbAtome$ iterations is obviously not possible.

Interestingly, condition~\eqref{eq:ERCTropp} is also related (along with the linear independence of the atoms $\{\atome(\paramt_\ell)\}_{\ell=1}^{\nbAtome}$) to the success of MP, OLS and some convex relaxation procedures.
In~\cite[Th.~2]{Soussen2013}, the authors showed that \eqref{eq:ERCTropp} is also necessary and sufficient for exact $\nbAtome$-step recovery of $\support^\True$ by OLS. 
Regarding MP, \eqref{eq:ERCTropp} ensures that the procedure only selects atoms in $\supportt$ but does not imply exact recovery after $\nbAtome$ iterations of the algorithm since the same atom can be selected many times (the least-squares update of the coefficients in Algorithm 1 is not carried out), see \eg \cite[{Th.~1}]{Gribonval2006}. 
Finally, in~\cite[Th.~3]{Fuchs2004}~\cite[Th.~8]{Tropp2006}, the authors show that \eqref{eq:ERCTropp} also ensures correct identification of $\supportt$ by some convex relaxation procedures as \eg BP or Lasso.

\paragraph{Coherence}

Tropp's condition is of limited practical interest to characterize the recovery of all supports of size $k$ since it requires to verify that $\eqref{eq:ERCTropp}$ holds for any $\support^\True$ with $\card(\support^\True)=\nbAtome$. In order to circumvent this issue, other sufficient conditions of success, weaker but easier to evaluate in practice, have been proposed in the literature. 
One of the most popular conditions is based on the \emph{coherence} $\mu$ of a normalized dictionary. Assuming the atoms of the dictionary are of unit norm, this condition writes (with the convention that $\tfrac{1}{0}=+\infty$):
\begin{equation}
	\label{eq:coherenceCondition}
	\nbAtome{} < \frac{1}{2}\kparen{1 + \frac{1}{\mu}}
\end{equation}
where
\begin{equation}
	\label{eq:def-coherence}
	\mu \;\triangleq \sup_{\substack{\param,\param'\in\paramSet\\ 
	\param\neq\param'}}\, \kvbar{\kernelPaper\kparen{\param, \param'}}.
\end{equation}
Condition~\eqref{eq:coherenceCondition}, together with the normalization of the dictionary, implies that \eqref{eq:ERCTropp} is verified for any $\support^\True$ with $\card(\support^\True)\leq \nbAtome$, and also implies the linear independence of any group of $\nbAtome$ atoms of the dictionary. 
It therefore implies that OMP and OLS correctly identify any $\supportt$ with $\card(\support^\True)\leq \nbAtome$ in exactly $\card(\support^\True)$ iterations. 
It also ensures the correct identification of any $\supportt$ with $\card(\support^\True)\leq \nbAtome$ by BP and Lasso. 
In \cite{Herzet2016}, the authors emphasize that condition \eqref{eq:coherenceCondition} can be slightly relaxed if the coefficients $\{\coeff_{\ell}\}_{\ell=1}^{\nbAtome}$ exhibit some decay.

The coherence of the dictionary can be seen as a particular measure of \say{proximity}\footnote{$\mu=0$ if all the atoms are pairwise orthogonal and $\mu\simeq 1$ if some atoms are very correlated.} between the atoms of the dictionary.
Other exact recovery conditions, based on different proximity measures, have been proposed in the literature.
In \cite[Th.~3.5]{tropp2004}, the author derived recovery conditions based on \say{cumulative coherence}, whereas in \cite{Huang2011,Liu2012,Maleh2011,Mo2012,Wang2012,Chang2014,Wen2013,Mo2015}, guarantees based on \say{restricted isometry constants} were proposed. Given that Tropp's condition is both necessary and sufficient, all such recovery conditions imply that the ERC holds for any support of size $\nbAtome$.

\subsection{Continuous setting}
	\label{subsec:continuous-setting}

\paragraph{General setup}
Sparse representations in continuous dictionaries are basically characterized by two main ingredients: 
\begin{enumerate}[label=\itshape\roman*)]
	\item a parameter set $\paramSet$, usually assumed to be a connected subset of $\kR^{\dimParam}$ with non-empty interior\footnote{These two assumptions implies that $\paramSet$ is uncountable~\cite[Ch.~1, Exercise 19d]{tinyRudin}.
	}
	or a torus in dimension $\dimParam$.
	We note that, in this paper, we restrict our attention to the case where $\paramSet{}=\kR^\dimParam$ for $\dimParam\geq1$.
	\item an \say{atom} function $\atome:\paramSet\rightarrow\spaceObs$, assumed to be continuous and injective.
\end{enumerate} 
This type of dictionary appears in numerous signal processing tasks such as sparse spike deconvolution or super-resolution where one aims to recover fine-scale details from an under-resolved input signal~\cite{Claerbout1973,candes2014,Denoyelle2019,DiCarlo2020}.

\paragraph{Irrelevance of existing analyses}
The continuity of $\atome(\cdot)$ does not allow most of the analyses performed in the context of discrete dictionaries to be extended to the continuous framework. 
In particular, all exact recovery conditions based on coherence or restricted isometry constants turn out to be violated whenever dealing with continuous dictionaries.~As for the coherence condition \eqref{eq:coherenceCondition}, it is easy to see that the continuity of $\atome(\cdot)$ implies the continuity of $\kernelPaper(\cdot,\cdot)$ with respect to both its arguments. This, in turn, implies that $\mu=1$ (for normalized atoms) and the coherence-based condition \eqref{eq:coherenceCondition} is never met, even for $\nbAtome=1$!  

In order to circumvent this issue, some specific exact recovery conditions for continuous dictionaries have been proposed in the literature, see \eg \cite{castro2012,candes2014,Duval2014,Azais2015}.
We review below the main ingredients grounding these conditions of recovery.
In the context of convex-relaxation approaches, these conditions originate from the analysis of the associated optimality conditions.

\paragraph{A separation condition for \BPCD}
In the context of BP for continuous dictionaries, the question of exact recovery can be rephrased as follows: if the  atoms $\{\atome(\paramt_\ell)\}_{\ell=1}^{\nbAtome}$ are linearly independent and $\Vobs=\sum_{\ell=1}^{\nbAtome} \coeff_\ell\, \atome(\paramt_\ell)$, is the solution of \BPCD{} unique and equal to a discrete measure supported on $\{\paramt_\ell\}_{\ell=1}^{\nbAtome}$ with weights $\{\coeff_\ell\}_{\ell=1}^{\nbAtome}$? \\
The case where each atom is a collection of Fourier coefficients ($\spaceObs=\kC^m$) has received a lot of attention due to its connection with the super-resolution problem.
Indeed, the latter scenario is equivalent to recovering infinitely resolved details (the parameters) from some  low-pass observation. 
Without further assumptions on the coefficients, the targeted measure is the unique solution of \BPCD{}  provided that~\cite[Th.~1.2 and~1.3]{candes2014}
\begin{equation}
	\label{eq:state_of_the_art:separation_BP}
	\min_{\substack{\ell,\ell'\in\intervint{1}{\nbAtome}\\
	\ell\neq\ell'}} \kvbar{\paramt_{\ell'} - \paramt_{\ell}} > \frac{C}{f_c},
\end{equation}
where $\kvbar{\cdot-\cdot}$ is the $\ell_\infty$ distance on the $D$-dimensional torus (maximum deviation in any coordinate), $C$ is a constant that depends on the parameter dimension $\dimParam$ and $f_c$ is the cut-off frequency of the observation low-pass filter.
A framing of the value of $C$ has been proposed by the same authors, further refined in~\cite[Cor.~1]{Duval2014} and \cite[Th.~2.2]{FernandezGranda2016}. 

We see that \eqref{eq:state_of_the_art:separation_BP} implies the recovery of $\{\paramt_\ell\}_{\ell=1}^{\nbAtome}$ and $\{\coeff_\ell\}_{\ell=1}^{\nbAtome}$ provided that the elements of $\{\paramt_\ell\}_{\ell=1}^{\nbAtome}$ verify some ``minimum separation'' condition. Interestingly, as shown in \cite[Th.~2.1]{castro2012}, this separation condition is no longer needed when dealing with \textit{positive} linear combination of atoms (that is, when all the coefficients $\{\coeff_\ell\}_{\ell=1}^\nbAtome$ are positive).
The authors showed moreover that this separation-free result for positive linear combinations holds for any dictionary such that the atom function $\atome$ forms a ``Chebyshev system''~\cite[Ch.~2]{Kren1977} and provided that $2\nbAtome+1$ observations are available (\ie $\dim(\spaceObs) \geq 2\nbAtome+1$).

\paragraph{Dual certificates for the \blasso{} problem}
In~\cite{Duval2014}, the authors derived several dual certificates for the \blasso{} problem generalizing the work done by Fuchs for the \lasso~\cite{Fuchs2004} to an infinite-dimensional setup.
They first show that the existence of a ``vanishing derivative pre-certificate''~\cite[Def.~6]{Duval2014} is necessary so that the support of the solution to the \blasso{} problem is exactly $\{\paramt_\ell\}_{\ell=1}^\nbAtome$~\cite[Prop.~8]{Duval2014}. 	
On the other hand, they show that a so-called ``non degenerate source condition''~\cite[Def.~5]{Duval2014} is sufficient to ensure the desired recovery~\cite[Th.~2]{Duval2014}. 
We note that these two conditions apply on dictionaries made up of differentiable atom functions $\atome:\paramSet\mapsto\spaceObs$ since they involve the first and second order derivatives of the inner product $\langle \atome(\param), \Vobs \rangle$ evaluated at $\{\paramt_\ell\}_{\ell=1}^\nbAtome$. This is in contrast with the ``CMF'' dictionaries considered in this paper which involve some non-differentiability in their kernel function, see \Cref{sec:CMF_dico}. 
Moreover, we also emphasize that both conditions involve the sign of the coefficients $\{\coeff_\ell\}_{\ell=1}^{\nbAtome}$.

The comparison with the discrete case goes even deeper: it can be shown that the solution of \BPCD{} is, in some sense, the limit of the solution of \blasso{}~\cite[Prop.~1]{Duval2014}.
Although out of scope of the present paper, we also mention the existence of a literature related to the robustness of \blasso{} in various noisy settings~\cite{Duval2014,Denoyelle2015,Poon2019}.

\conditionalPagebreak


\section{Main results}
	\label{sec:main_results}

In this section, we present the main results of the paper. 
In \Cref{subsec:recovery}, we describe the constitutive ingredients of the ``continuous'' setup addressed in this work and provide a rigorous definition of the notions of exact recovery that will be used in our statements. 
Our main results are presented in \Cref{subsec:contrib:abstract_conditions,subsec:contrib:ksparserecov}.
The family of ``CMF dictionaries'', central to our results in \Cref{subsec:contrib:ksparserecov}, is introduced in \Cref{sec:CMF_dico}.

\subsection{Main ingredients}
	\label{subsec:recovery}

We first present the three main properties that a ``continuous'' dictionary should verify, see~\eqref{eq:kernel_Gconditions:unitNorm},~\eqref{eq:kernel_Gconditions:continuity} and~\eqref{eq:contribution:evanecent} below.
We then elaborate on some differences between the implementation of \comp{} in the discrete and continuous settings. 
We finally give a precise definition of the notions of recovery that will be used in the statements of our results.


\paragraph{Continuous dictionary}
First, the space $\paramSet$ is usually assumed to be a connected metric space or a torus in dimension $\dimParam$.  
Hereafter, for the sake of conciseness, we will restrict our attention to the case where $\paramSet=\kR^\dimParam$ and assume that the kernel associated to the atoms obeys some vanishing property (see \eqref{eq:contribution:evanecent} below).\footnote{Our results can be adapted to any set $\paramSet$ such that step~\ref{line:algo:continuousOMP:findtheta} of \Cref{alg:continuousOMP} is well-posed, that is, at least one maximizer exists.}
A second common working hypothesis in the ``continuous'' setup is the continuity of function $\atome:\paramSet\to \spaceObs$, that is 
\begin{equation}
	\label{eq:continuity_atom}
	\lim_{\param'\rightarrow\param} \nor{\atome(\param')-\atome(\param)} = 0\quad\forall\param\in\paramSet. 
\end{equation}
In this paper, we will moreover suppose that the atoms of the dictionary are normalized:
\begin{equation}
	\label{eq:unit_norm_atom}
	\nor{\atome(\param)}=1\quad\forall\param\in\paramSet. 
\end{equation}
In the sequel, recovery conditions will be expressed as a function of the induced kernel $\kernelPaper\big(\param, \param'\big)$:
\begin{equation}
	\label{eq:contribution:hyp_scalar_product}
	\kernelPaper\kparen{\param, \param'} \triangleq \scalprod{\atome( \param )}{\atome( \param' )} \qquad \forall\param,\param'\in\paramSet.
\end{equation}
The \say{\textit{continuity}} and \say{\textit{unit-norm}} properties are equivalent to:
\begin{subequations}
	\label{eq:propertiesKernel}
	\begin{align}
		\mbox{\say{\textit{unit norm}}}: \qquad\,\,\kernelPaper\big(\param, \param\big)=1 &\quad\forall\param\in\paramSet
		,
		\label{eq:kernel_Gconditions:unitNorm}
		\\
		\mbox{\say{\textit{continuity}}}: \displaystyle{\lim_{\param'\rightarrow\param} \kernelPaper\big(\param, \param'\big)} = 1 &\quad\forall\param\in\paramSet
		.
		\label{eq:kernel_Gconditions:continuity}
	\end{align}
\end{subequations}
Moreover, we have from the Cauchy-Schwarz inequality that 
\begin{equation}
	\kvbar{\kernelPaper\big(\param, \param'\big)} \leq 1, \qquad \forall\param,\param'\in\paramSet.
\end{equation}
Lastly, in this work, we will restrict our attention to kernels that vanish at infinity, \ie
\begin{equation}
	\label{eq:contribution:evanecent}
	\kforall[\varepsilon>0, \forall\param\in\paramSet{}] 
	\;\exists K \text{ compact: }\; \sup_{\param'\in K^c} \kernelPaper(\param', \param) < \varepsilon
	,
\end{equation}
where $K^c$ is the complement of $K$ in $\paramSet$. This covers the case where $\paramSet$ is compact, by simply considering $K = \paramSet$ with the convention $\sup_{\param'\in\emptyset} \kernelPaper =0$.

\paragraph{OMP in continuous dictionaries} 
Although \Cref{alg:continuousOMP} corresponds to the standard definition of OMP in the discrete setting, its implementation in continuous dictionaries leads to two major differences.  
First, the ``atom selection'' step in~\Cref{line:algo:continuousOMP:findtheta} does not necessarily admit a maximizer. 
In such a case, the recursions defined in \Cref{alg:continuousOMP} are ill-posed since the procedure cannot elucidate the maximization problem in~\Cref{line:algo:continuousOMP:findtheta}. 
Second, even if a maximizer exists, solving the ``atom selection'' problem may be computationally intractable.
In particular, the function to be maximized in Line~\ref{line:algo:continuousOMP:findtheta} may have many local maxima and the problem is indeed NP-hard in certain cases (\textit{e.g.}, when the maximization step involves a rank-$1$ approximation of a tensor~\cite[Th.~1.13]{Hillar2013} as in \cite{elviraTensor2020}), while in other cases it is easy (for example, the SVD can be revisited in this framework and solved up to numerical precision~\cite[Sec.~2.2]{Chandrasekaran2012}).
The maximization problem of Line~\ref{line:algo:continuousOMP:findtheta} also appears in Frank-Wolfe type algorithms~\cite{Denoyelle2019}, where existing theoretical guarantees also hold under the hypothesis that this step can be solved.

\noindent
In this paper, for the sake of simplifying our theoretical analysis, we will nevertheless stick to the idealistic version of OMP described in \Cref{alg:continuousOMP}. 	
The results presented in this work should therefore be considered more for the theoretical insights they provide into the behavior of OMP in continuous dictionaries than for their practical implications. 

In our theoretical analysis we will only have to deal with residuals $\bsr$ which can be written as a linear combination of a finite number of atoms of $\dico$. 
In such a case, the following lemma shows that a maximizer to the ``atom selection'' problem always exists: 
\begin{lemma}
	\label{Fact:existence-maximizer}
	Let $\dico=\kset{\atome(\param)}{\param\in\paramSet}$ be a continuous dictionary with kernel $\kernelPaper\kparen{\param, \param'} \triangleq \scalprod{\atome( \param )}{\atome( \param' )}$ verifying the continuity property~\eqref{eq:kernel_Gconditions:continuity} and vanishing property~\eqref{eq:contribution:evanecent}.
	Then 
	\begin{align}
	 	\kargmax_{\param\in\paramSet} \kvbar{\kangle{\atome(\param),\bsr} }
	 	\neq\emptyset
	\end{align} 
	whenever $\bsr\in\spaceObs\backslash\{{\bf0}_{\spaceObs}\}$ is a finite linear combination of elements of $\dico$.
\end{lemma}
\noindent
A proof of this statement is available in Appendix~\ref{subsec:app:existence-maximizer}.

We finally emphasize that the solution to the OMP recursions may not be unique.
Indeed, in situations where the ``atom selection'' problem in~\Cref{line:algo:continuousOMP:findtheta} admits \textit{several} solutions, there may exist several output sets, $\supporte$ and $\{\coeff_\ell\}_{\ell=1}^{\widehat{\nbAtome}}$, verifying OMP recursions.  
Hereafter, given an observation vector $\Vobs$, we will call any set $\supporte$ which can be generated by OMP with $\Vobs$ as input, as a ``reachable support''.

\paragraph{Notions of recovery}
The recovery results stated in the next sections of the paper will involve the  following notions of success: \textit{``exact $\nbAtome$-step recovery of $\supportt$''} and  \textit{``exact $\support$-delayed recovery of $\supportt$''}.
We devote the remainder of this section to rigorously defining these two notions.

We say that \comp{} achieves exact recovery of coefficients $\{\coeff_\ell\}_{\ell=1}^{\nbAtome{}}\subset\kR*$ and atom parameters $\supportt\triangleq\{\paramt_\ell\}_{\ell=1}^{\nbAtome}$ if $\{\coeff_\ell\}_{\ell=1}^{\nbAtome{}}$ and $\supportt$ can be unambiguously identified from \textit{any} reachable outputs of \comp{} ($\supporte$ and $\{\coeffe_{\ell}\}_{\ell=1}^{\widehat{\nbAtome}}$) run with $\Vobs = \sum_{\ell=1}^{\nbAtome{}} \coeff_\ell\, \atome(\paramt_\ell)$ as input. 	
We note that a simple necessary and sufficient condition for exact recovery of $\{\coeff_\ell\}_{\ell=1}^{\nbAtome{}}$ and $\supportt$ reads
\begin{align}
	\supportt\subseteq \supporte,
\end{align}
for each reachable support $\supporte$.
This can be seen from the following arguments.
If there is a reachable support such that $\supporte\nsupseteq \supportt$, then exact recovery is obviously not attained since there exists some $\paramt_\ell\in\supportt$ that is not identified in $\supporte$. 
Conversely, if $\supportt \subseteq \supporte$ holds, one must have 
\begin{equation} \label{eq:recovery:supportt identification}
	\forall \ell\in \intervint{1}{\widehat{\nbAtome}}, \, 
	\begin{cases}
		\coeffe_\ell = \coeff_\ell & \mbox{if $\parame_\ell\in\supportt$}\\
		\coeffe_\ell = 0 
		 & \mbox{otherwise},
	\end{cases}
\end{equation}
because the atoms $\{\atome(\param):\param\in\supporte\}$ selected by OMP are always linearly independent and $\Vobs\in\mathrm{span}(\{\atome(\paramt_\ell)\}_{\ell=1}^{\nbAtome})$.
Therefore, $\supportt=\{\paramt_\ell\}_{\ell=1}^{\nbAtome{}}$ can be unambiguously identified from the non-zero elements of $\{\coeffe_\ell\}_{\ell=1}^{\widehat{\nbAtome}}$.

In the literature related to the conditions of success of \comp{}, a distinction is usually made between the cases ``$\supportt = \supporte$'' and ``$\supportt \subseteq \supporte$'': the former is referred to as ``$\nbAtome$-step recovery''  because it implies that \comp{} identifies $\supportt$ and $\{\coeff_\ell\}_{\ell=1}^{\nbAtome{}}$ in exactly $\nbAtome$ steps; 
the latter  is known as ``delayed recovery''  because \comp{} may require (if the inclusion is strict) to carry out more than $\nbAtome$ iterations to identify $\supportt$ and $\{\coeff_\ell\}_{\ell=1}^{\nbAtome{}}$.

In this paper  we will focus on conditions ensuring the correct identification of a given support $\supportt$ of cardinality $k$ for \textit{any} choice of the non-zero weighting coefficients $\{\coeff_\ell\}_{\ell=1}^{\nbAtome}$. 
The notions of \textit{``exact $\nbAtome$-step recovery of $\supportt$''} and \textit{``exact $\support$-delayed recovery of $\supportt$''} announced at the beginning of this section then read as follows. We say that \comp{} achieves \textit{``exact $\nbAtome$-step recovery of $\supportt$''} if $\supportt=\supporte$ for any choice of $\{\coeff_\ell\}_{\ell=1}^{\nbAtome}\subset\kR*$ and any reachable output $\supporte$.
This implies that there is only one reachable output. 
Moreover, given some set $\support\subset\paramSet$, we say that OMP achieves \textit{``exact $\support$-delayed recovery of $\supportt$''} if
\begin{align}
	\supportt \subseteq \supporte \subseteq \support
\end{align}
for any choice of $\{\coeff_\ell\}_{\ell=1}^{\nbAtome}\subset\kR*$ and any reachable output $\supporte$. 
``$\support$-delayed recovery'' can be regarded as a refined version of ``delayed recovery'' where the set of parameters that OMP  may select is guaranteed to belong to some set $\support\supseteq\supportt$. 
Uninterestingly this is always the case with $\support = \paramSet$, so what will be important in our results is to establish conditions such that we can identify a \emph{finite} set $\support$, determined by the only specification of $\supportt$, such that $\support$-delayed recovery of $\supportt$ holds. 
We note that $\support$-delayed recovery implies that OMP identifies $\supportt$ in at most $\card(\support)$ iterations.  
Finally, we emphasize that ``exact $\supportt$-delayed recovery of $\supportt$'' is equivalent to ``exact $\nbAtome$-step recovery of $\supportt$''. We will sometimes use the former in the formulation of our results to have more compact statements. We will also always implicitly assume that OMP achieves $\nbAtome$-step recovery of $\supportt$ when $\supportt=\emptyset$ since this implies that $\Vobs={\bf0}_\spaceObs$ and OMP returns the empty support $\supporte=\emptyset$ at iteration 0 in this case.


\subsection{Exact recovery of a given support: sufficient conditions}
\label{subsec:contrib:abstract_conditions}

In this section, we highlight some instrumental properties of the dictionary $\dico$ and support $\supportt$ which allow OMP to achieve exact $\card(\support)$-step recovery of each $\support\subseteq\supportt$ (see \Cref{th:MainAbstractTheorem}).
These conditions are the basis of our results on \say{CMF dictionaries} stated in \Cref{subsec:contrib:ksparserecov}. 

We first notice that, in the context of continuous dictionaries, the $\nbAtome$-step analysis of~\Cref{th:Tropp's-ERC} still applies: condition \eqref{eq:ERCTropp} along with the linear independence of the atoms $\{\atome(\param_\ell)\}_{\ell=1}^\nbAtome$ are still necessary and sufficient for exact recovery of a support~$\supportt$.\footnote{We note in particular that \eqref{eq:ERCTropp} ensures that the ``atom selection" step in \Cref{line:algo:continuousOMP:findtheta} of \Cref{alg:continuousOMP} is well-defined since the maximizers are ensured to belong to the finite set $\supportt$.} 
However, the standard formulation 
\begin{equation}
	\label{eq:ERCmax}
	\max_{\param\in \paramSet\setminus\supportt} \kvvbar{ \kinv{\bfG} \bfg_{\param}}_1 < 1,
\end{equation}
equivalent to \eqref{eq:ERCTropp} in the discrete setting, does no longer hold in the case of continuous dictionaries as the supremum 
\begin{equation}
	\sup_{\param\in \paramSet\setminus\supportt} \kvvbar{ \kinv{\bfG} \bfg_{\param}}_1
\end{equation}
is always at least $1$.\footnote{Indeed, first notice that $\kvvbar{ \kinv{\bfG} \bfg_{\param}}_1=1$ for all $\param\in\supportt$. One then obtains that the supremum is at least $1$ by continuity of $\param\mapsto\kvvbar{ \kinv{\bfG} \bfg_{\param}}_1$.
\label{footnote:sup=1}}

In order to circumvent this problem, we identify below two simpler conditions, respectively on the dictionary $\dico$ (\emph{via} its induced kernel $\kernelPaper$) and the support $\supportt$, which imply that the atoms $\{\atome(\paramt_\ell)\}_{\ell=1}^\nbAtome$ are linearly independent and that \eqref{eq:ERCTropp} is verified, see \Cref{th:MainAbstractTheorem} below.
The following definition includes assumptions on the kernel ensuring that the dictionary atoms are normalized, and that the atom function $\param \mapsto \atome(\param)$ is injective and continuous.

\begin{definition}[Admissible kernel]
	\label{hyp:geometric:0}
	A kernel $\kernelPaper{}$ is said to be {\em admissible} if:
	\begin{enumerate}[label=\itshape\roman*)]
		\item it verifies \eqref{eq:propertiesKernel} and \eqref{eq:contribution:evanecent}.
		\label{item:admissible-kernel:unit-norm-and-vanishing}

		\item $0 \leq \kernelPaper{}\big(\param, \param'\big) < 1$ for any $\param \neq \param'$.
		\label{item:admissible-kernel:decreasing}
	\end{enumerate}
	\noindent
	By extension, a dictionary $\dico=\kset{\atome(\param)}{\param\in\paramSet}$ is said to be \emph{admissible} if its induced kernel is admissible.
\end{definition}

\begin{definition}[Admissible support with respect to kernel $\kernelPaper$]
	\label{def:admissible_support}
	A support $\supportt=\{\paramt_\ell\}_{\ell=1}^\nbAtome$  is {\em admissible} with respect to a kernel $\kernelPaper$ if the following holds for  any non-empty subset $T\subseteq\intervint{1}{\nbAtome}$ and any positive coefficients $\{\coeff_\ell\}_{\ell\in T}\subset\kR*+$ such that $\sum_{\ell \in T} \coeff_\ell < 1$:
	\begin{enumerate}[label=\itshape\roman*)]
		\item	 \label{hyp:geometric:1}
		The set of global maximizers of 
		\begin{equation}
			\kfuncdef[m]{\condSumPositivKernelSymb}{\paramSet}{\kR+}[\param][\displaystyle\sum_{\ell\in T} \coeff_\ell\, \kernelPaper(\param,\paramt_\ell)]
			,\label{eq:hyp:geometric:1:func}
		\end{equation}
		is a subset of $\kfamily{\paramt_\ell}{\ell\in T}$.
		\item 	\label{hyp:geometric:2}
		If $\ell\in \intervint{1}{\nbAtome}\setminus T$ satisfies $\condSumPositivKernelSymb(\param) - \kernelPaper\kparen{\param,\paramt_{\ell}} \leq 0$ for all $\param\in\{\paramt_{\ell'}\}_{\ell'\in T}$, then
		\begin{equation}
			\label{eq:hyp:geometric:2:condition}
			\kforall[\param\in\paramSet{}]\quad \condSumPositivKernelSymb(\param) - \kernelPaper\kparen{\param,\paramt_{\ell}} \leq 0 .
		\end{equation} 
	\end{enumerate}
	\noindent
	By extension, the support $\supportt$ is said to be \emph{admissible} with respect to dictionary $\dico=\kset{\atome(\param)}{\param\in\paramSet}$ if $\supportt$ is admissible with respect to the kernel induced by $\dico$.
\end{definition}

\noindent
With these definitions, our first recovery result reads:\vspace{0.1cm}
\begin{theorem}
	\label{th:MainAbstractTheorem}
	Assume $\dico$ is admissible and $\supportt$ is admissible with respect to $\dico$.
	Then, OMP achieves exact $\card(\support)$-step recovery of each $\support\subseteq \supportt$.\vspace{0.1cm}
\end{theorem}

\noindent
A proof of this result is available in \Cref{sec:technical_details}. 
\Cref{th:MainAbstractTheorem} provides some sufficient conditions for exact $\card(\support)$-step recovery of any $\support\subseteq \supportt$ via the definitions of ``admissible dictionary'' and ``admissible support''. 
In particular, the conditions of \Cref{th:MainAbstractTheorem} imply that \eqref{eq:ERCTropp} is satisfied for any $\support\subseteq \supportt$. As we will see in \Cref{subsec:contrib:ksparserecov}, the admissibility of $\dico$ and $\supportt$ may be much easier to prove in some cases than verifying directly that \eqref{eq:ERCTropp} holds. 

As the admissibility conditions stated in Definitions \ref{hyp:geometric:0} and \ref{def:admissible_support} may appear somewhat technical, we discuss hereafter the different items appearing in these definitions in order to shed some light on the scope of \Cref{th:MainAbstractTheorem}.

In \Cref{hyp:geometric:0}, \eqref{eq:propertiesKernel} ensures that the kernel $\kernelPaper$ induced by $\dico$ is continuous and that the dictionary $\dico$ only contains unit-norm atoms. The continuity assumption is crucial in the derivation of our result since it induces a specific structure on the dictionary.
The unit-norm hypothesis is only secondary but allows to avoid some unnecessary technicalities in the proofs.
Hypothesis \eqref{eq:contribution:evanecent} ensures the well-posedness of the ``atom selection'' step in \Cref{line:algo:continuousOMP:findtheta} of \Cref{alg:continuousOMP}, see \Cref{Fact:existence-maximizer}. Finally, $0 \leq \kernelPaper(\param, \param')$ implies that the inner product between two atoms of $\dico$ is always nonnegative, whereas $\kernelPaper(\param, \param')<1$ guarantees that $\atome(\param)\ne \atome(\param')$ for $\param \ne \param'$, \emph{i.e.}, that $\param \mapsto \atome(\param)$ is an injective function (remember that we assume $\kernelPaper(\param,\param)=1$ for all $\param\in \paramSet$; the fact that atoms are distinct is thus a direct consequence of the Cauchy-Schwarz inequality).
The atoms $\atome(\param)\ne \atome(\param')$
for $\param \ne \param'$ being normalized, distinct, and positively correlated, they are also linearly independent.

As for \Cref{def:admissible_support}, item~\ref{hyp:geometric:1} ensures that a correct atom selection always occurs when the residual $\bsr$ is a \emph{positive} combination\footnote{\label{footnote:positive-combili} In particular, we note that if item~\ref{hyp:geometric:1} of \Cref{def:admissible_support} is true, then its conclusion still holds without the hypothesis ``$\sum_{\ell\in T}\coeff_\ell<1$''.} of the atoms of the support and the kernel $\kernelPaper$ is admissible.
Indeed, if $\bsr = \sum_{\ell=1}^\nbAtome \coeff_\ell\,  \atome(\paramt_\ell)$ with $\coeff_1,\ldots,\coeff_{\nbAtome}>0$ and the kernel is admissible then from \Cref{hyp:geometric:0}
\begin{equation}
	\kvbar{\scalprod{\atome(\param)}{\bsr}} = \sum_{\ell=1}^\nbAtome \coeff_\ell\, \kernelPaper(\param,\paramt_\ell).
\end{equation}
In such a case, item~\ref{hyp:geometric:1} of \Cref{def:admissible_support} then implies
\begin{equation}
	\kargmax_{\param\in\paramSet}\;\kvbar{\scalprod{\atome(\param)}{\bsr}}\subseteq \supportt. 
\end{equation}
Item~\ref{hyp:geometric:2} of \Cref{def:admissible_support} does not have such a simple interpretation but a careful  inspection of our proof in \Cref{sec:technical_details} shows that this condition is instrumental for deriving the result stated in \Cref{th:MainAbstractTheorem}. 
Altogether, given some admissible dictionary $\dico$, \Cref{th:MainAbstractTheorem} allows us to establish recovery results valid without sign constraints by only proving the two assumptions gathered in \Cref{def:admissible_support}, which somehow correspond to establishing the result for the easier case of positive combinations of atoms.


\subsection{CMF dictionaries}
\label{sec:CMF_dico}

In the next section, we will particularize \Cref{th:MainAbstractTheorem} to a family of dictionaries whose kernel is defined via a completely monotone function (CMF).
In this section, we provide a precise definition of this family of dictionaries and some of their properties that will be used throughout the paper. 

\noindent
We first recall the definition of a CMF:

\begin{definition}[CMF~\protect{\cite[Def.~7.1]{Wendland2005}}]
	\label{def:intro:def_cmf}
	A function $\cmf{}:\kR+\longmapsto\kR$ is completely monotone on $\kintervco{0}{+\infty}$ if it is infinitely differentiable on $\kintervoo{0}{+\infty}$, right continuous at $0$, and if its derivatives obey
	\begin{equation}
		\label{eq:intro:def_cmf}
		(-1)^{n} \cmf[][(n)](x) \geq 0 \quad \forall x, n\in\kR*+\times\kN.
	\end{equation}
\end{definition}

\noindent
As described in the following example, many well-known functions are CMFs: 
\begin{example}
	\label{ex:CMF}
	The following functions are completely monotone~\citep{miller2001}:
	\begin{itemize}
		\item the function $x \mapsto \cste^{-\lambda x}$ for $\lambda>0$ which gives birth to the Laplace kernel,
		\item the function $x\mapsto \tfrac{1}{1+\lambda x}$ for $\lambda>0$,
		\item ratios of modified Bessel functions of the first kind,
		\item a subset of the confluent hypergeometric functions (Kummer’s function),
		\item a subset of the Gauss hypergeometric functions.\\
	\end{itemize}
\end{example}

\noindent
By definition, CMFs are non-negative, non-increasing and convex functions.
Moreover, they admit an integral formulation in terms of Laplace transform of a Borel measure: 
\begin{lemma}[Bernstein-Widder theorem, {\cite[Th. 7.11]{Wendland2005}}]
	\label{lemma:geometric:cmf:integral_representation}
	A function $\cmf$ is completely monotone on $\kintervco{0}{+\infty}$ if and only if there exists a non-negative finite measure $\nu$ on Borel sets of $\kintervco{0}{+\infty}$ such that 
	\begin{equation}
		\label{eq:geometric:cmf:integral_representation}
		\cmf{}(x) =
		\int_{0}^{+\infty}
		\cste^{-\varIntCmf x} \diff\nu(\varIntCmf),
	\end{equation}
	where the integral converges for all $x \geq 0$ since $\nu$ is finite and $e^{-u x} \leq 1$.
\end{lemma}
\noindent 
We note for example that the Laplace kernel (see \Cref{ex:CMF}) is a CMF with representation measure equal to $\nu=\spike{\lambda}$ with $\lambda>0$.

In the sequel we will consider the following class of kernels and dictionaries whose definitions rely on the concept of CMF:
\begin{definition}[CMF kernel and dictionary]
	\label{def:intro:cmf_kernel}
	The class of {\em CMF kernels} in dimension $\dimParam\geq1$, denoted $\cmfkernelclassP[\dimParam]$, consists of all kernels $\kernelPaper: \kR^{\dimParam}\times\kR^{\dimParam} \to \kR*+$ such that 
	\begin{equation}
		\label{eq:def:cmf_kernel}
		\kernelPaper{}(\param,\param') 
		= \cmf{}\kparen{ \kvvbar{\param -  \param'}_p^p } \quad \forall \param,\param'\in\kR^{\dimParam}
	\end{equation}
	where $\cmf{}$ is a CMF verifying $\cmf(0)=1$, $\lim_{x\to+\infty}\cmf(x)=0$ and $0 < p \leq 1$.\\
	\noindent
	By extension, we say that $\dico$ is a CMF dictionary in dimension $\dimParam\geq 1$ if its induced kernel belongs to $\cmfkernelclassP[\dimParam].$
\end{definition}
\noindent
We note that the constraint $\cmf(0)=1$ in the previous definition ensures that the \say{unit-norm} hypothesis \eqref{eq:kernel_Gconditions:unitNorm} is satisfied.
We also mention that the constraint $\lim_{x\to+\infty}\cmf(x)=0$ is necessary so that CMF kernels satisfy the vanishing property~\eqref{eq:contribution:evanecent}.

At this point, a legitimate question is whether kernel \eqref{eq:def:cmf_kernel} can be induced by some dictionary $\dico$? 
The answer is positive and is a corollary of the following lemma: 
\begin{lemma}\label{lemma:pd_CMF}
	Let $\cmf : \kR+ \to \kR$ be a CMF such that $\cmf(0)=1$, $\lim_{x \to \infty}\cmf(x)=0$ and $0<p\leq 1$. Then, any function of the form
	\begin{equation}
		\kfuncdef{\rho}{\kR^\dimParam}{\kR}[\boldsymbol{\omega}][\cmf{}(\|\boldsymbol{\omega}\|_p^p)]
	\end{equation}
	is positive definite. 
\end{lemma} 
\noindent
A proof of this result is provided in \Cref{sec:app:cmf_related}.
We refer the reader to \cite[Def.~6.1]{Wendland2005} for a precise definition of positive (semi-) definite functions and \cite[Th.~6.2]{Wendland2005} for a review of some of their basic properties. 
In particular, the positive definite nature of $\cmf{}(\|\cdot\|_p^p)$ used in conjunction with standard results in the theory of ``reproducing kernel Hilbert spaces'' (see \eg \cite[Th.~3.11]{Taylor2004}) implies the following corollary:
\begin{corollary}[Existence of CMF dictionaries] \label{cor:existence CMF dictionary}
	For any $\kernelPaper\in \cmfkernelclassP[\dimParam]$, there exists some Hilbert space $\spaceObs$ and some (continuous) function $\atome:\kR^\dimParam \to \spaceObs$ such that \eqref{eq:contribution:hyp_scalar_product} holds.~Moreover, any finite collection of distinct elements from $\dico=\{\atome(\param):\param\in\kR^\dimParam\}$ is linearly independent.
\end{corollary}

\noindent
We see from the last part of the corollary that CMF dictionaries are necessarily defined in infinite-dimensional Hilbert spaces $\spaceObs$. If not, any collection of $\dim(\spaceObs)+1$ elements of $\dico$ would be linearly dependent which is in contradiction with \Cref{cor:existence CMF dictionary}. 
The next example exhibits a family of atoms in $\spaceObs = L^2(\kR)$ which is a CMF dictionary in $\bbR$.
\begin{example}
	Let $\paramSet=\kR$ and consider the dictionary $\dico$ defined by 
	\begin{equation}
		\label{eq:intro:1D:atom_decayingPulse}
		\begin{split}
			\atome : \kR \;\longrightarrow\;& L^2(\kR) \\
			\param \;\longmapsto\;& 
					f_\param(t) = \sqrt{2\lambda}\,\cste^{-\lambda(t - \param)} {\bf1}_{\kbrace{t \geq \param}}(t)
		\end{split}
	\end{equation}
	for some $\lambda>0$, where ${\bf1}_{\kbrace{t \geq \param}}$ is the \say{indicator} function which is equal to $1$ if $t \geq \param$ and 0 otherwise.
	Straightforward calculations both show that $\Vert \atome({\param}) \Vert = 1$ for any $\param$ and the inner product in $\spaceObs = L^2(\kR)$ between two atoms writes 
	$\scalprod{\atome\big(\param\big)}{\atome\big(\param'\big)} = \cste^{-\lambda |{\param-\param'}|} $.
	The latter function corresponds to the so-called \say{Laplace kernel}. 
	This kernel is an element of $\cmfkernelclassP[1]$ according to \Cref{ex:CMF}.
\end{example}

We conclude this section by introducing a particular CMF kernel which will be used in the statement of some of our results in \Cref{subsec:contrib:ksparserecov}:

\begin{definition}[Generalized Laplace kernel and dictionary]
	\label{def:contrib:laplace_dico}
	The class of \emph{Generalized Laplace kernels} in dimension $\dimParam{}$, denoted $\laplacekernelclassP[\dimParam]$,  consists of all kernels $\kernelPaper: \kR^{\dimParam}\times\kR^{\dimParam} \to \kR*+$ such that 
	\begin{equation}\label{eq:laplacekernel}
		\kernelPaper{}(\param,\param') 
		=
		\cste^{- \lambda \|\param -  \param'\|_p^p } \quad \forall \param,\param'\in\kR^{\dimParam}
	\end{equation}
	where $\lambda>0$
	and $0 < p \leq 1$. \\ 
	\noindent
	By extension, a \emph{Generalized Laplace dictionary} in dimension $\dimParam\geq1$ is a collection of atoms $\dico=\kset{\atome(\param)}{\param\in\kR^\dimParam}$ whose induced kernel belongs to  $\laplacekernelclassP[\dimParam]$.
\end{definition}
\noindent
One immediately sees that $\laplacekernelclassP[\dimParam]\subset\cmfkernelclassP[\dimParam]$ since the function $t\mapsto \cste^{-\lambda t}$ defined on $\kR+$ is a CMF (see \Cref{ex:CMF}).


\subsection{Recovery conditions in {CMF} dictionaries}
	\label{subsec:contrib:ksparserecov}

In this section, we provide recovery results for OMP in CMF dictionaries. The proofs of our results are based on the sufficient conditions presented in \Cref{th:MainAbstractTheorem} and are reported in Appendix~\ref{sec:app:proofs-CMF}.
A first surprising result holds when $\paramSet=\kR$:
\begin{theorem}
	\label{th:contrib:cmf_uniformRecov_1D}
	Assume $\dico$ is a CMF dictionary in dimension $1$.
	Then, \comp{} achieves exact $\card(\supportt)$-step recovery of each finite support $\supportt\subset\kR$.
\end{theorem}

In essence, \Cref{th:contrib:cmf_uniformRecov_1D} identifies a class of dictionaries for which exact $\nbAtome$-step recovery is possible for \emph{any} support $\supportt$ of \emph{any} finite size $\nbAtome$.
We note that the notions of exact recovery of a support $\supportt$ defined in \Cref{subsec:recovery} do not involve any sign constraint on the coefficients $\{\coeff_\ell\}_{\ell=1}^\nbAtome$ used to generate the observation vector $\Vobs$. 
As a comparison, the results ensuring the success of continuous BP/BLasso with no sign constraints on $\{\coeff_\ell\}_{\ell=1}^\nbAtome$ require some ``minimum separation condition'' between parameters $\{\paramt_\ell\}_{\ell=1}^{\nbAtome}$ to hold (see~\eqref{eq:state_of_the_art:separation_BP} and related discussion).
Conversely, the recovery results for BLasso obtained in \cite{castro2012} without separation condition require weighting coefficients $\{\coeff_\ell\}_{\ell=1}^\nbAtome$ to be positive. 
The novelty of \Cref{th:contrib:cmf_uniformRecov_1D} is thus a separation-free recovery result for any signed finite linear combination of atoms. 
The strength of the result obtained in \Cref{th:contrib:cmf_uniformRecov_1D} comes however at a price: it applies to a specific family of dictionaries, namely CMF dictionaries.
In particular, as mentioned in \Cref{sec:CMF_dico}, the space $\spaceObs$ in which CMF dictionaries live is necessarily infinite-dimensional, and the corresponding kernels  exhibit a discontinuity in all their partial derivatives at $\param=\param'\in\paramSet$. 
Another price to pay is that the recovery guarantees are for OMP, an algorithm explicitly involving the search for the global maximum of an optimization problem, cf \Cref{line:algo:continuousOMP:findtheta} of \Cref{alg:continuousOMP}.

In higher dimension $\dimParam>1$, the ``universal'' exact recovery result stated in \Cref{th:contrib:cmf_uniformRecov_1D} no longer holds, as shown in the next example.
More precisely, if $\dimParam\geq3$, we emphasize that there always exists a configuration of parameters $\{\paramt_\ell\}_{\ell=1}^\nbAtome$ such that \comp{} fails at the first iteration for some $\{\coeff_\ell\}_{\ell=1}^\nbAtome\subset\kR*$:
\begin{example}
	\label{ex:intro:erc_required_dim2}
	Let $\dimParam \geq 3$ and $3 \leq \nbAtome\leq \dimParam$. Consider $\supportt\triangleq\{\paramt_\ell\}_{\ell=1}^{\nbAtome} \subset \kR^{\dimParam}$ and $\Delta>0$ such that
	\begin{align}\nonumber
		\begin{array}{rll}
		\|\paramt_\ell - \paramt_{\ell'}\|_p^p &= 2\Delta^p & \forall \ell\neq {\ell'}\\
		\|\paramt_\ell - {\bf0}_{\dimParam}\|_p^p &= \Delta^p & \forall \ell.
		\end{array}
	\end{align}
	\noindent
	Let $\atome: \kR^{\dimParam}\mapsto \spaceObs$ define a CMF dictionary in $\kR^{\dimParam}$ with kernel $\kernelPaper=\cmf\kparen{\|{\cdot-\cdot}\|_p^p}$.
	We next show that, if $\Delta$ is sufficiently small, there always exists a linear combination of $\{\atome(\paramt_\ell)\}_{\ell=1}^\nbAtome$ such that OMP selects a parameter \emph{not} in $\supportt$ at the first iteration. 

	Let us consider $\Vobs=\sum_{\ell=1}^{\nbAtome}\coeff_\ell\atome(\paramt_\ell)$ and assume that all coefficients $\coeff_\ell$ are equal. 
	We then have
	\begin{equation}
		\label{eq:contrib:ratio_scalprod}
		\frac{
			\scalprod{\atome({\bf0}_{\dimParam})}{\Vobs} 
		}{
			\scalprod{\atome(\paramt_\ell)}{ \Vobs }
		}
		=
		\frac{
			\nbAtome{} \cmf(\Delta^p)
		}{
			1 + (\nbAtome{} - 1) \cmf(2\Delta^p)
		} 
		.
	\end{equation}
	Then, $\param={\bf0}_{\dimParam}$ will be preferred to all \say{ground-truth} parameters $\paramt_{\ell}$ at the first iteration of OMP as soon as the quantity in \eqref{eq:contrib:ratio_scalprod} is larger than 1, or, equivalently,
	\begin{equation}
		\label{eq:contrib:ratio_scalprod2}
		(\nbAtome{} - 1) \cmf(2\Delta^p) - \nbAtome{}\cmf(\Delta^p) + 1 < 0.
	\end{equation}
	Let us show that~\eqref{eq:contrib:ratio_scalprod2} holds whenever $\Delta^p$ is \say{sufficiently small}.
	For simplicity, consider first the case where $\cmf(t) = \cste^{-\lambda t}$ with $\lambda>0$.
	Condition~\eqref{eq:contrib:ratio_scalprod2} writes
	\begin{equation}
		\label{eq:contrib:ratio_scalprod3}
		(\nbAtome-1)x^2 - \nbAtome x + 1 < 0
	\end{equation}
	with $x=\cmf(\Delta^p)=\cste^{-\lambda\Delta^p}$. 
	As $k \geq 3$, the left-hand side of~\eqref{eq:contrib:ratio_scalprod3} is a second order polynomial with two distinct roots, namely  $(\nbAtome-1)^{-1}$ and $1$.
	Therefore, OMP prefers ${\bf0}_{\dimParam}$ as soon as $\kinv{(\nbAtome-1)}<x<1$ or, equivalently, $\Delta^p < \lambda^{-1}\log(\nbAtome-1)$.
	The latter condition implies a \emph{necessary separation} condition such that \comp{} does not fail at the first iteration. 
	We note that it is possible to draw similar conclusions whenever $\cmf$ is a CMF function right-differentiable at zero. 
	The proof of this result requires extra work that is detailed in Appendix~\ref{sec:other:calcul_separation}.
\end{example}

Although a ``universal'' $\nbAtome$-step recovery result such as \Cref{th:contrib:cmf_uniformRecov_1D} no longer holds in CMF dictionaries when $\dimParam>1$, it is nevertheless possible to show that some form of exact recovery of a support $\supportt$ is possible under an additional condition on the kernel induced by the CMF dictionary (see \Cref{th:contrib:dimD:grid_recovery}). This additional condition is referred to as ``axis admissibility'' hereafter and is encapsulated in \Cref{def:contrib:grille_admissibility} below. 
Before moving on to this definition, it is necessary to introduce the notions of ``Cartesian grid'' and ``set augmenter operator'':
\begin{definition}[Cartesian grid]
	\label{def:contrib:cartesian_grid}
	A finite set $\calG\subset \kR^{\dimParam}$ is a \emph{Cartesian grid} in dimension $\dimParam\geq1$ if there exists $\dimParam$ one-dimensional finite sets $\{\support_d\}_{d=1}^\dimParam$ such that
	\begin{equation}
		\calG = \prod_{d=1}^\dimParam \support_d,
	\end{equation}
	where $\prod$ denotes the Cartesian product.
\end{definition}
\noindent
We moreover define the following \say{set augmenter} operator that, given a finite set $\support\subset\kR^\dimParam$, returns the smallest Cartesian grid containing $\support$:
\begin{equation}
	\label{eq:erc_ell1:defAtomeVirtuel}
	\SetAug(\support) \triangleq \prod_{d=1}^{\dimParam} \kset{ \param\element{d} }{ \param\in\support \vphantom{\tfrac{1}{\Vert_p}} }
	.
\end{equation}
It is quite straightforward to see that $\support\subseteq \SetAug(\support)$ for any finite set $\support\subset\kR^\dimParam$ and that the operator $\SetAug$ is idempotent.
We illustrate the definition of $\SetAug(\support)$ in Fig.~\ref{fig:laplaceKernel:defVirtualPonts} in dimension $\dimParam=2$ for $\support = \kbrace{\param_1, \param_2, \param_3}$.\\

\begin{figure}
	\centering
	\includegraphics[width=0.9\columnwidth]{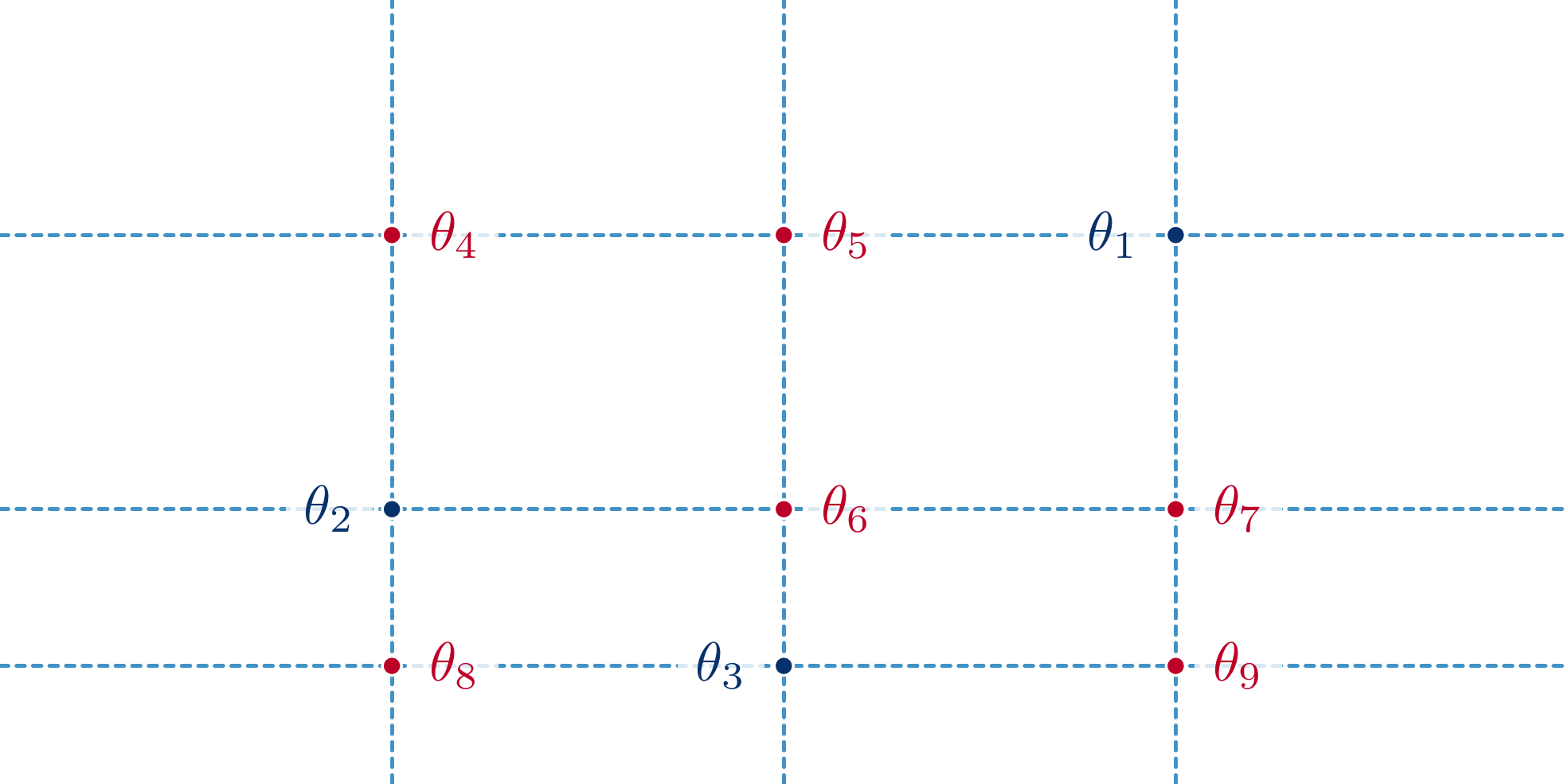}
	\caption{
		\label{fig:laplaceKernel:defVirtualPonts}
		Illustration in dimension $\dimParam=2$ with $\nbAtome=3$ of the definition of the set augmenter $\SetAug$ defined in~\eqref{eq:erc_ell1:defAtomeVirtuel}.
		The blue points, denoted $\param_\ell$ for $\ell\in\{1,2,3\}$, form the support $\support$.
		The red points, denoted $\param_\ell$, $\ell\in\intervint{4}{9}$ represent the elements of $\SetAug(\support)\setminus\support$.
	}
\end{figure}

\noindent
We are now ready to introduce the notion of ``axis admissibility'':

\begin{definition}[Axis admissibility with respect to a kernel]
	\label{def:contrib:grille_admissibility}
	\noindent
	A Cartesian grid $\calG = \prod_{d=1}^\dimParam\support_{d}=\{\param_\ell\}_{\ell=1}^{\card(\calG)}$ is said to be axis admissible with respect to a kernel $\kernelPaper$ if and only if 
	$\forall d\in\intervint{1}{\dimParam}$,  $\forall \param \in \kR^\dimParam$ with $\param[d]=0$ and $\forall \{\coeff_\ell\}_{\ell=1}^{\card(\calG)}\subset\kR$ such that the function
	\begin{equation}
		f_d(t) = \kvbar{\sum_{\ell=1}^{\card(\calG)} \coeff_\ell \, 
			\kernelPaper(\param + t\bfe_d, \param_\ell)}
	\end{equation}
	is not identically zero, we have
	\begin{equation}
		\label{eq:def-grid_admissibility}
		\emptyset \neq \kargmax_{t\in\kR} f_d(t) \subseteq \support_d.
	\end{equation}

	\noindent
	By extension, a Cartesian grid $\calG$ is said to be axis admissible with respect to a dictionary $\dico$ if it is axis admissible with respect to the kernel induced by $\dico$.
\end{definition}

The notion of axis admissibility will be central in our next result to ensure the exact recovery of some support $\supportt=\{\paramt_\ell\}_{\ell=1}^\nbAtome$ in a CMF dictionary. 
In particular, we will see that axis admissibility of $\SetAug(\supportt)$ ensures exact $\SetAug(\supportt)$-delayed recovery of each $\support\subseteq \supportt$.
Moreover, exact $\supportt$-delayed recovery of each $\support\subseteq \supportt$ is achievable by combining axis admissibility of $\SetAug(\supportt)$ with the following restricted version of the ERC:\footnote{We remind the reader that $\Gtrue$ is invertible as the Gram matrix of a set of linearly independent atoms (see \Cref{cor:existence CMF dictionary}). \label{footnote:inversibility-bfG}}
\begin{align}
			\stepcounter{equation}
			\tag{\theequation-R-ERC}
			\label{eq:contrib:dimD:l1erc}
			\max_{\param\,\in\,\textnormal{$\SetAug(\supportt)$}\setminus\supportt} \kvvbar{ 
				\kinv{\bfG} \gtheta
			}_1 \;<\; 1 
\end{align}
where
\begin{align}
	\begin{array}{rlr}
		\Gtrue{}\element{\ell,\ell'}&\triangleq\scalprod{\atome(\paramt_\ell)}{\atome(\paramt_{\ell'})} & \forall \ell,\ell'\in\intervint{1}{\nbAtome{}}\\
		\gtheta\element{\ell}&\triangleq\scalprod{\atome(\param)}{\atome(\paramt_\ell)} & \forall \ell\in\intervint{1}{\nbAtome}
	\end{array}		
	.
\end{align}
Formally, our next result writes as follows:
\begin{theorem}
	\label{th:contrib:dimD:grid_recovery}
	Let $\dico$ be a CMF dictionary in $\kR^\dimParam$ with induced kernel $\kernelPaper$ and let $\supportt=\{\paramt_\ell\}_{\ell=1}^{\nbAtome}$. 
	\begin{itemize}
		\item If $\SetAug(\supportt)$ is axis admissible with respect to $\kernelPaper$, then \comp{} achieves $\SetAug(\supportt)$-delayed recovery of each $\support\subseteq \supportt$.
		If \eqref{eq:contrib:dimD:l1erc} moreover holds, \comp{} achieves $\supportt$-delayed recovery of each $\support\subseteq \supportt$.
		\item Conversely, if \eqref{eq:contrib:dimD:l1erc} does not hold, there exists not all-zero coefficients $\{\coeff_\ell\}_{\ell=1}^{\nbAtome}$ such that \comp{} with $\Vobs = \sum_{\ell=1}^{\nbAtome} \coeff_\ell\, \atome(\paramt_\ell)$ as input selects some $\param \notin \supportt$ at the first iteration. 
\end{itemize} 
\end{theorem}

A first outcome of \Cref{th:contrib:dimD:grid_recovery} is a (pessimistic) upper bound on the number of iterations needed to identify $\supportt$ when $\dico$ is a CMF dictionary and $\SetAug(\supportt)$ is axis admissible with respect to the kernel induced by $\dico$.  
In particular, the first part of the theorem states that \comp{} needs no more than $\card(\SetAug(\supportt)) \leq 
\nbAtome^\dimParam$ iterations to succeed. 
As shown in the second part of the theorem, this (rather pessimistic) upper bound on the number of iterations can be decreased to $\nbAtome$ if an additional restricted ERC \eqref{eq:contrib:dimD:l1erc} is verified.
Interestingly,  whereas the parameter space $\paramSet$ is a continuum, \eqref{eq:contrib:dimD:l1erc} only depends on a finite subset of the elements of $\paramSet$ (namely $\SetAug(\supportt)$) and its numerical evaluation is therefore possible. 
We will see in \Cref{prop:intro:erc_ell1:separation} below, that this restricted ERC allows us to derive a separability condition for exact $\nbAtome$-step recovery in Generalized Laplace dictionaries.
Besides, we note that additional strategies could be investigated to improve the upper bound, exploiting, \textit{e.g.}, coefficients decay~\cite{Herzet2016}.

In our next result, we show that the property of ``axis admissibility'' can be (at least) satisfied for some CMF dictionaries.
In particular, the next lemma emphasizes that any Cartesian grid is axis admissible for Generalized Laplace dictionaries (see \Cref{def:contrib:laplace_dico}): 
\begin{lemma}
	\label{th:contrib:laplace_recovery}
	Let $\dico$ be a Generalized Laplace dictionary in $\kR^\dimParam$.
	Then {\em all} Cartesian grids $\calG$ are admissible with respect to $\dico$.
\end{lemma}

\noindent
A proof of this result is given in \Cref{subsec:proof:laplace-grid-are-admissible}. 
Combining this lemma with \Cref{th:contrib:dimD:grid_recovery} immediately leads to the following corollary:

\begin{corollary}
	\label{cor:contrib:laplace_recovery}
	Let $\dico$ be a Generalized Laplace dictionary in $\kR^\dimParam$.
	Then \comp{} achieves exact $\SetAug(\supportt)$-delayed recovery of each finite support $\supportt \subset \kR^\dimParam$.   
\end{corollary}

Interestingly, although \Cref{ex:intro:erc_required_dim2} showed that exact $\card{(\supportt)}$-step recovery does not hold  for arbitrary $\supportt$ in CMF dictionaries, \Cref{cor:contrib:laplace_recovery} emphasizes that exact $\SetAug(\supportt)$-delayed recovery is achievable by OMP in Generalized Laplace dictionaries for any $\supportt$ and any $\nbAtome=\card(\supportt)\in\kN*$.
Following our remark below \Cref{th:contrib:dimD:grid_recovery}, \comp{} is thus ensured to identify any support of size $\nbAtome$ in at most $\nbAtome^\dimParam$ iterations in this type of dictionaries.
Similar to \Cref{th:contrib:cmf_uniformRecov_1D}, no separation assumptions nor sign constraints are needed here to ensure our recovery result, although it applies to a very specific family of dictionaries.
We will see in \Cref{prop:intro:erc_ell1:separation} below that adding some separation condition on the elements of $\supportt$ enables to verify \eqref{eq:contrib:dimD:l1erc} and therefore leads to an exact-recovery result in at most $\nbAtome$ steps.

Before moving on to the statement of this result, let us mention that, 
although \Cref{th:contrib:laplace_recovery} shows that any Cartesian grid is axis admissible with respect to Generalized Laplace dictionaries, such a result does in general not hold for CMF dictionaries without extra assumptions on the grid. 
Nevertheless, our empirical evidence suggests that the admissible grid assumption is only an artifact of our proof technique. 
We conjecture that \Cref{th:contrib:dimD:grid_recovery} remains valid even when the Cartesian grid $\calG$ is not axis admissible. 
To support our conjecture, we show in \Cref{sec:app:dimD:uniformrecovery_k_2} that 
the second part of \Cref{th:contrib:dimD:grid_recovery} still holds for {\em any} CMF dictionary and 
for any $\supportt\subset \kR^\dimParam$ with $\card(\supportt)=2$, even though $\SetAug(\supportt)$ is generally not axis admissible. 
The proof of this kind of result in the general case is still under investigation.


In the last result of this section, we particularize \eqref{eq:contrib:dimD:l1erc} to derive a separation condition on the elements of $\supportt=\{\paramt_\ell\}_{\ell=1}^{\nbAtome}$ that ensures exact $\card(\support)$-step recovery of 
each  $\support \subseteq\supportt$ in Generalized Laplace dictionaries.  
We first note that, following standard results of the literature (see \eg \cite{tropp2004}), \eqref{eq:contrib:dimD:l1erc} can be relaxed to a mutual coherence condition:
\begin{equation}
	\mu < \frac{1}{2\nbAtome-1}
\end{equation}
where 
\begin{align}
	\label{eq:restricted-mutual-coherence}
	\mu \;\triangleq \max_{\substack{\param,\param'\in\SetAug(\supportt)\\ 
	\param\neq\param'}} 
	\kvbar{
		\scalprod{\atome(\param)}{\atome(\param')}
	}.
\end{align}	
Our separation result is then a simple consequence of this mutual coherence condition: 	
\begin{theorem}
	\label{prop:intro:erc_ell1:separation}
	Let $\dico$ be a Generalized Laplace dictionary in $\kR^\dimParam$ with  parameters $\lambda>0$ and $0<p \leq 1$. 
	Consider $\supportt$$=\{\paramt_\ell\}_{\ell=1}^{\nbAtome}$ and let
	\begin{equation}
		\label{eq:laplaceKernel:minSepNonZero}
		\minSepNonZero \, \triangleq 
		\min_{d\in\intervint{1}{\dimParam}} \min \left\{
		 \kvbar{\paramt_\ell\element{d} - \paramt_{\ell'}\element{d} }:
		 \ell,\ell'\in\intervint{1}{\nbAtome}\ \text{s.t. } 
			\paramt_\ell\element{d} \neq \paramt_{\ell'}\element{d}\right\}.
	\end{equation}
	If
	\begin{equation}
		\label{eq:laplaceKernel:erc_ell1:minSepNonZero}
		\minSepNonZero^{p} \geq \frac{\log(2k-1)}{\lambda}
	\end{equation}
	then, \comp{} achieves exact $\card(\support)$-steps recovery of each $\support\subseteq\supportt$.
\end{theorem}
\begin{proof}
	By definition of $\minSepNonZero$ and of $\SetAug(\supportt)$, we have $\| \param - \param' \|_p^p \geq \minSepNonZero^p$ for all $\param,\param' \in \SetAug(\supportt)$. Hence, using the definition of the mutual coherence in \eqref{eq:restricted-mutual-coherence} we have $\mu=\exp(-\lambda \| \param - \param' \|_p^p)$ for some $\param,\param'\in\supportt$ so $\mu \leq \exp(-\lambda \minSepNonZero^{p})$ and \eqref{eq:laplaceKernel:erc_ell1:minSepNonZero} implies that $\mu < (2\nbAtome-1)^{-1}$ holds.
\end{proof}

\Cref{prop:intro:erc_ell1:separation} states that, with Generalized Laplace dictionaries, \comp{} recovers any linear combination of $\nbAtome$ sufficiently separated atoms in $\nbAtome$ steps.
Although condition~\eqref{eq:laplaceKernel:erc_ell1:minSepNonZero} is expressed in terms of minimal distance between parameters, it can be seen as a condition on the mutual coherence between atoms.
However, in contrast to the discrete case, this mutual coherence guarantee is only related to a particular finite subset of the (continuous) Generalized Laplace dictionary, namely the atoms with parameters in $\SetAug(\supportt)$.

Furthermore, condition~\eqref{eq:laplaceKernel:erc_ell1:minSepNonZero} is reminiscent of the separation condition for off-the-grid super-resolution proposed in~\cite{candes2014}, see~\eqref{eq:state_of_the_art:separation_BP}.
The so-called \emph{separation condition} discussed in~\eqref{eq:state_of_the_art:separation_BP} is expressed on a $\dimParam$-dimensional torus preventing also high values of $\nbAtome$.
For example, in a unit-length $1$-dimensional torus and with the notations of~\eqref{eq:state_of_the_art:separation_BP}, the minimum separation condition for BP requires $\nbAtome \leq \tfrac{f_c}{C}-1$. 
Note however that these results involve different dictionaries and settings making relevant comparison tedious.

\conditionalPagebreak


\section{Conclusion - discussion}
	\label{sec:discussion}

In this work, we have shown that the study of the recovery properties of greedy procedures such \emph{Orthogonal Matching Pursuit} (\comp{}) can be extended to the setting of continuous dictionaries where the atoms continuously depend on some parameters.
Capitalizing on the formulation of \comp{} in terms of inner products between atoms, our results rely on the properties of the kernel implicitly defined by the inner product between atoms.
More particularly, we have identified two key notions which we have called \emph{admissible kernel} and \emph{admissible support}, that are sufficient to ensure exact recovery irrespective of the value of the coefficients involved in the representation.
For the class of CMF dictionaries, we have shown that when the dimension of the parameter space is $1$, all implicitly defined kernels as well as all supports are admissible.
Up to our knowledge, this is the first class of kernels for which no separation is needed to achieve exact recovery, even for signed combinations of atoms.
However, such a ``universal'' recovery result comes at a price since CMF dictionaries can only live in infinite-dimensional observation spaces $\spaceObs$ and the corresponding kernels exhibit some discontinuities in their derivatives.

\noindent
Although exact recovery can also be ensured for CMF dictionaries with a parameter space of dimension greater then 1, extra conditions have to be imposed on the support to be recovered, as some supports may not be admissible anymore.
The cornerstone of our analysis in the multi-dimensional case is the notion of \emph{axis admissible} Cartesian grid.
Indeed, axis admissibility is sufficient to allow \comp{} to identify supports, leading to a form of \say{delayed recovery} for all supports of size $\nbAtome$ embedded in some admissible Cartesian grid.
For such supports, exact $\nbAtome$-step recovery can also be achieved whenever a condition on a finite number of 
(known) atoms is fulfilled. 
In the special case of Generalized Laplace dictionaries, any Cartesian grid turns out to be axis admissible, and a simplified \emph{coherence-based analysis} can be revisited, leading to exact $\nbAtome$-step recovery under a \emph{minimal separation condition}.\\

\noindent
We now review some prospects of this work:

\paragraph{Beyond axis-admissible grids for CMF kernels}
Our analysis for multi-dimensional parameter sets relies on the notion of axis-admissible grids.
While axis admissibility holds for any grid with respect to Generalized Laplace dictionaries, this is apparently no longer the case with respect to more general CMF dictionaries.
Even for grids which seem to violate the axis-admissibility condition with respect to a CMF dictionary, empirical evidence suggests that \Cref{th:contrib:dimD:grid_recovery} remains valid.
As a first step towards a better understanding of this phenomenon, we showed in \Cref{sec:app:dimD:uniformrecovery_k_2} that, for supports of size 2, axis-admissibility is not necessary for the conclusion of \Cref{th:contrib:dimD:grid_recovery} to hold.

\paragraph{Connection with TV-minimization}
In light of the existing links between Tropp's ERC \citep{tropp2004} and recovery guarantees for $\ell_1$ minimization~\cite{Tropp2006}, an interesting question is whether the guarantees developed in this paper can be extended to sparse spike recovery with total variation norm minimization (see~\Cref{sec:state_of_the_art}).
More particularly, one could benefit from the \emph{null-space properties for measures}~\cite{castro2012} which characterize the solution of the continuous version of Basis Pursuit.
Such a connection may yield support recovery results for signed combinations of atoms with TV-norm minimization without separation conditions.

\paragraph{Robustness to estimation error}
In the discrete setting, one advantage of greedy procedures over convex relaxations is that the associated recovery guarantees involve solutions provided by actual \emph{algorithms} rather than merely expressed as the minimizer of some \emph{optimization problem}. 
In the continuous setting, this has to be tempered with the fact that implementing \comp{} requires a (possibly intractable) global maximization procedure at each iteration.
Our current analysis does not take into account the resulting numerical estimation error or the fact that there may be spurious local maxima.
One could envision overcoming some of these limitations by analyzing the behavior of \comp{} when a small error is systematically done when maximizing the inner product in \Cref{line:algo:continuousOMP:findtheta} of \Cref{alg:continuousOMP}.
Note that such an approximation error may also be useful to account for discretized implementations of the latter step of \comp{} using a fine grid over the parameter set $\paramSet$.

\conditionalPagebreak

\appendix


	\newcommand{\hyprec}{\calH_r}

	\newcommand{\Gk}{\overline{\Gtrue}}
	\newcommand{\gkthetaSymb}{\overline{\gthetaSymb}}
	\newcommand{\gktheta}[1][\param]{\gkthetaSymb_{#1}}
	\newcommand{\gkthetanew}{\gktheta[\nbAtome']}
	\newcommand{\gkthetanewTransp}{\ktranspose{\gkthetaSymb}_{\nbAtome'}}
	\newcommand{\supportk}{\overline{\support}}

	\newcommand{\Gnew}{\Gtrue}
	\newcommand{\gnew}[1][\paramt_{\nbAtome}]{\gtheta[#1]}
	\newcommand{\gnewTransp}[1][\paramt_{\nbAtome}]{\ktranspose{\gthetaSymb}_{#1}}
	\newcommand{\gnewtheta}[1][\param]{\gtheta[#1]}
	\newcommand{\gnewthetaTransp}[1][\param]{\ktranspose{\gthetaSymb}_{#1}}
	\newcommand{\supportnew}{\supportt}

	\newcommand{\vecOne}[1][\nbAtome]{{\bf1}_{#1}}
	\newcommand{\vecOneTransp}[1][\nbAtome]{\ktranspose{\bf1}_{#1}}

	\newcommand{\shortcutgkinv}{ ( 1 - \gkthetanewTransp\kinv{\Gk}\gkthetanew{} )^{-1} }
	\newcommand{\Gpolar}{\bfW}

\section{Proof of \texorpdfstring{\Cref{th:MainAbstractTheorem}}{Theorem~\ref{th:MainAbstractTheorem}}}
\label{sec:technical_details}

Let $\support\subseteq \supportt=\{\paramt_\ell\}_{\ell=1}^{\nbAtome}$.~Without loss of generality, we assume that $\support\neq \emptyset$ corresponds to the first $\card(\support)$ elements of $\supportt$, that is $\support=\{\paramt_\ell\}_{\ell=1}^{\card(\support)}$. 

We first notice that, as a direct consequence of \Cref{def:admissible_support}, if $\supportt$ is admissible with respect to $\kernelPaper$ then any $\support\subseteq\supportt$ is also admissible with respect to $\kernelPaper$.
The result stated in \Cref{th:MainAbstractTheorem} is then a direct consequence of \Cref{th:Tropp's-ERC} and the following proposition: 	
\begin{proposition}
	\label{lemma:MainAbstractTheorem}
		Assume kernel $\kernelPaper$ is admissible and $\support=\{\paramt_\ell\}_{\ell=1}^{\card(\support)}$ is admissible with respect to $\kernelPaper$. Then we have that 
		\begin{enumerate}[label=\itshape\roman*)]
			\item the atoms $\{\atome(\paramt_\ell)\}_{\ell=1}^{\nbAtome}$ are linearly independent,
			\label{item:prop-buf-theorem-abstract-condition:linear-indep}
			\item $\kforall[\param\in\paramSet\backslash\support{}] \kvvbar{ \kinv{\Gtrue} \gtheta }_1 < 1$, 
			where 
			\begin{align}
			\Gtrue\element{\ell,\ell'} &\;\triangleq\; \kernelPaper\kparen{\paramt_\ell, \paramt_{\ell'}} &&\forall \ell,\ell'\in\intervint{1}{\card(\support)} \phantom{.}
			\label{eq:MainAbstractTheorem:Gtrue} \\
			\gtheta\element{\ell}  &\;\triangleq\; \kernelPaper\kparen{\param, \paramt_\ell}  && \hphantom{,j}\forall \ell\in\intervint{1}{\card(\support)} \label{eq:MainAbstractTheorem:gtheta}
			.
			\end{align}
			\label{item:prop-buf-theorem-abstract-condition:erc}
		\end{enumerate}
\end{proposition}

\noindent
We thus spend the rest of this section in proving \Cref{lemma:MainAbstractTheorem}. \\

\noindent
\textit{Proof of item~\ref{item:prop-buf-theorem-abstract-condition:linear-indep} of \Cref{lemma:MainAbstractTheorem}}.
Let $\{\coeff_\ell\}_{\ell=1}^{\card(\support)} \subset\kR$ be such that $\Vobs\triangleq\sum_{\ell=1}^{\card(\support)} \coeff_\ell\, \atome(\paramt_\ell) = {\bf0}_{\spaceObs}$, and let $T$ be the set of indices such that $\coeff_\ell\neq0$.
Without loss of generality, we can assume that $\sum_{\ell=1}^{\card(\support)} |\coeff_\ell|<1$. 
We will prove by contradiction that $T$ is empty.
	
Assuming that $T$ is not empty, we first prove that the sign of the coefficients $\{c_{\ell}\}_{\ell\in T}$ cannot be all equal.
To this end,  let us assume (without loss of generality) that $\coeff_\ell>0$ for all $\ell \in T$ and show that  a contradiction occurs with the hypothesis of admissibility of $\support$. 
Since $\Vobs={\bf0}_{\spaceObs}$, the function $\psi:\param \mapsto \scalprod{\atome(\param)}{\Vobs}$ is identically equal to zero. Hence, on the one hand, any point of $\paramSet$ is a maximizer.
On the other hand, since all the elements of $\{c_{\ell}\}_{\ell\in T}$ are positive and $\support$ is (by hypothesis) admissible with respect to $\kernelPaper$, we have from item~\ref{hyp:geometric:1} of \Cref{def:admissible_support} that the maximizers of $\psi$ must belong to $\support$. This implies that $\paramSet{}\subset\support$ which contradicts the definition of $\support$ and $\paramSet$. Therefore, if $T$ is not empty, not all the elements of $\{c_{\ell}\}_{\ell \in T}$ have the same sign.

We can thus partition $T$ into two non-empty disjoint subsets:
\begin{align}
	T_{+} &= \kset{\ell \in T}{c_{\ell}>0},\nonumber\\
	T_{-} &= \kset{\ell \in T}{c_{\ell}<0}. \nonumber
\end{align}
Similarly, we let 
\begin{align}
	\support_{+} &= \kset{\paramt_\ell \in \support}{\ell\in T_+},\nonumber\\
	\support_{-} &= \kset{\paramt_\ell \in \support}{\ell\in T_-}.\nonumber
\end{align}
Since the elements of $\support \supseteq \support_{-} \cup \support_{+}$ are pairwise distinct, we have $\support_{-} \cap \support_{+} = \emptyset$. 
Defining $\Vobs_{+} = \sum_{\ell \in T_{+}} \coeff_\ell\, \atome(\paramt_\ell)$ and $\Vobs_{-} = \sum_{\ell \in T_{-}} \coeff_\ell\, \atome(\paramt_\ell)$, we note that $\Vobs=\Vobs_+ + \Vobs_-$.
Using the fact that $\Vobs={\bf0}_{\spaceObs}$, one deduces that $\Vobs_+=-\Vobs_-$. 
Moreover, $\Vobs_+$ (resp. $-\Vobs_-$) is a positive linear combination of atoms with parameters in $\support_+\subset \support$ (resp. $\support_-\subset \support$).
Therefore, since $\support$ is admissible with respect to $\kernelPaper$, item~\ref{hyp:geometric:1} of \Cref{def:admissible_support} applies and we have that any maximizer of $\psi:\param\mapsto\scalprod{\atome(\param)}{\Vobs_+}=\scalprod{\atome(\param)}{-\Vobs_-}$  must belong to $\support_+\cap\support_-$.
Now, on the one hand, by \Cref{Fact:existence-maximizer}, the set of maximizers of $\psi$ cannot be empty.
On the other hand $\support_+\cap\support_-= \emptyset$.
This leads to a contradiction.
Therefore we must have $T = \emptyset$. In other words, $\Vobs{} = {\bf0}_{\spaceObs}$ implies $\coeff_\ell=0$ $\forall\ell\in\support$, so that the atoms $\{\atome(\paramt_\ell)\}_{\ell=1}^{\card(\support)}$ are linearly independent.

As a consequence of this first part of the proposition, the Gram matrix of any subset of  $\{\atome(\paramt_\ell)\}_{\ell=1}^{\card(\support)}$ is a positive definite matrix, and therefore invertible. 
In particular, the inverse of matrices $\Gtrue$ and $\Gk$ appearing in the second part of the proof is always well-defined. \\

\noindent
\textit{Proof of item~\ref{item:prop-buf-theorem-abstract-condition:erc} of \Cref{lemma:MainAbstractTheorem}}. 
Recall that, as a consequence of \Cref{def:admissible_support}, if $\support$ is admissible with respect to $\kernelPaper$, then any support $\support'\subset\support$ is also admissible.
We thus show our result by induction on the cardinality of $\support'$. 
For notational convenience, we let hereafter $\nbAtome'\triangleq \card(\support')$. 
We prove by induction on $\nbAtome'$ that:
	\begin{enumerate}[label=\textit{\alph*)}]
		\item $\kinv{\Gtrue}\vecOne[\nbAtome']$ has nonnegative entries,
		\label{item:geometric:a}
		\item $\kforall[\param\in\paramSet] \ \kinv{\Gtrue}\gtheta$ has nonnegative entries,
		\label{item:geometric:b}
		\item $\forall \param \in \paramSet\setminus\support'$,\quad $\|\kinv{\Gtrue} \gtheta\|_1 < 1$.
		\label{item:geometric:c}
	\end{enumerate}
\noindent
The quantities $\Gtrue$ and $\gtheta$ appearing above are defined in \eqref{eq:MainAbstractTheorem:Gtrue}-\eqref{eq:MainAbstractTheorem:gtheta} with the substitution $\support'\leftrightarrow \support$. Item~\ref{item:geometric:c} corresponds to result~\ref{item:prop-buf-theorem-abstract-condition:erc} of \Cref{lemma:MainAbstractTheorem}.
Items~\ref{item:geometric:a} and~\ref{item:geometric:b} are intermediate results that allow a subdivision of the proof into steps.

\paragraph{Initialization: $\nbAtome'=1$} 	
In this case, both $\Gtrue$ and $\gtheta$ are scalars.
Since $\kappa$ is admissible, we have $\Gtrue=1$ and $\gtheta \geq 0$ (cf \Cref{hyp:geometric:0}).
Therefore, items~\ref{item:geometric:a} and~\ref{item:geometric:b} are fulfilled and $\| \kinv{\Gtrue}\gtheta\|_1 = \gtheta = \kernelPaper(\param,\paramt_1)$ and, using \Cref{hyp:geometric:0}-\ref{item:admissible-kernel:decreasing}, we have $\kernelPaper(\param,\paramt_1) < 1$.
Hence, item~\ref{item:geometric:c} is also true.

\paragraph{Induction: $1 < \nbAtome'\leq \nbAtome$}
We assume items \ref{item:geometric:a}-\ref{item:geometric:b}-\ref{item:geometric:c} hold for any $\support'\subset\support$ of cardinality $\nbAtome'-1 \geq 1$.
Considering $\support'\subseteq \support$ an arbitrary support of size $\nbAtome'$, we show that items \ref{item:geometric:a}-\ref{item:geometric:b}-\ref{item:geometric:c} also hold for $\support'$.
Without loss of generality, we will assume that $\support'$ corresponds to the first $\card{(\support')}$ elements of $\support$, that is $\support'=\{\paramt_\ell\}_{\ell=1}^{\nbAtome'}$.

\noindent
We consider $\supportk = \{\paramt_\ell\}_{\ell=1}^{\nbAtome'-1} \subset \support'$ and use over-lined notations for quantities related to $\supportk$: we denote by $\Gk\in\kR^{(\nbAtome'-1)\times(\nbAtome'-1)}$, $\gktheta\in\kR^{\nbAtome'-1}$ the quantities defined in~\eqref{eq:MainAbstractTheorem:Gtrue}-\eqref{eq:MainAbstractTheorem:gtheta} for $\supportk$ and by $\Gnew{}\in\kR^{\nbAtome'\times\nbAtome'},\gnewtheta\in\kR^{\nbAtome'}$ the same quantities for $\support'$.
Likewise, the notations $\gktheta[\ell]\in\kR^{\nbAtome'-1}, \gnew[\ell']\in\kR^{\nbAtome'}$ for $\ell=1\ldots\nbAtome'-1$, $\ell'=1\ldots\nbAtome'$ will refer to the columns of $\Gk$ and $\Gnew$, respectively.
With these notations we have:
\begin{align}
	\gnewtheta \;=\;&
	\begin{pmatrix}
		\gktheta \\  \kernelPaper(\param,\paramt_{\nbAtome'})
	\end{pmatrix}
	\in \kR^{\nbAtome'} \quad \kforall \param\in\paramSet
	\label{eq:geometric:proof:recurrence:defgnew} \\
	\Gnew \;=\;&
	\begin{pmatrix}
		\Gk           & \gkthetanew \\
		\gkthetanewTransp   & 1 
	\end{pmatrix}
	\in \kR^{\nbAtome'\times\nbAtome'}
	\label{eq:geometric:proof:recurrence:defGnew}
\end{align}
where we denote $\gkthetanew\triangleq\gktheta[{\param_{k'}^\star}]$ for notational convenience. 
We note that, as mentioned above, item~\ref{item:prop-buf-theorem-abstract-condition:linear-indep} of \Cref{lemma:MainAbstractTheorem} ensures that both $\Gnew$ and $\Gk$ are invertible.

\paragraph{Item~\ref{item:geometric:a}}
We show that the last entry of $\bfu \triangleq \kinv{\Gnew}{\bf1}_{\nbAtome'}$ is positive.
Since the reasoning holds for any ordering of the $\paramt_\ell$'s, we then deduce that all the entries of $\bfu$ are positive.
Block inversion results \citep[Cor.~2.8.9]{Bernstein2005} give
\begin{equation}
	\label{eq:geometric:proof:recurrence:GnewInv}
	\kinv{\Gnew} \;=\;
	\begin{pmatrix}
		\kinv{\Gk} + s \kinv{\Gk} \gkthetanew{} \gkthetanewTransp \kinv{\Gk}
		& - s\kinv{\Gk} \gkthetanew{} \\
		- s\gkthetanewTransp \kinv{\Gk} &  s
	\end{pmatrix}
	,
\end{equation}
where $s \triangleq  \shortcutgkinv{}$.
Notice that
\begin{align}
	\gkthetanewTransp\kinv{\Gk}\gkthetanew{}
	\;\leq \;
	& 
	\kvvbar{ \gkthetanew{} }_\infty \, \kvvbar{ \kinv{\Gk}\gkthetanew{} }_1
	\nonumber \\
	\;\leq \;
	& 
	\kvvbar{ \kinv{\Gk}\gkthetanew{} }_1 
	\nonumber \\
	\;<\;
	& 1.
	\label{eq:geometric:proof:recurrence:showThat_s_isPositive2}
\end{align}
The first inequality is a consequence of Hölder's inequality, the second of \Cref{hyp:geometric:0} and the third follows from induction hypothesis \ref{item:geometric:c}. 
Hence $s>0$.

\noindent
The last entry of $\bfu = \kinv{\Gnew}{\bf1}_{\nbAtome'}$ now writes $\bfu\element{\nbAtome'} \;=\; s ( 1 - \gkthetanewTransp \kinv{\Gk} \vone{}_{\nbAtome'-1})$.
By induction hypothesis~\ref{item:geometric:b}, we have $\|\kinv{\Gk} \gkthetanew{}\|_1 =\gkthetanewTransp \kinv{\Gk} \vone{}_{\nbAtome'-1}$.
Using~\eqref{eq:geometric:proof:recurrence:showThat_s_isPositive2} and the fact that $s>0$, we thus have $\bfu\element{\nbAtome'}>0$.

\paragraph{Item~\ref{item:geometric:b}}
We first show that the last entry of $\bfv \triangleq \kinv{\Gnew}\gnewtheta$ is non-negative.
Given the decomposition of $\Gnew$ in~\eqref{eq:geometric:proof:recurrence:GnewInv}, the last entry of $\kinv{\Gnew}\gnewtheta$ writes
\begin{equation}
	\bfv\element{\nbAtome'}
	= 
	s \kparen{ \gnewtheta\element{\nbAtome'}  - \gkthetanewTransp \kinv{\Gk} \gktheta }
	=
	s \kparen{ \kernelPaper(\param,\paramt_{\nbAtome'})  - \gkthetanewTransp \kinv{\Gk} \gktheta }
	.
\end{equation}

Since $s>0$ (see~\eqref{eq:geometric:proof:recurrence:showThat_s_isPositive2}) it is then sufficient to show that $\kernelPaper(\param, \paramt_{\nbAtome'}) - \ktranspose{\overline{\bfv}}\gktheta{} \geq 0$, where $\overline{\bfv}\triangleq \kinv{\Gk}\gkthetanew$, in order to show that $\bfv\element{\nbAtome'}\geq0$.
This will be achieved by studying this quantity seen as a function of $\param$.
Consider $T \subseteq \intervint{1}{\nbAtome'-1}$ the (possibly empty) set defined by $T\triangleq\kset{\ell}{\overline{\bfv}\element{\ell}\neq0}$ and define 
\begin{equation}
    \kfuncdef{\condSumPositivKernelSymb_1}{
	    \paramSet}{\kR+}[
		\param][\ktranspose{\overline{\bfv}}\gktheta{}
	= 
	\sum_{\ell=1}^{\nbAtome'-1} \overline{\bfv}[\ell] \kernelPaper(\param, \paramt_{\ell})
	= \sum_{\ell \in T} \overline{\bfv}[\ell] \kernelPaper(\param,\paramt_{\ell}).
	]
\end{equation}
Notice that:
\begin{itemize}
    \item $\condSumPositivKernelSymb_1(\param) = \gkthetanewTransp \kinv{\Gk} \gktheta{}$,
    \item the entries of $\overline{\bfv}$ are nonnegative by the induction hypothesis~\ref{item:geometric:b}.
	Moreover, from induction hypothesis~\ref{item:geometric:c}, we have $\sum_{\ell=1}^{\nbAtome'-1}\overline{\bfv}[\ell]= \|\kinv{\Gk}\gkthetanew\|_1<1$,
	\item for $j \in \intervint{1}{\nbAtome'-1}$ and $\param = \paramt_{j}$ we have $\gktheta = \gktheta[j] = \Gk \overline{\bfe}_{j}$, where $\overline{\bfe}_j$ is the $j$-th canonical vector of $\kR^{\nbAtome'-1}$.
	Hence, $\condSumPositivKernelSymb_1(\paramt_j) = \gkthetanewTransp \kinv{\Gk}\gktheta[j] = \gkthetanewTransp \overline{\bfe}_j = \kernelPaper(\paramt_j,\paramt_{\nbAtome'})$ and $\condSumPositivKernelSymb_1(\paramt_j) - \kernelPaper(\paramt_j,\paramt_{\nbAtome'}) = 0$ $\forall j\neq \nbAtome'$.
\end{itemize}
Since $\support$ is admissible with respect to $\kernelPaper$, we can apply item~\ref{hyp:geometric:2} of \Cref{def:admissible_support} with $\ell = \nbAtome'$ to any $\emptyset \neq T \subseteq \intervint{1}{\nbAtome'-1}$. This leads to: 
\begin{equation}
	\label{eq:proof:th2:eq-related-to-item-ii}
	\kernelPaper(\param, \paramt_{\nbAtome'}) - \ktranspose{\overline{\bfv}}\gktheta{} 
	= \kernelPaper(\param,\paramt_{\nbAtome'})  - \condSumPositivKernelSymb_1(\param)
	\geq 0
\end{equation}
for all $\param\in\Theta$.
The same obviously holds if $T$ is empty as $\condSumPositivKernelSymb_1(\param)$ is identically zero and the admissibility of $\kernelPaper$ implies that it is nonnegative (see \Cref{hyp:geometric:0}).
Since this result does not depend on the ordering of the $\paramt_\ell$'s, we can finally conclude that all the elements of $\kinv{\Gtrue}\gnewtheta$ are nonnegative.

\paragraph{Item~\ref{item:geometric:c}}
Let
\begin{equation}
	\kfuncdef[m]{\condSumPositivKernelSymb_2}{\paramSet{}}{\bbR}[\param][\kvvbar{\kinv{\Gnew} \gnewtheta{}}_1].
\end{equation}
We need to prove that $\condSumPositivKernelSymb_2(\param) < 1$ for all $\param \notin \support'$.

From item~\ref{item:geometric:b}, we know that $\kinv{\Gnew} \gnewtheta{}$ has nonnegative entries, so that $\|\kinv{\Gnew} \gnewtheta{}\|_1 = \ktranspose{\bf1}_{\nbAtome'} \kinv{\Gnew} \gnewtheta{}$. 
Letting $\bfu \triangleq \kinv{\Gnew}{\bf1}_{\nbAtome'}\in\kR^{\nbAtome'}$, the function $\condSumPositivKernelSymb_2(\param)$ can then be written as
\begin{equation}
	\condSumPositivKernelSymb_2(\param) \triangleq \sum_{\ell=1}^{\nbAtome'} \bfu[\ell] \kernelPaper\kparen{\param, \paramt_\ell} .
\end{equation}
Moreover we have $\bfu[\ell]\geq 0$ $\forall \ell$ since we showed in item~\ref{item:geometric:a} that $\kinv{\Gnew}{\bf1}_{\nbAtome'} = \bfu$ has nonnegative entries. 
We also note that $\bfu\neq {\bf0}_{\nbAtome'}$ since $\Gnew \bfu = {\bf1}_{\nbAtome'}$.

Applying item~\ref{hyp:geometric:1} of~\Cref{def:admissible_support} together with the comment in \Cref{footnote:positive-combili}, we have that the maximizer of $\condSumPositivKernelSymb_2(\param)$ must belong to $\kset{\paramt_\ell}{\bfu\element{\ell}\neq0}  \subseteq \support'$. 
Now,
\begin{equation}
	\kforall[j\in\intervint{1}{\nbAtome'}]\quad \psi_2(\paramt_{j})=\ktranspose{\bf1}_{\nbAtome'}\kinv{\Gnew}\gnewtheta[\paramt_{j}] = \ktranspose{\boldsymbol{1}}_{\nbAtome'} \bfe_{j} = 1.
\end{equation}	
Therefore, $\condSumPositivKernelSymb_2(\param) < 1$ for all $\param \notin\support'$.

\conditionalPagebreak


\section{Proofs related to CMF dictionaries}
	\label{sec:app:proofs-CMF}

This appendix contains the proofs of the results related to CMF dictionaries presented in \Cref{sec:CMF_dico,subsec:contrib:ksparserecov}. We first state and prove a technical lemma which will used in the proofs of \Cref{lemma:pd_CMF} and \Cref{th:contrib:cmf_uniformRecov_1D}:

\begin{lemma}\label{lemma:strictCMF}
	Let $\cmf$ be a CMF such that $\cmf(0) = 1$ and $\lim_{t \to \infty}\cmf(t) = 0$. Then the Borel measure $\nu$ appearing in the integral representation of the CMF in \Cref{lemma:geometric:cmf:integral_representation} is nonzero and satisfies $\nu(\{0\})=0$ and $\nu(\kR+^{*})=1$. Moreover, $\cmf$ is strictly positive and strictly decreasing on $\kR+$, with $\cmf[][(1)](t)<0$ on $\kR_{+}^{*}$.
\end{lemma}
\begin{proof}
	By the integral representation of \Cref{lemma:geometric:cmf:integral_representation} we have $\nu(\kR_{+}) = \cmf(0)=1$ hence $\nu$ is nonzero. Moreover $\cmf(x) = \int_{0}^{+\infty} \cste^{-ux} d\nu(u) \geq \nu(\{0\}) \geq 0$  for each $x \geq 0$.
	As $\lim_{x \to +\infty} \varphi(x)=0$, it follows that $\nu(\{0\})=0$ and therefore $\nu(\kR+^{*})=1$.
	This proves the first part of the statement. 

	The positivy of $\varphi$ follows from the fact that $\nu$ is nonzero and $\cste^{-ux}>0$ for all $u,x \geq 0$.
	Hence, the integral representation \eqref{eq:geometric:cmf:integral_representation} of $\varphi$ yields $\cmf(x)>0$ for each $x>0$.
	Finally, we prove by contradiction that $\cmf[][(1)](t)<0$ on $\kR_{+}$. 

	Assume the existence of $t_{0}>0$ such that $\cmf[][(1)](t_{0})=0$.
	As $\cmf[][(1)]$ is continuous and non-decreasing on $\kR_{+}^{*}$ with $\cmf[][(1)](t)\leq 0$ for each $t>0$, it follows that $\cmf[][(1)](t)=0$ for each $t \geq t_{0}$, hence $\cmf(t) = \cmf(t_{0})$ for each $t \geq t_{0}$. 
	As we have just seen, we have $\cmf(t_{0})>0$, hence this contradicts the assumption $\lim_{t \to +\infty} \cmf(t) = 0$.
\end{proof}

\conditionalPagebreak

\subsection{Proof of \texorpdfstring{\Cref{lemma:pd_CMF}}{Lemma~\ref{lemma:pd_CMF}}}
\label{sec:app:cmf_related}

\newcommand{\varMuette}{\omega}{}

The outline of the proof is as follows.
We first show that for any $0<p\leq1$ and $\boldsymbol{\omega}\in\kR^{\dimParam}$, the quantity $\cste^{-\varIntCmf\|\boldsymbol{\omega}\|^p_p}$ is related to the characteristic function of some $\dimParam$-dimensional random vector $\bfZ_u$.
We then use this formulation of $\cste^{-\varIntCmf\|\boldsymbol{\omega}\|^p_p}$ as a characteristic function together with the Bernstein-Widder representation of CMFs (see \Cref{lemma:geometric:cmf:integral_representation}) to show that the function $\varphi(\|\boldsymbol{\omega}\|_p^p)$ is positive definite.

\paragraph{Proof that $\cste^{-\varIntCmf\|\boldsymbol{\omega}\|^p_p}$ is the characteristic function of some random vector $\bfZ_u$}
In probability theory, the characteristic function of a real-valued random vector $\bfZ \in \kR^{\dimParam}$ is the function  $\boldsymbol{\omega} \in \kR^{\dimParam}\mapsto \bbE_{\bfZ}[\cste^{\csti \ktranspose{\boldsymbol{\omega}} \bfZ}]$ where $\bbE$ denotes the expectation operator and $\csti$ is the imaginary number.

First we consider for any $u \geq 0$ the scalar-valued function 
\begin{equation*}
	\kfuncdef[m]{\Phi_u}{\kR}{\kR}[\varMuette][\cste^{-\varIntCmf \kvbar{\varMuette}^p}]
\end{equation*}
and show that for $u >0$ it is the characteristic function of some random variable $Z_u$ which admits a density with respect to the Lebesgue measure.
Our proof leverages a result due to P\'olya \protect{\cite[Th.~1]{polya1949}}.
We reproduce this result hereafter for self-containedness of the paper: 
\begin{theorem}
	\label{th:app:cmf:polya_theorem}
	Let $\Phi$ be a real-valued function defined on $\kR$ and such that:
	\begin{itemize}
		\item $\Phi$ is continuous and even, 
		\item $\Phi$ is  convex on $\kR^{*}_{+}$,
		\item $\Phi{}(0) = 1$, 
		\item $\displaystyle \lim_{\varMuette\to+\infty} \Phi(\varMuette) = 0$.
	\end{itemize}
	Then, $\Phi$ is the characteristic function of some random variable which admits a density with respect to the Lebesgue measure. Moreover this density is even and continuous everywhere, except possibly at zero. 
\end{theorem}

\noindent
Observe that $\Phi_{u}$ is even, continuous and verifies $\Phi_u(0) = 1$ and $\lim_{\varMuette\to+\infty} \Phi_u(\varMuette) = 0$ since $u>0$. 
Moreover, for $p\in\kintervoc{0}{1}$ and $\varMuette>0$, its second derivative on $\kR+$ is 
\begin{equation*}
	\Phi_u^{(2)}(\varMuette) = \varIntCmf p\,\varMuette^{p-2} \kparen{(1-p) + \varIntCmf p\varMuette^p }\cste^{-\varIntCmf \varMuette^p} .
\end{equation*}
Hence, $\Phi_u^{(2)}(\varMuette) > 0$ for all $\varMuette > 0$ and $\Phi_u$ is convex on $\kR^{*}_+$. 
As a consequence, $\Phi_{u}$ satisfies the assumptions of \Cref{th:app:cmf:polya_theorem} and it is the characteristic function of some scalar random variable $Z_u$ which admits a (continuous, except possibly at zero) density with respect to the Lebesgue measure, that is
\begin{equation*} 
	\Phi_u(\varMuette)  = \bbE_{Z_{\varIntCmf}} \kbracket{ \cste^{\csti \varMuette Z_{u}} } 
	.
\end{equation*} 

We are now ready to show that the function $\cste^{-\varIntCmf\|\boldsymbol{\omega}\|^p_p}$ is the characteristic function of some random vector $\bfZ_u$.
To this end, let us define the random vector $\bfZ_u = \inlinevec{Z_u^1,\dotsc,Z_u^\dimParam}$ as the concatenation of $\dimParam$ independent copies of $Z_u$.
We thus have 
\begin{equation*}
	\begin{split}
		\bbE_{\bfZ_u} \kbracket{ \cste^{ \csti  \ktranspose{\boldsymbol{\varMuette}} \bfZ_u} }
		=
		\prod_{d=1}^\dimParam \bbE_{Z_u^d} \kbracket{ \cste^{ \csti \boldsymbol{\varMuette}\element{d} Z_u^d} }
		=
		\prod_{d=1}^\dimParam \Phi_u(\boldsymbol{\varMuette}\element{d})
		=
		\cste^{-\varIntCmf \kvvbar{\boldsymbol{\varMuette}}_p^p}.
	\end{split}
\end{equation*}
We note that for all $u>0$ and $\varMuette \in \kR$ we have $\Phi_u(\varMuette) = \Phi_1(u^{{1}/{p}}\varMuette)$. 
Hence, the function $\Phi_u(\varMuette)$ can also be written as an expectation with respect to the random variable $Z_1$ for all $u>0$ and $\varMuette \in \kR$: 
\begin{equation}
	\label{eq:exp=expectation_wrt_Z1}
	\Phi_u(\varMuette) = \bbE_{Z_1} \kbracket{ \cste^{\csti u^{1/p}\varMuette Z_1} }
	.
\end{equation}
\Cref{eq:exp=expectation_wrt_Z1} obviously also holds for $u=0$ since both sides of the equality are equal to 1 in that case. Using \eqref{eq:exp=expectation_wrt_Z1}, we obtain
\begin{equation}
	\label{eq:proof_cmf_kernel:exp_to_characteristicFunc}
	\bbE_{\bfZ_1} \kbracket{ \cste^{\csti  u^{1/p}\ktranspose{\boldsymbol{\varMuette}} \bfZ_1} }
	= \cste^{-\varIntCmf \kvvbar{\boldsymbol{\varMuette}}_p^p} \qquad \forall u\geq 0,
\end{equation}
where $\bfZ_1 = \inlinevec{Z_1^1,\dotsc,Z_1^\dimParam}$ is the concatenation of $\dimParam$ independent copies of $Z_1$. 
We will use the latter representation in the second part of the proof.

\paragraph{Proof that $\varphi(\|\boldsymbol{\varMuette}\|_p^p)$ is a positive definite function}
We want to show that for any $\nbAtome\in\kN$, any $\{\param_\ell\}_{\ell=1}^\nbAtome\subset\kR^\dimParam$ and any $\coeffv\in\kC^{\nbAtome}\setminus\kbrace{{\bf0}_{\nbAtome}}$, we have 
\begin{equation*}
	{\coeffv}^H\Gtrue\coeffv > 0,
\end{equation*}
where $(\cdot)^H$ denotes the conjugate transpose operator and 
\begin{equation*}
	\Gtrue\element{\ell, \ell'} \triangleq 
	\varphi(\|\param_{\ell'}-\param_\ell\|_p^p)
	\quad \forall\ell,\ell'\in\intervint{1}{\nbAtome}. 
\end{equation*}
Note that in practice this will only be used for real-valued coefficients, but the result is established for complex-values $\coeffv$ to fit with the standard definition of a positive definite function. 
Since $\cmf$ is a CMF, \Cref{lemma:geometric:cmf:integral_representation} ensures the existence of a non-negative finite  Borel measure $\nu$ such that for all $\boldsymbol{\varMuette}\in\kR^\dimParam$:
\begin{equation*}
	\varphi(\|\boldsymbol{\varMuette}\|_p^p) = \int_0^{+\infty} \cste^{-\varIntCmf \kvvbar{\boldsymbol{\varMuette}}_p^p } \diff\nu(\varIntCmf).
\end{equation*}
Hence,
\begin{align}
	\coeffv^H\Gtrue\coeffv
	\;\overset{\hphantom{\eqref{eq:proof_cmf_kernel:exp_to_characteristicFunc}}}{=}\;& 
	\sum_{\ell=1}^{\nbAtome} \sum_{\ell'=1}^{\nbAtome} \coeffv\element{\ell}^* \coeffv\element{\ell'} \varphi\kparen{\|\param_{\ell'}- \param_\ell\|_p^p}  \nonumber \\
	\;\overset{\hphantom{\eqref{eq:proof_cmf_kernel:exp_to_characteristicFunc}}}{=}\;& 
	\sum_{\ell=1}^{\nbAtome} \sum_{\ell'=1}^{\nbAtome} \coeffv\element{\ell}^* \coeffv\element{\ell'} \int_0^{+\infty} \cste^{-\varIntCmf \kvvbar{\param_{\ell'} - \param_\ell}_p^p } \diff\nu(\varIntCmf) \label{eq:psf_int1}\\
	\;\overset{\hphantom{\eqref{eq:proof_cmf_kernel:exp_to_characteristicFunc}}}{=}\;& 
	\sum_{\ell=1}^{\nbAtome} \sum_{\ell'=1}^{\nbAtome} \coeffv\element{\ell}^* \coeffv\element{\ell'} \int_0^{+\infty} 
	\mathbb{E}_{\bfZ_{1}} \kbracket{ \cste^{\csti{u^{1/p}}\ktranspose{(\param_{\ell'} - \param_\ell)}\bfZ_1} }
	\diff\nu(\varIntCmf) \label{eq:psf_int2}
	,
\end{align}
where the last equality follows \eqref{eq:proof_cmf_kernel:exp_to_characteristicFunc}. 
By linearity of the expectation, it follows:
\begin{align}
	\coeffv^H\Gtrue\coeffv
	\;=\;&
	\int_0^{+\infty} \bbE_{ \bfZ_1} \kbracket{ \sum_{\ell=1}^{\nbAtome} \sum_{\ell'=1}^{\nbAtome} \coeffv\element{\ell}^* \coeffv\element{\ell'}   \cste^{- \csti {u^{1/p}}\ktranspose{\kparen{\param_\ell - \param_{\ell'}}} \bfZ_1 } } \, \diff\nu(\varIntCmf) \nonumber \\
	\;=\;&
	\int_0^{+\infty} \mathbb{E}_{ \bfZ_1} \kbracket{  
		\kparen{\sum_{\ell=1}^{\nbAtome} \coeffv\element\ell \cste^{ \csti {u^{1/p}}\ktranspose{\param}_\ell \bfZ_1 }}^*
		\kparen{\sum_{\ell'=1}^{\nbAtome} \coeffv\element{\ell'} \cste^{ \csti {u^{1/p}}\ktranspose{\param}_{\ell'} \bfZ_1 }}
	} \, \diff\nu(\varIntCmf) \nonumber \\
	\;=\;&
	\int_0^{+\infty} \mathbb{E}_{\bfZ_1} \kbracket{  
		\kvbar{\sum_{\ell=1}^{\nbAtome{}} \coeffv\element\ell \cste^{ \csti {u^{1/p}}\ktranspose{\param}_\ell \bfZ_1 }}^2 
	} \, \diff\nu(\varIntCmf) \label{eq:app:cmf:psdkernelallnorm}
	\ \geq 0.
\end{align}
Since this holds for any $\coeffv\in \kC^\nbAtome$ this shows that $\varphi(\|\cdot\|_p^p)$ is  positive semi-definite. 

To establish that $\varphi(\|\cdot\|_p^p)$ is a positive definite function we now show that the equality in \eqref{eq:app:cmf:psdkernelallnorm} can only occur if $\coeffv = {\bf0}_\nbAtome$. 
To this end, denote $\psi(\bfz) \triangleq |\sum_{\ell=1}^{\nbAtome} \coeffv\element{\ell} \cste^{\csti \ktranspose{\param}_\ell \bfz }|^2$ and $\Psi(u)\triangleq  \mathbb{E}_{\bfZ_1} \kbracket{\psi(u^{1/p}\bfZ_1)}$ for $u\in\kR+$, and assume that equality holds in \eqref{eq:app:cmf:psdkernelallnorm}, that is to say $\int_{0}^{+\infty} \Psi(u) \diff\nu(u) = 0$. We next show that this implies that $\coeffv = {\bf0}_\nbAtome$. 

First, we note that $\Psi$ is continuous since
\begin{align*}
	\Psi(u) 
	&= \mathbb{E}_{\bfZ_1} \kbracket{\psi(u^{1/p}\bfZ_1)}\\
	&= \sum_{\ell=1}^{\nbAtome{}}\sum_{\ell'=1}^{\nbAtome{}} \coeffv\element\ell^* \coeffv\element{\ell'}\,\mathbb{E}_{\bfZ_1} \kbracket{\cste^{ \csti {u^{1/p}}\ktranspose{(\param_{\ell'}-\param_{\ell})} \bfZ_1 }}\\
	&= \sum_{\ell=1}^{\nbAtome{}}\sum_{\ell'=1}^{\nbAtome{}} \coeffv\element\ell^* \coeffv\element{\ell'}\,\cste^{ -\csti u \|\param_\ell-\param_{\ell'}\|_p^p }
\end{align*}
where the last equality follows \eqref{eq:proof_cmf_kernel:exp_to_characteristicFunc}.

Second, since $\cmf$ is a CMF satisfying $\cmf(0) = 1$ and $\lim_{x \to \infty} \cmf(x)=0$, by \Cref{lemma:strictCMF} the non-negative finite Borel measure $\nu$ satisfies $\nu(\kR+^{*})=1$.  
Since $\kR_{+}^{*} = \cup_{n \geq 1} [n,n+1] \bigcup \cup_{n \geq 1} [1/(n+1),1/n]$, there must exist $n\geq 1$ such that either $\nu\left([n,n+1]\right)>0$ or $\nu\left([1/(n+1),1/n]\right)>0$. Without loss of generality consider the case $\nu\left([n,n+1]\right)>0$ (the other one can be treated similarly).~Since $\Psi$ is continuous over the compact set $[n,n+1]$, it attains its infimum over $[n,n+1]$ at some $u_0\in[n,n+1]$. Now, because $\Psi(u)\geq 0$ for every $u \geq 0$ we have
\begin{equation*}
	0 = \int_{0}^{+\infty} \Psi(u) \diff\nu(u) \geq \int_{n}^{n+1} \Psi(u) \diff\nu(u) \geq \Psi(u_0) \nu([n,n+1])
\end{equation*}
and we obtain $\Psi(u_0)=0$ since $\nu([n,n+1])>0$. By construction $u_{0}>0$. 

Finally, we have that:
\begin{itemize}
	\item the distribution of $\bfZ_{1}$ has a density with respect to the Lebesgue measure and its density is continuous, except possibly at points where one coordinate vanishes. 
	\item $\psi$ is nonnegative and continuous (by construction as the squared modulus of a finite linear combination of exponentials).
\end{itemize}
Hence, using the definition of $\Psi(u_0) = \bbE_{\bfZ_1}[\psi(u_0^{1/p}\bfZ_1)]$, we deduce that there exist ${\bfz}_0\in\kR^\dimParam$ (with non-vanishing coordinates) and $r>0$ such that $\psi(\bfz)=0$ $\forall\bfz\in\calB(\bfz_0, r)$, where $\calB(\bfz_0, r)$ is the open ball of radius $r$  centered at $\bfz_0$. 
For any $\bfy\neq {\bf0}_\dimParam$ and $0\leq t\leq {r}/{\|\bfy\|_2}$, we have $\bfz_0+t\bfy \in \calB(\bfz_0,r)$ and therefore
\begin{equation}\label{eq:analyticfunction=0}
	\sum_{\ell=1}^{\nbAtome} \coeffv\element\ell \cste^{\csti \ktranspose{\param}_\ell (\bfz_{0}+t\bfy)}= 0 \quad \forall t \in [{0},{{r}/{\|\bfy\|_2}}].
\end{equation}
If $\bfy$ is such that $\ktranspose{\param}_\ell\bfy \neq \ktranspose{\param}_{\ell'}\bfy$ for all $\ell\neq \ell'$, then the functions $\{t \mapsto \cste^{ \csti t\ktranspose{\param}_\ell\bfy}\}_{\ell=1}^\nbAtome$ are linearly independent on $[{0},{{r}/{\|\bfy\|_2}}]$ and \eqref{eq:analyticfunction=0} holds if and only if $\coeffv={\bf0}_{\nbAtome}$.

It thus remains to show that there exists some $\bfy\in\kR^\dimParam$ such that $\ktranspose{\param}_\ell\bfy \neq \ktranspose{\param}_{\ell'}\bfy$ for all $\ell\neq \ell'$. 
To this end, let us consider the following finite set of vectors: 
\begin{equation}
	\calN\triangleq\kset{\param_\ell - \param_{\ell'}}{\ell,\ell'\in\intervint{1}{\nbAtome}, \ell\neq \ell'}.
\end{equation}
As the parameters $\param_\ell$'s are pairwise distinct, each $\bfn\in\calN$ is nonzero. Denote  $H_\bfn$ the linear hyperplane whose normal vector is $\bfn$, and consider $H\triangleq \cup_{\bfn \in \calN} H_\bfn$.
Since $H$ is the union of a finite number of $\dimParam$-dimensional hyperplanes, $\kR^{\dimParam}\setminus H$ is not empty. 
Consider $\bfy\in\kR^{\dimParam}\setminus H$.
Then, by construction, $\ktranspose{\bfn}\bfy\neq0$ for all $\bfn\in\calN$ and therefore $\ktranspose{\param}_\ell{\bfy} \neq \ktranspose{\param}_{\ell'}{\bfy}$ for all $\ell\neq \ell'$.
This concludes the proof.

\conditionalPagebreak

\subsection{{Proof of \texorpdfstring{\Cref{th:contrib:cmf_uniformRecov_1D}}{Theorem~\ref{th:contrib:cmf_uniformRecov_1D}}}}
	\label{subsubsec:uniformRecovCmfDim1}

{
\newcommand{\maximizer}{\param_m}

{
\renewenvironment{proof}{}{\qed}
\newcommand{\ubreak}{u_\varepsilon}

Our proof leverages \Cref{th:MainAbstractTheorem} by showing that if $\dico$ is a CMF dictionary in dimension 1 with induced kernel $\kernelPaper$, then: 
\begin{enumerate}[label=\textit{\alph*)}]
	\item[\textit{a)}] $\kernelPaper$ is admissible in the sense of \Cref{hyp:geometric:0},
	\item[\textit{b)}] any finite support $\supportt=\{\paramt_\ell\}_{\ell=1}^{\nbAtome}$ is admissible with respect to kernel $\kernelPaper$ in the sense of \Cref{def:admissible_support}. 
\end{enumerate}
The result stated in \Cref{th:contrib:cmf_uniformRecov_1D} is then a direct consequence of \Cref{th:MainAbstractTheorem}.

\noindent
We begin with a more general lemma establishing the claim \textit{a)}:	
\begin{lemma}\label{lemma:CMFkerAdmissible}
	Any CMF kernel $\kernelPaper\in\cmfkernelclassP[D]$, $D \geq 1$, is admissible. 
\end{lemma}

\begin{proof}
	\noindent 
	\textit{Proof.}
	First, since $\kernelPaper\in\cmfkernelclassP[D]$, there exists a CMF $\varphi$ and scalar $p\in\kintervoc{0}{1}$ such that $\varphi(0)=1$ and $\kernelPaper(\param,\param')= \varphi(\|\param-\param'\|_p^p)$ for all $\param,\param'\in\kR^{D}$.
	Hence $\kernelPaper(\param,\param)=\varphi(0)=1$ for all $\param$. Moreover, the function $\param'\mapsto\kernelPaper(\param,\param')$ is continuous since both CMFs and $\ell_p$-norms are continuous.
	Hence $\kernelPaper$ satisfies~\eqref{eq:propertiesKernel}. 
	The fact that $\kernelPaper$ satisfies the vanishing property~\eqref{eq:contribution:evanecent} is a straightforward consequence of the fact that $\kernelPaper(\param,\param') = \cmf(\|\param-\param'\|_{p}^{p})$ and that $\lim_{t\to+\infty}\cmf(t)=0$.
	Finally, we prove item~\ref{item:admissible-kernel:decreasing} of \Cref{hyp:geometric:0}.
	As $\cmf$ satisfies the assumptions of  \Cref{lemma:strictCMF}, it is strictly positive and strictly decreasing.
	This implies $\kernelPaper(\param,\param') > 0$ for any $\param,\param'$. Moreover, if $\param \neq \param'$ then $t = \|\param-\param'\|_{p}^{p}>0$ and $\cmf(t)<\cmf(0)=1$. \vspace{0.4cm}
\end{proof}

The rest of this section is dedicated to the proof of claim \textit{b)}.
Let us consider a non-empty subset of indices $T\subseteq\intervint{1}{\nbAtome}$ with $t=\card(T)$.
Without loss of generality (up to some reordering), we assume that $T=\intervint{1}{t}$.
Let $\{\coeff_\ell\}_{\ell=1}^t\subset \kR*+$ be such that
\begin{equation}
	\label{eq:proof:1D:def_sum_coeff}
	s\triangleq \sum_{\ell=1}^{t} \coeff_\ell < 1
\end{equation}
and consider the function
\begin{equation}
	\label{eq:proof:1D:defFunCond2}
	\kfuncdef[m]{\condSumPositivKernelSymb}{\kR}{\kR+}[\param][\displaystyle\sum_{\ell=1}^{t} \coeff_\ell\, \kernelPaper(\param, \paramt_\ell).]
\end{equation}
Using the integral formulation of CMF (see \Cref{lemma:geometric:cmf:integral_representation}), we have that $\condSumPositivKernelSymb$ is twice differentiable at any $\param\in\kR\setminus \{\paramt_\ell\}_{\ell=1}^{t}$ \cite[proof of Th.~12a]{Widder1941} and its second derivative writes:
\begin{align}
	\condSumPositivKernelSymb^{(2)}\kparen{\param} = 
		\sum_{\ell=1}^{t} p(1-p) \coeff_\ell&\int_{0}^{+\infty} \varIntCmf \kvbar{\paramt_\ell - \param}^{p-2} \cste^{-\varIntCmf \kvbar{\paramt_\ell - \param}^p} \diff\nu(\varIntCmf)\nonumber		  \\
		\;+\;
		p^2 \coeff_\ell &\int_{0}^{+\infty} \varIntCmf^2 \kvbar{\paramt_\ell - \param}^{2(p-1)} \cste^{-\varIntCmf\kvbar{\paramt_\ell - \param}^p} \diff\nu(\varIntCmf)\label{eq:proof:1D:geometric1}
\end{align}			
We next show that items~\ref{hyp:geometric:1} and~\ref{hyp:geometric:2} of \Cref{def:admissible_support} hold.

\paragraph{Item~\ref{hyp:geometric:1} of \Cref{def:admissible_support}}
First, the vanishing property~\eqref{eq:contribution:evanecent} of admissible kernels ensures that $\condSumPositivKernelSymb$ admits at least one maximizer (see \Cref{Fact:existence-maximizer}).
We then show that any maximizer of $\condSumPositivKernelSymb$ must necessarily belong to $\{\paramt_\ell\}_{\ell=1}^{t}$.

Since $\condSumPositivKernelSymb$ is twice continuously differentiable on $\kR\setminus \{\paramt_\ell\}_{\ell=1}^{t}$, any maximizer $\maximizer\in\kR\setminus \{\paramt_\ell\}_{\ell=1}^{t}$ must verify the second-order optimality condition ``$\condSumPositivKernelSymb^{(2)}(\maximizer)\leq 0$''. 
Now this condition can never be fulfilled for $\maximizer\in\kR\setminus\{\paramt_\ell\}_{\ell=1}^{t}$.
Indeed, we see that each integral term in~\eqref{eq:proof:1D:geometric1} is positive since $\maximizer \notin \{\paramt_\ell\}_{\ell=1}^{t}$ and, by \Cref{lemma:strictCMF}, $\nu(\kR+^{*})>0$. 
Since $p\in\kintervoc{0}{1}$ and $\coeff_{\ell}>0$ $\forall\ell$, it follows that $\condSumPositivKernelSymb^{(2)}\kparen{\maximizer} > 0$.

\paragraph{Item~\ref{hyp:geometric:2} of \Cref{def:admissible_support}}
\newcommand{\funcCondThree}{\condOneNegKernelSymb}
Assume $T=\intervint{1}{t}\neq\intervint{1}{\nbAtome}$ and consider $t'\in\intervint{1}{\nbAtome}\setminus T$.
Without loss of generality, suppose that $t'=t+1$ and consider 	
\begin{equation}
	\label{eq:proof:th3:eq-related-to-item-ii}
	\kfuncdef[m]{\funcCondThree}{\kR}{\kR}[\param][\displaystyle
		\condSumPositivKernelSymb(\param) - \kernelPaper(\param,\paramt_{t+1}).]
\end{equation}
From the definition of $\condSumPositivKernelSymb$ in~\eqref{eq:proof:1D:defFunCond2}, $\funcCondThree$ writes
\begin{equation}
	\label{eq:proof:rewrite-condOneNegKernelSymb}
	\condOneNegKernelSymb(\param) = \sum_{\ell=1}^{t} \coeff_\ell \cmf(\kvbar{\param - \paramt_\ell}_p^p) - \cmf(\kvbar{\param - \paramt_{t+1}}_p^p)
	.
\end{equation}
Assume that $\funcCondThree(\paramt_{\ell}) \leq 0$ for every $\ell\in\intervint{1}{t}$.
We will then show that $\funcCondThree(\param) \leq 0$ for every $\param\in\kR\setminus\{\paramt_\ell\}_{\ell=1}^{t}$ so that item~\ref{hyp:geometric:2} of \Cref{def:admissible_support} is satisfied.

Suppose there exists $\param_0\in\kR$ such that $\phi(\param_0) > 0$ and let us show that this leads to a contradiction.

Let us first emphasize that the existence of some $\param_0\in\kR$ such that $\phi(\param_0) > 0$ implies that a maximizer of $\phi$ exists.
Indeed, we have that $\lim_{\param\to\pm\infty}\phi(\param)=0$ since 
\begin{equation}
	\kvbar{\phi(\param)}\leq \sum_{\ell=1}^t c_\ell\kernelPaper(\param, \paramt_\ell) + \kernelPaper(\param,\paramt_{t+1})
\end{equation}
and $\kernelPaper$ obeys the vanishing property \eqref{eq:contribution:evanecent} by hypothesis.
Hence, for any $0<\varepsilon<\phi(\param_0)$ there exists a compact set $K_\varepsilon$ such that
\begin{equation}
    \forall\param\in K_\varepsilon^c, \quad
	\phi(\param) \leq \varepsilon < \sup_{\param'\in K_\varepsilon} \phi(\param')
	,
\end{equation}
because $\phi$ is continuous. 
The extreme value theorem~\cite[Prop. A.8]{Bertsekas99ed2} then states that (at least) one maximizer of $\phi$, say $\maximizer$, exists and $\maximizer\in K_\varepsilon$.
We note that by definition $\phi(\maximizer)\geq \phi(\param_0)>0$.
We show below that we also must necessarily have $\phi(\maximizer)\leq 0$.
This leads to the desired contradiction and proves the result.

Let $\lambda_\ell \triangleq \kvbar{\paramt_\ell - \maximizer}$ and assume without loss of generality that 
\begin{equation}\label{eq:ordering_lambdas}
	\lambda_1 \leq \dots \leq \lambda_t. 
\end{equation}
We have that $\lambda_1>0$ because $\maximizer\notin\{\paramt_\ell\}_{\ell=1}^t$ since $\phi(\paramt_\ell)\leq 0$ $\forall\ell\in\intervint{1}{t}$ (by assumption) and $\phi(\maximizer)\geq \phi(\param_0)>0$.
We next show that the working assumptions also imply $\phi(\maximizer)\leq 0$ by distinguishing between three cases:

\vspace*{1em}
\noindent
\textit{{\bf Case 1}\textnormal{:} $\lambda_{t+1} \leq \lambda_1$.}
From \eqref{eq:ordering_lambdas}, we have:
\begin{equation}
	\forall u\geq 0,\quad \max_{1\leq \ell \leq t} \cste^{-\varIntCmf\lambda_\ell^p} \leq \cste^{-\varIntCmf\lambda_{t+1}^p}
	.
\end{equation}
Hence 
\begin{equation*}
	\funcCondThree(\maximizer) 
	\leq 
	\underbrace{\kparen{ \sum_{\ell=1}^{t} \coeff_\ell - 1}}_{<0 \text{ by hyp.}}
	\underbrace{\vphantom{\kparen{\sum_{\ell=1}^{t-1}}}\int_{0}^{+\infty} \cste^{-\varIntCmf\lambda_{t+1}^{p}}  \diff\nu(\varIntCmf)}_{>0}\quad <0.
\end{equation*}

\vspace*{1em}
\noindent
\textit{{\bf Case 2}\textnormal{:} $\lambda_{t+1} > \lambda_t$.}
We rely on the following technical lemma that exploits the notion of \say{sign changes of a finite sequence}.
This notion is defined as the number of times two consecutive elements of the finite sequence have opposite signs.  
For instance, the sequence $(1, 1, -1, 1)$ has two sign changes (respectively at the third and fourth positions).
\begin{lemma}
	\label{corollary:function_analysis:polynome_one_sign_change}
	Let $P(\varIntCmf) \triangleq \sum_{\ell=1}^{\nbAtome} \coeff_\ell\, \cste^{-\lambda_\ell \varIntCmf}$ be an exponential polynomial on $\kR+$ with $0 < \lambda_1<\ldots<\lambda_{\nbAtome}$ and $\{\coeff_\ell\}_{\ell=1}^{\nbAtome}\subset\kR*$.
	Assume that:
	\begin{itemize}
		\item the sequence $\coeff_1,\ldots, \coeff_{\nbAtome{}}$ has at most two sign changes;
		\item $P(0) < 0$ and $\displaystyle \lim_{\varIntCmf \rightarrow+\infty} P(\varIntCmf) = 0_+$.
	\end{itemize}
	Then there exists $\varIntCmf_0 > 0$ for which the following inequality holds
	\begin{equation}
		\int_{0}^{+\infty} f(\varIntCmf) P(\varIntCmf) \diff\nu(\varIntCmf)
		\geq
		f(\varIntCmf_0) \int_{0}^{+\infty} P(\varIntCmf) \diff\nu(\varIntCmf)
		\label{eq:corollary:function_analysis:polynome_one_sign_change:1}
		,
	\end{equation}
	for any non-decreasing function $f$ on $\kR+$ and any (unsigned) finite Borel measure $\nu$ on $\kR+$ such that the integrals converge.
\end{lemma}
\noindent
The proof of the lemma is postponed to~\Cref{subsec:app:analysis}.

\vspace*{1em}

As mentioned previously, we have on the one hand that $\maximizer \notin \{\paramt_\ell\}_{\ell=1}^{t}$. On the other hand, $\maximizer\neq\paramt_{t+1}$ because $|\maximizer-\paramt_{t+1}|=\lambda_{t+1}>\lambda_1>0$.
Therefore, $\maximizer\in\kR\setminus\{\paramt_\ell\}_{\ell=1}^{t+1}$.
Since $\phi$ is twice continuously differentiable on $\kR\setminus\{\paramt_\ell\}_{\ell=1}^{t+1}$, $\maximizer$ must necessarily verify the following second-order optimality condition:
\begin{equation}\label{eq:sooc}
	\condOneNegKernelSymb^{(2)}(\maximizer) \leq 0. 
\end{equation}
We next show that $C \condOneNegKernelSymb(\maximizer) < \condOneNegKernelSymb^{(2)}(\maximizer)$ for some positive constant $C>0$.
Hence, in view of \eqref{eq:sooc}, this leads to the desired contradiction: $\condOneNegKernelSymb(\maximizer) < 0$.

Assume first that $\lambda_1 < \dots < \lambda_t$ (the equality cases will be addressed later). 
From~\eqref{eq:proof:1D:geometric1} we have that $\funcCondThree^{(2)}=\funcCondThree_1^{(2)} + \funcCondThree_2^{(2)}$ with
\begin{align*}		%
	\funcCondThree_1^{(2)}(\maximizer) \;=\;&
	p(1-p) \int_{0}^{+\infty} \hspace*{-.0em} \varIntCmf \kparen{ \sum_{\ell=1}^{t} \coeff_\ell \lambda_\ell^{p-2} \cste^{- \varIntCmf \lambda_\ell^p} - \lambda_{t+1}^{p-2} \cste^{- \varIntCmf  \lambda_{t+1}^p} } \diff\nu(\varIntCmf)\\
	\funcCondThree_2^{(2)}(\maximizer) \;=\;& 
	p^2 \int_{0}^{+\infty} \hspace*{-0em} \varIntCmf^2 \kparen{ \sum_{\ell=1}^{t} \coeff_\ell \lambda_\ell^{2(p-1)} \cste^{- \varIntCmf \lambda_\ell^p} -  \lambda_{t+1}^{2(p-1)} \cste^{-\varIntCmf \lambda_{t+1}^p} } \diff\nu(\varIntCmf)
	.
\end{align*}
Using $\lambda_{t+1}>\lambda_\ell > 0$  $\forall \ell \in \intervint{1}{t}$, we obtain
\begin{align}
	\funcCondThree_2^{(2)}(\maximizer) \;=\;&
	\frac{p^2}{ \lambda_{t+1}^{2(1-p)} } \int_{0}^{+\infty} 
	\hspace*{-1em} \varIntCmf^2 \kparen{ \sum_{\ell=1}^{t} \coeff_\ell \kparen{\frac{\lambda_{t+1}}{\lambda_\ell}}^{2(1-p)} \cste^{- \varIntCmf \lambda_\ell^p} 
	- \cste^{-\varIntCmf \lambda_{t+1}^p} } \diff\nu(\varIntCmf)
	\nonumber \\
	\;>\;&
	\frac{p^2}{ \lambda_{t+1}^{2(1-p)} }
	\int_{0}^{+\infty} \hspace*{-.5em} \varIntCmf^2 \kparen{ \sum_{\ell=1}^{t} \coeff_\ell  \cste^{-\varIntCmf \lambda_\ell^p} - \cste^{-\varIntCmf \lambda_{t+1}^p} } \diff\nu(\varIntCmf)
	\label{eq:kernel_exact_recovery:proof_laplace_h_prime_positif:normp}
	.
\end{align}
Note now that:
\begin{itemize}
	\item the function $u\longmapsto u^2$ is increasing,
	\item $P(\varIntCmf) \triangleq \sum_{\ell=1}^{t} \coeff_\ell \cste^{-\varIntCmf\lambda_\ell^{p}} -  \cste^{-\varIntCmf\lambda_{t+1}^{p}}$ is an exponential polynomial with $0<\lambda_1<\dots<\lambda_t<\lambda_{t+1}$ and whose sequence of coefficients is $(\coeff_1,\dotsc,\coeff_{t},-1)$ and has exactly one sign change.
	\item As $\max_{1 \leq \ell \leq t} \lambda_\ell < \lambda_{t+1}$ by hypothesis, we have $P(\varIntCmf) > 0$ for sufficiently large $\varIntCmf$ so $\lim_{\varIntCmf\rightarrow+\infty} P(\varIntCmf) = 0^+$.
	\item Since $\sum_{\ell=1}^{t} \coeff_\ell < 1$ we have $P(0) < 0$.
\end{itemize}
Therefore, \Cref{corollary:function_analysis:polynome_one_sign_change} applies and there exists $\varIntCmf_0 > 0$ such that
\begin{equation}
	\label{eq:kernel_exact_recovery:proof_laplace_h_prime_positif}
	\funcCondThree^{(2)}_{2}(\maximizer)\ >\ \frac{p^2}{ \lambda_{t+1}^{2(1-p)} } \varIntCmf_0^2 \int_{0}^{+\infty} P(\varIntCmf) \diff\nu(\varIntCmf) = \frac{p^2}{ \lambda_{t+1}^{2(1-p)} }\varIntCmf_0^2 \funcCondThree\kparen{\maximizer}
	.
\end{equation}
This establishes that $\funcCondThree_2^{(2)}(\maximizer) > C_2 \funcCondThree(\maximizer)$ where $C_2 >0$ is a positive constant.
The same rationale leads to $\funcCondThree_1^{(2)}(\maximizer) > C_1 \funcCondThree(\maximizer)$ with $C_1\geq0$ ($C_1=0$ for $p=1$ since $\funcCondThree^{(2)}_1$ is identically zero).
Since $\funcCondThree^{(2)}= \funcCondThree_1^{(2)} + \funcCondThree_2^{(2)}$, one obtains that $\funcCondThree^{(2)}(\maximizer) > (C_1 + C_2) \funcCondThree(\maximizer)$, which concludes the proof for $\lambda_1 < \dots < \lambda_t$.

Let us now come back to the general case where $\lambda_1 \leq \dots \leq \lambda_t$. 
Denote by $\tilde{\lambda}_{1} < \ldots < \tilde{\lambda}_{t'}$, with $t'  \leq t$, the ordered distinct values in $\{\lambda_\ell\}_{\ell=1}^{t}$, and let $\tilde{\lambda}_{t'+1}=\lambda_{t+1}$.  
Moreover, for any $\ell'\in\intervint{1}{t'}$, let $\tilde{\coeff}_{\ell'}$ be equal to the sum of the coefficients $c_{\ell}$ such that $\lambda_{\ell}=\tilde{\lambda}_{\ell'}$, and let $\tilde{\coeff}_{t'+1}=-1$.
We note that, by definition, $\sum_{\ell'=1}^{t'} \tilde{\coeff}_{\ell'} = \sum_{\ell=1}^{t} \coeff_\ell < 1$.
We can then show that $\funcCondThree^{(2)}(\maximizer) > C \funcCondThree(\maximizer)$ with $C>0$ by applying the same reasoning as above to $\tilde{\lambda}_1,\dots,\tilde{\lambda}_{t'+1}$ and $\tilde{\coeff}_1,\dots, \tilde{\coeff}_{t'+1}$.

\vspace*{1em}
\noindent
\textit{{\bf Case 3}\textnormal{:} $\lambda_{1} < \lambda_{t+1} \leq \lambda_{t}$.}
There exists $\ell_0\in\intervint{1}{t-1}$ such that:
\begin{subnumcases}{}
 	\lambda_{\ell}<\lambda_{t+1}  & \mbox{for $\ell\leq \ell_0$}
 	\label{eq:separation of lambdas:a}
 	\\
 	\lambda_{\ell}\geq\lambda_{t+1} & \mbox{for $\ell>\ell_0$}.
 	\label{eq:separation of lambdas:b}
\end{subnumcases}
Denote $\varepsilon \triangleq 1-\sum_{\ell=1}^t \coeff_\ell>0$ and let $s_1\triangleq \sum_{\ell=1}^{\ell_0} \coeff_{\ell} + \tfrac{\varepsilon}{2}$ and $s_2\triangleq \sum_{\ell=\ell_0+1}^{t} \coeff_{\ell} + \tfrac{\varepsilon}{2}$ such that $s_1 + s_2 = 1$.
One can write
\begin{multline}
	\condOneNegKernelSymb(\maximizer)
	\;=\;
	s_1 \underbrace{\int_{0}^{+\infty} \kparen{ 
		\sum_{\ell=1}^{\ell_0} \frac{\coeff_{\ell}}{s_1}
		\cste^{-\varIntCmf \lambda_{\ell}^p}
		- \cste^{-\varIntCmf \lambda_{t+1}^p} 
	} \diff\nu(\varIntCmf) }_{\triangleq \; \phi_1(\maximizer)}
    \\
	+ 
	s_2 \underbrace{\int_{0}^{+\infty} \kparen{ 
		\sum_{\ell=\ell_0+1}^{t} \frac{\coeff_{\ell}}{s_2}
		\cste^{-\varIntCmf \lambda_{\ell}^p} 
		- \cste^{-\varIntCmf \lambda_{t+1}^p} 
	} \diff\nu(\varIntCmf) }_{\triangleq \; \phi_2(\maximizer)}
	.
\end{multline}
Using \eqref{eq:separation of lambdas:a} and the fact that $\sum_{\ell=1}^{\ell_0}\tfrac{\coeff_{\ell}}{s_1}<1$, we have that $\phi_1(\maximizer)<0$ by resorting to the same reasoning as in Case 2.
Similarly, using \eqref{eq:separation of lambdas:b} and the fact that $\sum_{\ell=\ell_0+1}^{t}\tfrac{\coeff_{\ell}}{s_2}<1$, we obtain that $\phi_2(\maximizer)<0$ by the same arguments as in Case 1.
Hence, we finally have $\condOneNegKernelSymb(\maximizer) \leq s_1\phi_1(\maximizer) + s_2 \phi_2(\maximizer) < 0$, which leads to the desired contradiction.

}
}

\conditionalPagebreak

\subsection{Proof of \texorpdfstring{\Cref{th:contrib:dimD:grid_recovery}}{Theorems~\ref{th:contrib:dimD:grid_recovery}} - Recovery in dimension \texorpdfstring{$D$}{D}}
	\label{subsubsec:proof_cmf_dimD}

{

	\newcommand{\maximizer}{\param_m}
	\newcommand{\paramn}{\widetilde{\param}}

	Let $\supportt = \{\paramt_\ell\}_{\ell=1}^{\nbAtome}$ and assume 	$\SetAug{}(\supportt)$ is axis admissible. 
	Let us first observe that, by virtue of \Cref{cor:existence CMF dictionary}, the elements of $\kset{\atome(\param)}{\param\in\SetAug(\supportt)}$ are linearly independent. 	
	Let 
	\begin{equation}
		\calR \triangleq \mathrm{span}(\kset{\atome(\param)}{\param\in\SetAug(\supportt)}).
	\end{equation}	
	If $\bsr\in\calR_{}\backslash\{{\bf0}_{\spaceObs}\}$, 
	we note from \Cref{Fact:existence-maximizer} that the function $f : \param \mapsto \kvbar{\kangle{\atome(\param),\bsr}}$ admits at least one maximizer since $\bsr$ results from a (non trivial) finite linear combination of atoms.
	Moreover, the maximum of $f$ must be strictly greater than zero. If not, $\kangle{\atome(\param),\bsr}=0$ $\forall \param\in \SetAug(\supportt)$ and one deduces that $\bsr\in\calR^\perp$, the orthogonal to $\calR$.
	Since $\bsr\in\calR_{}$, this leads to $\bsr={\bf0}_{\spaceObs}$ which is in contradiction with our initial assumption ``$\bsr\neq{\bf0}_{\spaceObs}$''.
	We also have that the maximizers of $f$ must belong to $\SetAug{}(\supportt)$ as shown by the following arguments.

If $\maximizer$ is a maximizer of $f$, then $t = \maximizer\element{d}$ is a maximizer of 
\begin{equation}
	f_d: t \mapsto \kvbar{\kangle{\atome(\maximizer + (t - \maximizer\element{d}) \bfe_d),\bsr}}. 
\end{equation}
Denoting $\param_0\triangleq \maximizer-\maximizer\element{d}\bfe_d$, we have $\param_0 \perp \bfe_d$ by construction and $f_d$ writes
\begin{equation}
	f_d(t) = \kvbar{
		\sum_{\ell=1}^{\nbAtome}
		\coeff_\ell
		\kernelPaper(\param_0 + t\,\bfe_d, \paramt_\ell)
	}
	.
\end{equation}	
Because $\SetAug{}(\supportt)$ is axis admissible with respect to $\kernelPaper$ that $f_d$ is not identically zero (since $f_d(\maximizer\element{d})=f(\maximizer)>0$), the maximizers of $f_d$ must belong to $\support_d\triangleq \{\paramt_\ell\element{d}\}_{\ell=1}^{\nbAtome}$.
As a consequence, since this conclusion holds for any $d\in\intervint{1}{\dimParam}$, we have that $\maximizer\in\prod_{d=1}^\dimParam\calS_d=\SetAug{}(\supportt)$. 
Formulated in a slightly different way, we thus just proved that:
\begin{equation}\label{eq:ERC_cartS*}
	\forall \bsr\in\calR_{}\backslash\{{\bf0}_\spaceObs\}, \forall \param\in \paramSet\backslash\SetAug{}(\supportt), \max_{\param'\in\SetAug{}(\supportt)} \kvbar{\kangle{\atome(\param'),\bsr}} > \kvbar{\kangle{\atome(\param),\bsr}}
	.
\end{equation}

\noindent
We are now ready to prove the statements of the theorem:

\paragraph{$\SetAug(\supportt)$-delayed recovery of any $\support\subseteq\supportt$} 
First, since $\support\subseteq\SetAug{}(\supportt)$, we have that $\Vobs\in \calR_{}$. 
Moreover, $\Vobs\neq {\bf0}_{\spaceObs}$ since it results from a nontrivial linear combination of linearly independent atoms. 
Hence \eqref{eq:ERC_cartS*} holds and \comp{} selects a parameter in $\SetAug{}(\supportt)$ at the first iteration.
Repeating the same argument at the next iterations, \comp{} selects parameters in $\SetAug{}(\supportt)$ until the residual $\bsr$ vanishes.
Now, because the atoms in $\kset{\atome(\param)}{\param\in\SetAug(\supportt)}$ are linearly independent, $\bsr = {\bf0}_{\spaceObs}$ if and only if the set of parameters selected by \comp{}, say $\widehat{\support}$, verifies $\support\subseteq\widehat{\support}$.
Since \comp{} never selects twice the same parameter and $\widehat{\support}\subseteq\SetAug{}(\supportt)$, we thus achieve $\SetAug{}(\supportt)$-delayed recovery of $\support$. 

\paragraph{$\supportt$-delayed recovery of any $\support\subseteq\supportt$} 
Since the elements of $\kset{\atome(\param)}{\param\in\SetAug(\supportt)}$ are linearly independent, \eqref{eq:contrib:dimD:l1erc} can equivalently be rewritten as (see \eg \cite[Prop.~3.15]{Foucart2013}):
\begin{equation*}
	\forall \bsr\in\calR_{\supportt} \backslash\{{\bf0}_{\spaceObs}\}, \forall \param\in \SetAug(\supportt)\backslash\supportt, \max_{\param'\in\supportt} \kvbar{\kangle{\atome(\param'),\bsr}} > \kvbar{\kangle{\atome(\param),\bsr}}
	,
\end{equation*}
where $\calR_{\supportt} \triangleq \mathrm{span}(\{\atome(\paramt_\ell)\}_{\ell=1}^{\nbAtome})$. 
Combining this result with \eqref{eq:ERC_cartS*} and using the fact that $\calR_{\supportt}\subset\calR$ lead to
\begin{equation*} 
	\forall \bsr\in\calR_{\supportt} \backslash\{{\bf0}_{\spaceObs}\}, \forall \param\in \paramSet\backslash\supportt, \max_{\param'\in\supportt} \kvbar{\kangle{\atome(\param'),\bsr}} > \kvbar{\kangle{\atome(\param),\bsr}}.
\end{equation*}	
Following the same arguments as above, we then have that: 1) \comp{} selects parameters in $\supportt$ until the residual $\bsr$ vanishes; 2) $\bsr = {\bf0}_{\spaceObs}$ if and only if the set of parameters selected by \comp{} verifies $\support\subseteq\widehat{\support}$. Since \comp{} never selects twice the same parameter and $\widehat{\support}\subseteq\supportt$, \comp{} thus achieves $\supportt$-delayed recovery of $\support$.

\paragraph{Sharpness of the result} 
If~\eqref{eq:contrib:dimD:l1erc} is not verified, we have from \cite[Prop.~3.15]{Foucart2013} that there exists some $\Vobs_{\texttt{bad}}\in\calR_{\supportt} \backslash\{{\bf0}_{\spaceObs}\}$ such that
\begin{equation}
	\max_{\param\in\SetAug(\supportt)\setminus\supportt} \kvbar{
			\scalprod{\atome(\param)}{\Vobs_{\texttt{bad}}}
	}
	>
	\max_{\param\in\supportt} \kvbar{
		\scalprod{\atome(\param)}{\Vobs_{\texttt{bad}}}
	}
	.
\end{equation}
In other words, \comp{} with $\Vobs{}_{\texttt{bad}}$ as input selects some $\param\notin\supportt$ at the first iteration.

\subsection{Proof of \texorpdfstring{\Cref{th:contrib:laplace_recovery}}{Lemma~\ref{th:contrib:laplace_recovery}}}
	\label{subsec:proof:laplace-grid-are-admissible}

{

\newcommand{\paramOnLine}{t}

Let $\calG= \prod_{d=1}^{\dimParam} \calS_d = \{\param_\ell\}_{\ell=1}^{\card{}(\calG)}$ be an arbitrary Cartesian grid in $\kR^{\dimParam}$ and $\{\coeff_\ell\}_{\ell=1}^{\card{}(\calG)}\subset\kR$ be a set of $\card{}(\calG)$ coefficients not all equal to $0$.
Consider $d\in\intervint{1}{\dimParam}$ and $\param_0\in\kR^\dimParam$ such that $\theta_{0}\element{d}=0$ and define
\begin{equation}
	\kfuncdef{f_d}{\kR}{\kR+}[\paramOnLine][
		\displaystyle
		\kvbar{\sum_{\ell=1}^{\card{(\calG)}} \coeff_\ell \kernelPaper(\param_0 + \paramOnLine\bfe_d, \param_\ell)
		}
		=
		\kvbar{\sum_{\ell=1}^{\card({\calG})} \coeff_\ell \cste^{-\lambda \kvvbar{\param_0 + \paramOnLine\bfe_d - \param_\ell}_p^p}}
			].
\end{equation}
We assume that $f_d$ is not identically zero as in the statement of \Cref{def:contrib:grille_admissibility}.
$f_d$ can be rewritten as
\begin{equation*}
	f_{d}(\paramOnLine) 
	=
	\kvbar{
	\sum_{\ell=1}^{\card(\calG)} \coeff_\ell 
		\cste^{
			- \lambda \kvbar{\paramOnLine - \param_\ell\element{d}}^p
			- \lambda \sum_{j=1,j\neq d}^\dimParam\kvbar{\param_0\element{j} - \param_\ell\element{j}}^p
			}
	}
	.
\end{equation*}
Let $q$ denote the number of distinct elements of $\{\param_\ell\element{d}\}_{\ell=1}^{\card(\calG)}$ and suppose (up to some renumbering) that $\{\param_\ell\element{d}\}_{\ell=1}^q$ are pairwise distinct.
We note that $q\geq 1$ because otherwise there is a contradiction with our hypothesis ``$f_d$ not identically zero". 
We can then rewrite $f_d$ as
\begin{equation}
	f_{d}(\paramOnLine)
	=
	\kvbar{
	\sum_{\ell=1}^{q} \widetilde{\coeff}_\ell 
		\cste^{
			- \lambda \kvbar{\paramOnLine - \param_\ell\element{d}}^p
			}
	}
\end{equation}
where the terms proportional to $\cste^{|t-\param_\ell\element{d}|}$ for (possibly) identical values of $\param_\ell[d]$ have been merged together, and the scalars $\widetilde{c}_\ell$ take into account the constant terms in the exponentials that do not depend on $t$.

\noindent
Let $\dico_{1} = \{\atome[1][](\paramOnLine'):\ \paramOnLine' \in \mathbb{R}\}$ be a Generalized Laplace dictionary in dimension 1 (see \Cref{def:contrib:laplace_dico}).
Then, $f_d(\paramOnLine)$ can also be interpreted as the inner product between atom $\atome[1][](\paramOnLine)\in\dico_{1}$ and $\Vobs_1\triangleq \sum_{\ell=1}^q \widetilde{\coeff}_\ell \atome[1][](\param_\ell\element{d})$. 
Let $\widetilde{\calS}\triangleq\{\param_\ell[d]:\widetilde{\coeff_\ell}\neq0, \ell=1\dots q\}$. 
Applying \Cref{th:contrib:cmf_uniformRecov_1D}, we have that \comp{} with $\Vobs_1$ as input achieves exact $\card{}(\widetilde{\calS})$-step recovery of $\widetilde{\calS}$.
In particular, this implies that \comp{} selects a parameter in $\widetilde{\calS}$ at the first iteration, that is:
\begin{equation}
	\forall\paramOnLine'\in\kR\setminus\widetilde{\calS},
	\quad
	f_d(\paramOnLine') 
	< \max_{t\in \widetilde{\calS}} f_d(t).
\end{equation}
Hence, the maximizers of $f_d$ belong to $\widetilde{\calS}\subseteq\{\param_\ell\element{d}\}_{\ell=1}^q \subseteq \calS_d$.

}

\conditionalPagebreak

\section{Miscellaneous}
	\label{sec:app:misc}


\subsection{Proof of \texorpdfstring{\Cref{Fact:existence-maximizer}}{Lemma~\ref{Fact:existence-maximizer}}}
\label{subsec:app:existence-maximizer}

Let $\bsr\in\spaceObs$, $k\in\kN$ and assume that there exists $\nbAtome$ parameters $\{\paramt_\ell\}_{\ell=1}^{\nbAtome}\subset\paramSet$ and $\nbAtome$ nonzero coefficients $\{\coeff_\ell\}_{\ell=1}^{\nbAtome}\subset\kR*$ such that $\bsr=\sum_{\ell=1}^{\nbAtome} \coeff_\ell \atome{}(\paramt_\ell)$.
Define function $\varphi$ as
\begin{equation}
	\label{eq:proof-fact-maximizer:inner-prod}
	\kfuncdef{\varphi}{\paramSet}{\kR+}[
		\theta
	][
		\kvbar{\kangle{\atome(\param), \bsr}}
		=
		\kvbar{\sum_{\ell=1}^\nbAtome \coeff_\ell \kappa(\param, \paramt_\ell)}
	],
\end{equation}
that is, the function involved in step~\ref{line:algo:continuousOMP:findtheta} in \Cref{alg:continuousOMP}.
We now prove that a maximizer of~\eqref{eq:proof-fact-maximizer:inner-prod} exists.
To that aim, we distinguish two cases.

\paragraph{Case 1: $\forall\param\in\paramSet,\,\varphi(\param)=0$}
In that case, any parameter $\param\in\paramSet$ is a maximizer of $\varphi$.

\paragraph{Case 2: $\exists\param_0\in\paramSet,\,\varphi(\param_0)>0$}	
Denote $\varepsilon\triangleq \varphi(\param_0)$.
We then have
\begin{equation}
	\sup_{\param\in\paramSet} \varphi(\param) \geq \varphi(\param_0) =\varepsilon>0.
\end{equation}
Hence, by condition~\eqref{eq:contribution:evanecent}, there exists $\nbAtome{}$ compact sets $\kfamily{K_\ell}{\ell=1}^{\nbAtome}$ such that for all $\ell\in\intervint{1}{\nbAtome}$, $\param_0\in K_\ell$ and
\begin{equation}
	\label{eq:proof-fact-maximizer:bounds1}
	\forall\param\in K_\ell^c, \quad
	\kappa(\param, \paramt_\ell) < \frac{1}{\sum_{\ell'=1}^{\nbAtome} \kvbar{\coeff_{\ell'}}} \varepsilon
	.
\end{equation}
Note that the right-hand-side of~\eqref{eq:proof-fact-maximizer:bounds1} is well defined: by positive-definiteness of $\kangle{\cdot,\cdot}$, $\bsr\neq0$ so $\nbAtome{}>0$ and $\sum_{\ell=1}^\nbAtome|\coeff_\ell|>0$ necessarily.
Define $K=\cup_{\ell=1}^{\nbAtome}K_\ell$.
Since $K^c=\cap_{\ell=1}^{\nbAtome} K_\ell^c$, we have using the triangular inequality
\begin{equation}
	\label{eq:proof-fact-maximizer:bounds2}
	\forall \param\in K^c, \quad
	\varphi(\param) \leq \sum_{\ell=1}^{\nbAtome} \kvbar{c_\ell} \kappa(\param, \paramt_\ell) < \varepsilon
	.
\end{equation}
See now that $\varphi$ is continuous by continuity of $\kernelPaper$, $K$ is compact as a union of compact sets.
Then, the extreme value theorem ensures that there exists $\param_m$ such that $\varphi(\param_m)\geq \varphi(\param)$ for all $\param\in K$.
\Cref{Fact:existence-maximizer} follows by seeing that $\varphi(\param_m)\geq \varphi(\param)$ for all $\param$ by~\eqref{eq:proof-fact-maximizer:bounds2}.


\subsection{Proof of \texorpdfstring{\Cref{corollary:function_analysis:polynome_one_sign_change}}{Lemma~\ref{corollary:function_analysis:polynome_one_sign_change}}}
	\label{subsec:app:analysis}

{
\renewenvironment{proof}{}{\qed}

\noindent
The proof of \Cref{corollary:function_analysis:polynome_one_sign_change} is based on the following result:
\begin{lemma}[Laguerre's generalization of Descartes's rule of signs \protect{\cite[p.~319]{fejzic2009}}]
	\label{lemma:geometric:polynome_zero}
	Let $a_1,\ldots,a_\nbAtome$ be nonzero real coefficients and $0 < x_1 < \dots < x_\nbAtome$ be real numbers.
	Let $z$ be the number of real roots of the function $P(u) = \sum_{\ell=1}^{\nbAtome} a_\ell^{\phantom{u}} x_\ell^u$, and $n_c$ be the number of changes in sign in the sequence of numbers $a_1,\ldots,a_\nbAtome$.
	Then $z \leq n_c$.
\end{lemma}

\begin{figure}
	\includegraphics[width=\columnwidth]{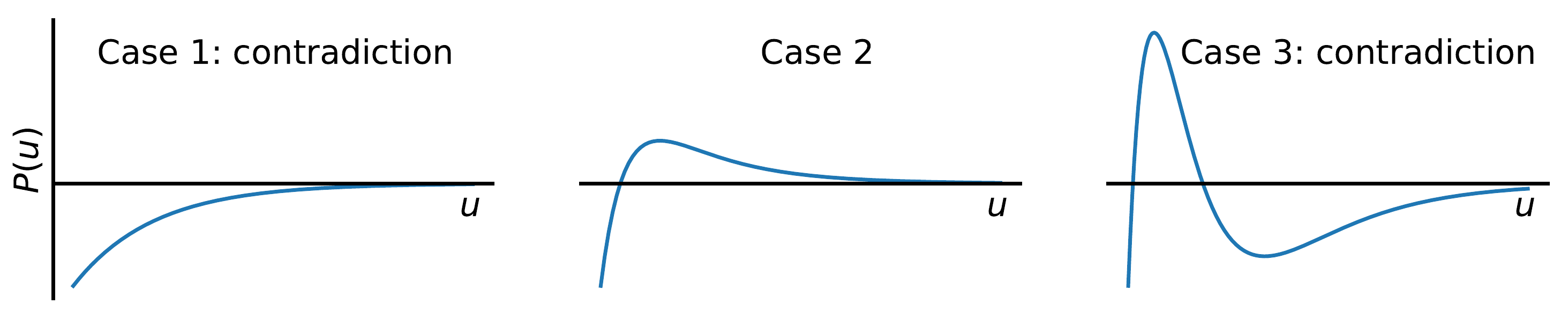}
	\caption{
		\label{fig:sum_exp_sign_cases}
		Shape of $P$ (see proof of \Cref{corollary:function_analysis:polynome_one_sign_change}) with constraints \textit{\RNum{1})} $P$ is continuous, \textit{\RNum{2})} $P(u)<0$ and \textit{\RNum{3})} $\exists u_0>0$ such that $P(u)>0$ for all $u>u_0$.
		One see that the constraints cannot be satisfied in cases 1 and 3.
	}
\end{figure}

The sequence of coefficients $a_{\ell} = \coeff_{\nbAtome+1-\ell}$ with $\ell\in\intervint{1}{\nbAtome}$ has only two sign changes by hypothesis.
By applying \Cref{lemma:geometric:polynome_zero} with $x_{\ell} = \cste^{-\lambda_{\nbAtome+1-\ell}}$, one sees that $P$ has at most two real roots, so at most two sign changes on $\kR_+$.
However, $P$ must satisfy the following constraints: 
\begin{enumerate}[label=\itshape\roman*)]
	\item $P$ is continuous on $\kintervco{0}{+\infty}$,
	\item $P(0)<0$,
	\item there exists $u_0>0$ such that $P(u)>0$ for all $u>u_0$. 
\end{enumerate}
As illustrated in \Cref{fig:sum_exp_sign_cases}, these three constraints cannot be verified simultaneously if $P$ has $0$ or $2$ roots. 
Thus $P$ has exactly one sign change on $\kR_+$ and there exists $\varIntCmf_0>0$ such that $\varIntCmf < \varIntCmf_0 \implies P(\varIntCmf) < 0$ and $\varIntCmf > \varIntCmf_0 \Longrightarrow P(\varIntCmf) > 0$.
One then has, for any non-decreasing function $f$ and any (non-negative) measure $\nu$ on $\kR+$:
\begin{align*}
	\int f(\varIntCmf) P(\varIntCmf) &\diff\nu(\varIntCmf) 
	\nonumber \\
	\;=\; &
		\int_{0}^{\varIntCmf_0} \underbrace{ f(\varIntCmf) }_{ \text{non-decreasing} } \underbrace{P(\varIntCmf)}_{ \leq 0} \diff\nu(\varIntCmf) 
		+
		\int_{\varIntCmf_0}^{+\infty} \underbrace{ f(\varIntCmf) }_{\text{non-decreasing}} \underbrace{P(\varIntCmf)}_{\geq 0} \diff\nu(\varIntCmf) \\
	\;\geq\;&
		\int_{0}^{\varIntCmf_0} f(\varIntCmf_0) P(\varIntCmf) \diff\nu(\varIntCmf) 
		+
		\int_{\varIntCmf_0}^{+\infty} f(\varIntCmf_0) P(\varIntCmf) \diff\nu(\varIntCmf) \\
	\;=\;&
		f(\varIntCmf_0) \int_{0}^{+\infty}  P(\varIntCmf) \diff\nu(\varIntCmf). \qquad \qquad \hfill \qedhere
\end{align*}
}

\conditionalPagebreak


\section{Details related to Example~\ref{ex:intro:erc_required_dim2}}
	\label{sec:other:calcul_separation}

{
\renewcommand{\varMuette}{x}
\newcommand{\funcRatio}{f}

Assume that $\cmf$ is right differentiable at $0$.
We first prove by contradiction that $\cmf[][(1)](0)<0$.
Assume that $\cmf[][(1)](0)=0$.
As $\cmf$ is a CMF, $\cmf[][(1)](t) \leq 0$ for all $t \in \kR+$ and $\cmf[][(1)]$ is non-decreasing on $\kR+$. 
It follows that $\cmf[][(1)]$ is identically zero.
Hence $\cmf$ is constant and equal to $\cmf[][](0)=1$ which contradicts the assumption $\lim_{t\to+\infty} \cmf(t) = 0$.

\noindent
Consider now the function defined for all $\varMuette\geq0$ by $\funcRatio: \varMuette\mapsto(\nbAtome-1)\cmf(2\varMuette) - \nbAtome\cmf(\varMuette) + 1$.
We note that $\funcRatio$ corresponds to the quantity involved in~\eqref{eq:contrib:ratio_scalprod2} with the substitution $\varMuette\leftrightarrow\Delta^p$.
Since $\cmf(0)=1$, we have $\funcRatio(0)=0$.
Moreover $f$ is differentiable for any $\varMuette>0$ and
\begin{align}
	\funcRatio^{(1)}(\varMuette) 
	\;=\;&
		2(\nbAtome-1) \cmf{}^{(1)}(2\varMuette)
		-
		\nbAtome \cmf{}^{(1)}(\varMuette)
	\nonumber \\
	\;=\;&
	\nbAtome\cmf{}^{(1)}(\varMuette) \kbracket{
		2\kparen{1 -\tfrac{1}{\nbAtome}} \tfrac{\cmf{}^{(1)}(2\varMuette)}{\cmf{}^{(1)}(\varMuette)}
		-
		1
	}
	\label{eq:other:calcul_separation:diffratio}
	.
\end{align}
Since $\cmf$ is right differentiable at $0$, the ratio $\tfrac{\cmf{}^{(1)}(2\varMuette)}{\cmf{}^{(1)}(\varMuette)}$ tends to $1$ as $\varMuette$ tends to $0$ (remember that we proved that $\cmf[][(1)](0) < 0$).
Since $k\geq3$ we have $2\kparen{1 -\tfrac{1}{\nbAtome}}>1$, hence there exists $\varMuette_0>0$ such that 
\begin{equation}
	\varMuette < \varMuette_0 \Rightarrow 2\kparen{1 -\tfrac{1}{\nbAtome}} \tfrac{\cmf{}^{(1)}(2\varMuette)}{\cmf{}^{(1)}(\varMuette)}-1 >0.
\end{equation}
By \Cref{lemma:strictCMF}, we have that $\cmf[][(1)]{}(\varMuette)<0$ for all $\varMuette>0$. Hence $\funcRatio^{(1)}(\varMuette)<0$ for $\varMuette < \varMuette_0$, that is $f$ is decreasing on $[0, \varMuette_0]$.
Combining this result with $\funcRatio(0)=0$, we deduce that~\eqref{eq:contrib:ratio_scalprod2} holds whenever $\Delta^p<\varMuette_0$, that is the wrong parameter ${\bf0}_{\dimParam}$ will be preferred to any of the $\{\paramt_\ell\}_{\ell=1}^{\nbAtome}$.

}
\conditionalPagebreak


\section{Exact recovery in higher dimensions - CMF kernel and \texorpdfstring{$k=2$}{k=2}}
	\label{sec:app:dimD:uniformrecovery_k_2}
\newcommand{\varu}{\param_1}
\newcommand{\varv}{\param_2}
\newcommand{\varw}{\param_3}

In this section, we elaborate on the notion of ``axis admissibility'' (see \Cref{def:contrib:grille_admissibility}) for general CMF kernels.
We first show (see \Cref{example:cart_grid_k2}) that there exist some Cartesian grids which are not axis admissible with respect to some CMF kernels.
We then emphasize in the case $\nbAtome=2$ that the notion of ``axis admissibility'' is not necessary to achieve $\nbAtome$-step recovery in CMF dictionaries.

\begin{example}
	\label{example:cart_grid_k2}
	Let $\dico$ be a CMF dictionary with induced kernel $\kernelPaper=\cmf{}(\kvvbar{\cdot-\cdot}_p^p)$ and consider $\Delta>0$, $c_1 = c_2 = 1$, $c_3=c_4= - \tfrac{1 + \cmf(\Delta^p)}{\cmf(\Delta^p) + \cmf(2\Delta^p)}$ and
    \begin{equation}
		\kforall[t\in\kR]\quad
		f_1(t) = \left |\sum_{\ell=1}^4 \coeff_\ell \kappa(t \bfe_1, \param_\ell) \right |
	\end{equation}
	where $\bfe_{d}$ the $d$-th canonical basis vector of $\kR^2$.

	Let $\calG\triangleq\{\param_\ell\}_{\ell=1}^4\subset\kR^2$ be a Cartesian grid with $\param_1=(0,0)={\bf0}_2$, $\param_2= (\Delta,0) = \Delta \bfe_1$, $\param_3= (0,\Delta) = \Delta \bfe_2$, $\param_4= (\Delta,\Delta) = \Delta{\bf1}_2$. 
	Simple algebraic manipulations then show that $f_1(0)=f_1(\Delta) = 0$ and
	\begin{equation}\label{eq:expre_f1_example}
		f_1\kparen{\tfrac{\Delta}{2}} = 2 \cmf(\tfrac{1}{2^p}\Delta^p)
		\kvbar{
			1 - \frac{1 + \cmf(\Delta^p)}{\cmf(\Delta^p) + \cmf(2\Delta^p)}
			\frac{\cmf(\Delta^p + \tfrac{1}{2^p}\Delta^p)}{\cmf(\tfrac{1}{2^p}\Delta^p)}
		}
		.
	\end{equation}
	If $\cmf$ and $\Delta$ are such that $f_1\kparen{\tfrac{\Delta}{2}} \neq 0$, one can conclude that the maximizers of $f_1$ are distinct from 0 and $\Delta$.
	In view of \Cref{def:contrib:grille_admissibility}, this shows that $\calG$ is not axis admissible with respect to $\dico$. 

	For instance, this is the case for the CMF $\cmf:x\mapsto\tfrac{1}{1+x}$ (cf  \Cref{ex:CMF}).
	Indeed, in this case \eqref{eq:expre_f1_example} particularizes to 
	\begin{align}
		f_1(\tfrac{\Delta}{2}) \;=\;&  \frac{2}{1 + \tfrac{\Delta^p}{2^p}} \kvbar{
			1 - \frac{
			2+\Delta^{p}}{1 + \tfrac{1+ \Delta^p}{1 + 2\Delta^p}} \frac{1 + \tfrac{\Delta^p}{2^p}}{1 + \Delta^p + \tfrac{\Delta^p}{2^p}}
		}
		.
	\end{align}
	As the factor inside the absolute value in the right-hand side is a non-zero rational function of $x = \Delta^{p}$, we have $f_{1}(\Delta/2) \neq 0$ except possibly on a set of values of $\Delta$ which has Lebesgue measure equal to zero.
	Hence there exists $\Delta > 0$ such that $f_{1}(\Delta/2) > 0$.

	We finally note that the construction presented here for the case $\dimParam=2$ easily extends to $\dimParam > 2$ by zero-padding of the $\param_{\ell}$'s.
\end{example}

We next show that $\nbAtome$-step recovery of $\supportt$ with $\nbAtome=\card(\supportt)=2$ may be possible in CMF dictionaries even when the axis admissibility assumption fails to hold. First, we state and prove a useful technical lemma:

\begin{lemma}
	\label{lemma:inequality_cmf_kernel}
	Let $\kernelPaper$ be a CMF kernel in dimension $\dimParam$ in the sense of \Cref{def:intro:cmf_kernel}. 
	For any $\varu,\varv,\varw \in \kR^\dimParam$, the following result holds:
	\begin{equation}
			\label{eq:inequality_cmf_kernel:2}
			\kernelPaper(\varu,\varv) \, \kernelPaper(\varv,\varw)
			\leq \kernelPaper(\varu,\varw).
	\end{equation}
\end{lemma}
\begin{proof}
	By definition, there exists a CMF $\cmf$ such that $\kernelPaper(\cdot,\cdot)= \cmf(\kvvbar{\cdot - \cdot}_p^p)$ and $\cmf(0)=1$.
	Since $\cmf$ is nonnegative and decreasing, we have for all $x,y\geq0$
	\begin{equation}
		\label{eq:inequality_cmf:2}
		\cmf(x+y) \geq \cmf(x+y)\cmf(x+y) \geq \cmf(x)\cmf(y)
		.
	\end{equation}
	Using this result with $x=\|\varu - \varv\|_p^p$ and $y=\|\varv - \varw\|_p^p$, we have
	\begin{align}
		\kernelPaper(\varu, \varv)  \kernelPaper(\varv, \varw)
		\;\leq\;& 
		\cmf(
				\kvvbar{\varu - \varv}_p^p + \kvvbar{\varv - \varw}_p^p
			)
		.
	\end{align}
	Since the quasi-norm $\|\cdot\|_p^p$ satisfies a triangular inequality, we have $\|\varu - \varw\|_p^p \leq \|\varu - \varv\|_p^p + \|\varv - \varw\|_p^p $.
	As any CMF is decreasing, \eqref{eq:inequality_cmf_kernel:2} follows.

\end{proof}

\noindent
We are now ready to state our recovery result:
\begin{lemma}[Exact recovery for CMF dictionaries when $\nbAtome=2$]
	\label{lemma:exact-recovery-k2}
	Let $\dico$ be a CMF dictionary in dimension $\dimParam \geq 1$ with induced kernel $\kernelPaper$. 
	Consider a support $\supportt = \kbrace{\paramt_1, \paramt_2}$ where $\paramt_1 \neq \paramt_2$, and let $\Gtrue\in\kR^{2\times2}$ be the matrix defined by $\Gtrue\element{\ell,\ell'} = \kernelPaper(\paramt_\ell,\paramt_{\ell'})$.
	Assume that
	\begin{equation}
		\label{eq:app:erc_k2}
		\kforall[\param\in\SetAug(\supportt)\setminus\supportt] \quad \kvvbar{\kinv{\Gtrue}\gtheta}_1 < 1
	\end{equation}
	where $\gtheta\in\kR^2$ is defined by $\gtheta\element{\ell}=\kernelPaper(\param,\paramt_\ell)$ for $\ell=1,2$. 
	Then \comp{} achieves exact $2$-step recovery of $\supportt$.
\end{lemma}
\begin{proof}
	\renewcommand{\maximizer}{\param_m}
	By \Cref{lemma:CMFkerAdmissible}, $\kernelPaper$ is admissible in the sense of \Cref{hyp:geometric:0}. 
	We show below that since~\eqref{eq:app:erc_k2} holds, $\supportt$ is admissible with respect to $\kernelPaper$ in the sense of \Cref{def:admissible_support}.
	\Cref{lemma:exact-recovery-k2} then follows from \Cref{th:MainAbstractTheorem}.

	Consider a non-empty subset of indices $T\subseteq\{1,2\}$ and $t\triangleq\card(T)$.
	Let also $\{c_\ell\}_{\ell\in T}$ be such that $c_\ell>0$ and $\sum_{\ell\in T}c_\ell<1$.
	Define 
	\begin{equation}
		\label{eq:app:def_psi_k2}
		\kfuncdef{\psi}{\kR}{\kR_+}[
			\param
		][
			\sum_{\ell\in T} c_\ell \kernelPaper(\param,\paramt_\ell)
		]
	\end{equation}
	We next show that items~\ref{hyp:geometric:1} and~\ref{hyp:geometric:2} of \Cref{def:admissible_support} are satisfied.

	\paragraph{Item~\ref{hyp:geometric:1} of \Cref{def:admissible_support}}
	We distinguish two cases:
	\begin{itemize}
		\item If $t=1$, we can assume without loss of generality that $T=\{1\}$.
		Since $\kernelPaper(\param,\paramt_1)<1$ for all $\param\neq\paramt_1$, one immediately sees that $\psi(\param) =c_1 \kernelPaper(\param,\paramt_1) < c_1 = \psi(\paramt_1)$ for all $\param\neq\paramt_1$. Hence, $\paramt_1$ is the unique global maximizer of $\psi$.
	
		\item If $t=2$,	let $\maximizer$ be a maximizer of $\psi$. 
        We note that $\psi$ can also be written as $\psi(\param) = |\scalprod{\atome(\param)}{\Vobs}|$ where $\Vobs\triangleq \sum_{\ell=1}^\nbAtome \coeff_\ell \atome(\paramt_\ell)$ hence a maximizer always exists by virtue of \Cref{Fact:existence-maximizer}.
		Since $\maximizer$ maximizes the $D$-dimensional function $\psi$, its $d$-th entry $\theta_{m}[d]$ is a maximizer of the one-dimensional section of $\psi$ along the $d$-th canonical direction, denoted $\psi_{d}$: 
		\begin{equation}
			\label{eq:appE:maximizer-dimd}
			\kfuncdef{\psi_d}{\kR}{\kR+}[x][
				\sum_{\ell\in T} c_\ell \cmf(\kvbar{x - \paramt_\ell\element{d}}^p + \sum_{j\neq d}\kvbar{\maximizer\element{j} - \paramt_\ell\element{j}}^p)
			.
			]
		\end{equation} 
		Applying the same reasoning as in the proof of \Cref{th:contrib:cmf_uniformRecov_1D} (see part of the proof dedicated to establishing ``item~\ref{hyp:geometric:1} of \Cref{def:admissible_support}''), we have $\forall x\notin\{\paramt_\ell\element{d}\}_{\ell=1}^{2}$: $\psi_d$ is twice differentiable and $\psi_d^{(2)}(x)>0$. Hence, no $x\in\{\paramt_\ell\element{d}\}_{\ell=1}^{2}$ can be a maximizer and  necessarily $\maximizer\element{d}\in\{\paramt_\ell\element{d}\}_{\ell=1}^{2}$. Since this result is valid for all $d\in\intervint{1}{\dimParam}$, we finally have $\maximizer\in\SetAug{}(\supportt)$.

	\noindent
	Therefore, since~\eqref{eq:app:erc_k2} holds, we have
	\begin{equation}
		\begin{split}
			\max_{\param\in\SetAug(\supportt)\setminus\supportt} \psi(\param) &= \max_{\param\in\SetAug(\supportt)\setminus\supportt} \kvbar{\scalprod{\atome(\param)}{\Vobs}} \\
			&< \max_{\paramt\in\supportt} \kvbar{\scalprod{\atome(\paramt_\ell)}{\Vobs}} = \max_{\paramt\in\supportt} \psi(\paramt_\ell) 
		\end{split}
	\end{equation}
	Hence all maximizers of $\psi$ belong to $\supportt$.
\end{itemize}

\paragraph{Item~\ref{hyp:geometric:2} of \Cref{def:admissible_support}}
From the working assumptions of item~\ref{hyp:geometric:2}, the set $T$ satisfies $T\neq\emptyset$ and there exists $\ell \in \{1,2\} \backslash T$.
Hence, we have $T\neq\{1,2\}$, that is $T$ is a singleton.
We assume without loss of generality that $T=\{1\}$.
Hence $\psi(\param) = c_{1} \kernelPaper(\param,\paramt_1)$ for some $0<c_1<1$. 
If $\psi_1(\paramt_1) - \kernelPaper(\paramt_1,\paramt_2) \leq 0$, then $c_1 - \kernelPaper(\paramt_1,\paramt_2) = \psi_1(\paramt_1) - \kernelPaper(\paramt_1,\paramt_2) \leq 0$.
Hence $c_1 \leq \kernelPaper(\paramt_1,\paramt_2)$ and $\psi(\param) \leq \kernelPaper(\paramt_1,\paramt_2)\kernelPaper(\param,\paramt_1)$ for each $\param$.
Using  \Cref{lemma:inequality_cmf_kernel} with $\varu=\param$, $\varv=\paramt_1$, $\varw=\paramt_2$, we obtain for each $\param\in\paramSet$:
\begin{equation*}
	\psi(\param) - \kernelPaper(\param,\paramt_2) \leq \kernelPaper(\param,\paramt_1)\kernelPaper(\paramt_1,\paramt_2)  - \kernelPaper(\param,\paramt_2) \leq 0
	.
\end{equation*}

\end{proof}

\conditionalPagebreak


\section{Table of notations}
	\label{sec:table_of_notations}

\begin{table}[ht]
	\newcommand{\ttile}[1]{\scshape#1}
	\centering
	\begin{tabular}{cl}
		\toprule
		\ttile{Notation} & \ttile{Comment} \\
		\midrule
		\multicolumn{2}{c}{\textit{General notations}} \\
		\midrule
		$\spaceObs$, $\Vobs$	 		& (Hilbert) observation space and observation \\
		$\dico{}, \atome(\cdot)$ 	& Dictionary $\dico$ made of parametric atoms $\atome$ \\
		$\mu$				& Coherence between atoms of a support \\
		$\coeffv\in\kR^\nbAtome$ 			& Weighting coefficients \\
		$\paramSet, \param$ & Parameter set and element \\
		$\support$, $\supportt$ & Set of parameters \\
		$\calG$ & Cartesian grid \\
		$\nbAtome, \ell$	& Number of atoms, most frequent index \\
		$\SetAug$			& Set augmenter, see~\eqref{eq:erc_ell1:defAtomeVirtuel} \\
		$\cmf$				& CMF (see~\Cref{def:intro:def_cmf}) \\
		$\kernelPaper$		& Kernel function $\paramSet\times\paramSet\to\kR+$ \\
		$\cmfkernelclassP[\dimParam]$	& Set of CMF kernels in dimension $\dimParam$ \\
		$\laplacekernelclassP[\dimParam]$ & Set of Laplace kernels in dimension $\dimParam$ \\
		$f^{(n)}$ & $n$-th derivative of function $f$ \\
		$\bfe_\ell$ & $\ell$-th element of the canonical basis \\
		$\bbE$ & Expectation operator \\
		$\csti$ & Imaginary number \\
		& \\
		\multicolumn{2}{c}{\textit{Technical notations}} \\
		\midrule
		$\Gtrue$, $\gtrue_\ell$ & Gram matrix related to a support $\support$, columns of $\Gtrue$ \\
		$\gtheta$ 				& parametric vector related to a support $\support$ \\
		\multirow{2}{*}{$\bfu,\bfv$}  		& Vector of $\kR^\nbAtome$ for some $\nbAtome$ often defined as  \\
		&	\qquad \qquad $\bfu,\bfv=\kinv{\Gtrue}\gtheta$ for some $\param\in\paramSet$ \\
		\bottomrule
	\end{tabular}
	\caption{
		\label{table:mr}
		Table of notations.
	}
\end{table}

\clearpage

\section*{Acknowledgments}

Part of this work has been funded thanks to the Becose ANR project no. ANR-15-CE23-0021.

\addcontentsline{toc}{section}{References}


\end{document}